\documentclass[12pt, margin=1.25in]{article}
 
\usepackage{amsfonts,amssymb,amsmath,amsthm,epsfig,euscript,bbm,amsthm}
\usepackage{algorithm}
\usepackage{algpseudocode}
\usepackage{amsthm}
 \usepackage{array}
    \usepackage{graphicx} 
    \usepackage{multirow}
    \usepackage{hhline,booktabs,amsmath}
    \usepackage{makecell}
\usepackage{tikz}
\usetikzlibrary{positioning, fit, shapes.geometric, arrows}
\usepackage {graphicx}
\usepackage{enumerate}
\pdfoutput=1
\usepackage{setspace}
\usepackage{subcaption}
\usepackage{caption}
\usepackage[authoryear]{natbib}
\usepackage[titletoc,title]{appendix}
\usepackage{hyperref}
\usepackage{mathtools}
 \usepackage{xcolor}

\usepackage{graphicx}
\usepackage{subcaption}
 
 \usepackage{xr}
 \definecolor{orange}{RGB}{230,170,120}
  \definecolor{green}{RGB}{120,200,120}

 \hypersetup{
 	colorlinks=true,
 	linkcolor=blue,
 	filecolor=blue,      
 	urlcolor=cyan,
 }
 \hypersetup{colorlinks,linkcolor={blue},citecolor={blue},urlcolor={red}} 
 \usepackage{geometry}                
 \usepackage{multirow}
\usepackage[section]{placeins}
 \usepackage{bbm}
 \usepackage{graphicx}
 \usepackage{amssymb}
 \usepackage{algorithm}
\usepackage{algpseudocode}
\def\addlegendimage{\csname pgfplots@addlegendimage\endcsname}

\pgfkeys{/pgfplots/number in legend/.style={%
        /pgfplots/legend image code/.code={%
            \node at (0.295,-0.0225){#1};
        },%
    },
}
\usepackage{algorithm,algpseudocode}

\newcounter{algsubstate}
\renewcommand{\thealgsubstate}{\alph{algsubstate}}

 \usepackage{epstopdf}
 \usepackage{tikz} 
 \usepackage{amsmath} 
\usepackage[utf8]{inputenc}
\usepackage{algorithm}
\usepackage{pgfplots}
 \theoremstyle{plain}
 \newtheorem{thm}{Theorem}[section]
 \newtheorem{lem}[thm]{Lemma}
 \newtheorem{prop}[thm]{Proposition}
 \newtheorem{cor}[thm]{Corollary}
 \usepackage{tabstackengine}
\usepackage{fixltx2e}

\def\mystrutwidth{2.5\baselineskip}
\def\mystrutdepth{1\baselineskip}
\def\mystrut{\rule[-\mystrutdepth]{0ex}{\mystrutwidth}}
\newcommand\bngo[1]{\kern-\fboxrule\fbox{\mystrut\makebox[\mystrutwidth]{#1}}}%
\setstackgap{L}{\mystrutwidth+2\fboxsep+\fboxrule}
\usepackage{setspace}
 \usepackage{epstopdf}
\usepackage{titletoc}
 \theoremstyle{definition}

 \newtheorem{exmp}{Example}[section]
  \newtheorem{ass}{Assumption}[section]
  \usepackage[thinlines]{easytable}
 \theoremstyle{definition}
 \newtheorem{rem}{Remark}
 
 \makeatletter
 \def\BState{\State\hskip-\ALG@thistlm}
 \makeatother
 
 \usepackage{authblk}
 \usepackage{amsfonts}
\usepackage{booktabs}
\usepackage{siunitx}

\def\spacingset#1{\renewcommand{\baselinestretch}%
{#1}\small\normalsize} \spacingset{1}
 \usepackage{caption}
\definecolor{coquelicot}{rgb}{1.0, 0.22, 0.0}

\title{ Experimental Design under Network Interference \footnote{First Version: March, 2020. Previous versions are available at https://arxiv.org/abs/2003.08421. I am grateful to the editor, associate editor and referees for constructive comments. I thank Jelena Bradic, Graham Elliott, James Fowler, Brian Karrer, Karthik Muralidharan, Yixiao Sun, and Kaspar W\"{u}thrich for helpful comments and discussion. I acknowledge support from the Harvard Griffin Fund in Economics and NSF Grant SES 2447088. All mistakes are mine. 
} }
\author{ Davide Viviano \footnote{Department of Economics, Harvard University. Correspondence: dviviano@fas.harvard.edu.}  }
\date{This version: \today }

\begin{document}
\maketitle

\begin{abstract}
This paper studies how to design two-wave experiments in the presence of spillovers for precise inference on treatment effects. We consider units connected through a single network, local dependence among individuals, and a general class of estimands encompassing average treatment and average spillover effects. We introduce a statistical framework for designing two-wave experiments with networks, where the researcher optimizes over participants and treatment assignments to minimize the variance of the estimators of interest, using a first-wave (pilot) experiment to estimate the variance. We derive guarantees for inference on treatment effects and regret guarantees on the variance obtained from the proposed design mechanism. Our results illustrate the existence of a trade-off in the choice of the pilot study and provide sufficient conditions on the pilot's size to achieve near-optimal precision. 
 Simulations using simulated and real-world networks illustrate the advantages of the method. 
\end{abstract}

\noindent%
{\it Keywords:} Experimental Design, Spillovers, Two-wave experimentation, Causal Inference. \\
\vfill
\newpage 

\section{Introduction} 
 \onehalfspacing


This paper studies the design of experiments for inference on treatment effects under network interference. Network interference induces (i) spillovers across units and (ii) statistical dependence. Our goal is to obtain precise estimates of common measures of treatment effects, such as direct, spillover, and overall effects. 

We consider a setting in which individuals are connected in a single network and interact locally (i.e., with neighbors).\footnote{This assumption is known as \emph{local interference} \citep{manski2013identification} and can be tested, for instance, using \citet{athey2018exact}. It is often assumed in practice \citep{egger2019general,dupas2014short,miguel2004worms,bhattacharya2013estimating,duflo2011peer} and studied in theoretical analyses \citep{forastiere2016identification,leung2019treatment,sinclair2012detecting}.}
Unlike typical clustered or saturation-design settings \citep[e.g.,][]{baird2018optimal}, independent clusters are not necessarily available. Instead, researchers observe a single network; they run a pilot (first-wave experiment) to estimate outcome variances and covariances and then optimally select participants and treatments in the main (second-wave) experiment.\footnote{Pilot studies are common practice; see, e.g., \citet{karlan2018failing,karlan2008credit,dellavigna2018predicting}.}
Relevant applications include online and field experiments in which interference naturally occurs \citep[e.g.,][]{karrer2021network,muralidharan2017general}.

The experiment, which we call ``Experiment under Local Interference'' (ELI), is designed to recover one or more estimands, including (i) the overall effect, (ii) the direct effect, and (iii) spillover effects. For example, in a cash transfer program \citep{barrera2011improving}, one may be interested in effects on recipients (direct effects), on non-recipients living near recipients (spillovers), and on their sum (overall effects). For each estimand, we consider a single user-specified estimator linear in observed outcomes; the choice of estimator may arise from a researcher’s spillover model \citep[e.g.,][]{muralidharan2017general,kreindler2023optimal,egger2019general}. 

We rely on two conditions common in economic applications: (i) interference and dependence are local within the network; and (ii) effects may be heterogeneous in observed and low-dimensional network summary statistics.\footnote{Examples consistent with these conditions include models of spillovers in public programs \citep{muralidharan2017general}, cash transfer programs \citep{egger2019general}, health interventions \citep{dupas2014short}, and educational programs \citep{duflo2011peer}.}

We propose the following protocol: (1) we select a small subsample and conduct a pilot study; (2) using pilot information, we then choose participants and treatment assignments in the main experiment; finally, (3) we collect outcomes for participating units. Under this protocol, we  develop a statistical framework for the design and inference of a two-wave experiment (pilot and main) under network interference. 

As a first step, we show that the main experiment is unconfounded if pilot units and their neighbors are excluded from the main experiment; otherwise, this does not need to hold. This restriction creates a design trade-off: a larger pilot yields more precise variance estimates (useful for the second wave) but also imposes stricter constraints on the main experiment. As a result, we should select a pilot that is sufficiently “well separated’’ from (i.e., shares few neighbors with) the rest of the network. We embed this problem in a min-cut optimization that can be solved using off-the-shelf algorithms. We then select participants and treatment assignments in the main experiment to minimize the estimated variance of the target estimator(s).

We derive theoretical guarantees that inform the selection of the pilot. In particular, we establish the rate at which the variance of the two-wave design approaches that of an “oracle’’ design (regret); the oracle selects participants and assignments to minimize the true variance without a pilot and without additional constraints on the main experiment. The regret converges to zero at a rate governed by the inverse of the pilot size and by the ratio of the pilot to the main-experiment size. This, in turn, yields a regret-minimizing pilot size as a function of the main-experiment size and provides intuitive rules of thumb for choosing the pilot. A key step in the proof is to derive lower bounds for the oracle objective under stricter constraints on the experimenter’s decision space.

The optimization problem naturally induces nontrivial dependence among treatment assignments. Motivated by this, we derive asymptotic properties of the estimator under the proposed design, conditioning on the realized assignments. This, in turn, motivates our approach of minimizing the variance of the estimator to achieve precise model-based inference. In a series of extensions, we broaden the framework to (i) incorporate randomization to enable design-based inference and (ii) accommodate partially observed networks.

We conclude with simulation results. The proposed method significantly outperforms state-of-the-art competitors for estimating overall, spillover, and direct effects, especially in the presence of heteroskedasticity and nonzero covariances.

This paper connects to a recent literature in statistics and econometrics that studies experimental design under interference without using pilot information for inference on treatment effects. We show that incorporating a pilot can substantially improve precision. Relevant references include clustered experiments \citep{eckles2017design, taylor2018randomized, ugander2013graph} and saturation designs \citep{baird2018optimal, basse2016analyzing, pouget2018dealing}, which typically assume clustered observations. Additional related work includes \cite{basse2018model}, who assume Gaussian outcomes and no spillover effects; \cite{wager2019experimenting}, who study sequential randomization for optimal pricing under global interference (without focusing on inference for treatment effects); and \cite{kang2016peer}, who analyze encouragement designs (without variance-optimal design). \citet{basse2018limitations} discuss limits of design-based causal inference under interference, and \citet{jagadeesan2017designs, sussman2017elements} design experiments for direct effects only, whereas we consider a broader class of estimands, including overall and spillover effects. \citet{viviano2020policy} studies policy inference and welfare under unobserved networks and clusters. To our knowledge, none of these papers study variance-optimal two-wave designs.

We also relate to a large literature on experimental design in the i.i.d.\ batch setting, including one-stage procedures \citep{harshaw2019balancing, kasy2016experimenters, kallus2018optimal, barrios2014optimal} and two-stage procedures \citep{bai2019optimality, tabord2018stratification}. However, these works do not address network interference. More broadly, we connect to the literature on treatment effects with interference \citep{aronow2017estimating, hudgens2008toward, forastiere2016identification, manski2013identification, leung2019treatment, vazquez2017identification, athey2018exact, goldsmith2013social, savje2017average, ogburn2017causal, kitagawa2020should, viviano2019policy}, which focuses on identification and inference rather than variance-optimal experiments.


\section{Setup} \label{sec:framework}

In this section, we discuss the setup, model, and estimands.

We first introduce necessary notation. We consider $N$ units connected by a binary, symmetric adjacency matrix $\mathbf{A}$, with $\mathbf{A}_{i,j} \in \{0,1\}$. We denote $\mathcal{N}_i = \{j: \mathbf{A}_{i,j} = 1\}$ by the set of neighbors of unit $i$. The adjacency matrix is observed by the researcher. The researcher conducts two experiments: a pilot and the main experiment.
For each unit $i \in \{1, \dots, N\}$, we denote
\[
R_i = 1\Big\{i \text{ is in the main experiment}\Big\}, \quad P_i = 1\Big\{i \text{ is in the pilot experiment}\Big\},
\]
as the participation indicators in the main and pilot experiments, respectively. We let $\sum_{i=1}^N R_i \le \bar{n}$ and $\sum_{i=1}^N P_i \le \bar{m}$ for some pre-specified $\bar{n}$ and $\bar{m}$, which encode constraints on the number of participants in the main experiment and the pilot. For expositional convenience, we assume that $\bar{n}$ is proportional to $N$, i.e., $\bar{n} \propto N$.\footnote{This condition can be relaxed, e.g., by assuming that $\bar{n} \propto N^{c}$ for some constant $c < 1$. In this case, our guarantees continue to hold provided that the conditions on the maximum degree in Assumption~\ref{ass:sparsity} below hold with $\bar{n}$ in lieu of $N$.} We define $n := \sum_{i=1}^N R_i$ as the number of individuals in the main experiment.

Each unit $i$ is associated with an outcome, pre-treatment observables, and a binary assignment, defined as $
(Y_i, T_i, D_i),$
respectively. Here, $T_i$ may depend on $\mathbf{A}$ and covariates (e.g., $T_i = |\mathcal{N}_i|$). Importantly, we assume that $T_i \in \mathcal{T}$, where $\mathcal{T}$ is a discrete set. The discrete support of $T_i$ will play a prominent role when we consider difference-in-means estimators in Example~\ref{exmp:diff_means} below.

We define $R_{\mathcal{N}_i} = \big(R_j\big)_{j \in \mathcal{N}_i}$ and $D_{\mathcal{N}_i} = \big(D_j\big)_{j \in \mathcal{N}_i}$ as the vectors of selection indicators and treatments of the neighbors of individual $i$, and
\[
\small
\begin{aligned}
\mathbf{R} = \big(R_1, \cdots, R_N\big), \quad
\mathbf{D}^R = \big\{ D_i : R_i = 1 \text{ or } R_{j} = 1 \text{ for some } j \in \mathcal{N}_i \big\},
\end{aligned}
\]
as the vector of selection indicators and the vector of treatment assignments for participants and their neighbors, respectively. Similarly, $\mathbf{P} = (P_1, \cdots, P_N)$, $\mathbf{T} = (T_1, \cdots, T_N)$, $\mathbf{D} = (D_1, \cdots, D_N)$, and $\mathbf{1} = (1, \cdots, 1) \in \mathbb{R}^N$.
Throughout our discussion, we fix $\mathbf{A}$ and $\mathbf{T}$ (i.e., $\mathbf{A}$ and $\mathbf{T}$ are non-random), unless otherwise specified. We postulate nonreversible treatments; i.e., for pilot units with $D_i = 1$, the treatment status cannot be changed.



\subsection{Outcome model and dependence}

We let $Y_i(\mathbf{d})$ denote the potential outcome as a function of the treatment assignments $\mathbf{d} \in \{0,1\}^N$, with $Y_i = Y_i(\mathbf{D})$.

\begin{ass}[Potential outcomes]\label{ass:model}
Assume that for all $i \in \{1, \cdots, N\}$, for a known function $g_i: \{0,1\}^{|\mathcal{N}_i|} \to \mathcal{G}$ defined as the exposure mapping and measurable with respect to $\mathbf{A}$ and $\mathbf{T}$, we have
\begin{equation}\label{eqn:first}
\small
\begin{aligned}
Y_i(\mathbf{d}) \;=\; r\Big(\mathbf{d}_i,\, g_i(\mathbf{d}_{\mathcal{N}_i}),\, T_i,\, \varepsilon_i(\mathbf{d})\Big), 
\qquad 
\varepsilon_i(\mathbf{d}) \mid \mathbf{A}, \mathbf{T} \sim \mathcal{P}, 
\qquad \forall \mathbf{d} \in \{0,1\}^N,
\end{aligned}
\end{equation}
where $r(\cdot)$ and $\mathcal{P}$ are possibly unknown, and $\varepsilon_i(\mathbf{d}) = \varepsilon_i(\mathbf{d}')$ for all $\mathbf{d}, \mathbf{d}' \in \{0,1\}^N$. 

Assume in addition that $\mathcal{G}$ is a discrete set. 
\end{ass}

Assumption~\ref{ass:model} states that each individual’s outcome depends only on their neighbors’ treatment assignments through a known function $g_i$. Here, $g_i$ is the exposure mapping \citep{aronow2017estimating} and can be an arbitrary function (with discrete support) of $\mathbf{A}$, $\mathbf{T}$, and the neighbors’ treatments. In addition, once we condition on the exposure mapping, the network affects the outcome variable through arbitrary observables $T_i$. Finally, since $\varepsilon_i(\mathbf{d})$ is constant in $\mathbf{d}$, we write the unobservable simply as $\varepsilon_i$, omitting its argument.

Our framework encompasses several examples of interest, including exposure mappings that depend on the number or share of treated neighbors, whether at least one neighbor is treated, or interactions between the number of treated neighbors and observable characteristics of the neighbors. Assumption~\ref{ass:model} is consistent with local interference assumptions often documented in practice \citep[e.g.,][]{cai2015social} or studied in theoretical analyses \citep[e.g.,][]{leung2019treatment}. Local interference is testable \citep{athey2018exact}. Throughout the rest of our discussion, we denote
\begin{equation}
\mathbb{E}\big[r\big(d, s, l, \varepsilon_i\big)\big] \;=\; m(d,s,l),
\end{equation}
the expectation of the potential outcome evaluated at individual treatment $d$, exposure $s$, and individual-level covariate $l$.

\begin{exmp}\label{exmp:mainex}
\citet{sinclair2012detecting} study spillover effects for political decisions within households. The authors propose a model of the form
\begin{equation}
\small
\begin{aligned}
Y_i \;=\; \beta_0 + \beta_1 D_i 
+ \beta_2 \mathbf{1}\!\Big\{ \textstyle\sum_{j \in \mathcal{N}_i} D_j \ge 1 \Big\}
+ \beta_3 \mathbf{1}\!\Big\{ \textstyle\sum_{j \in \mathcal{N}_i} D_j \ge |\mathcal{N}_i|/2 \Big\}
+ \beta_4 \mathbf{1}\!\Big\{ \textstyle\sum_{j \in \mathcal{N}_i} D_j = |\mathcal{N}_i| \Big\}
+ \varepsilon_i,
\end{aligned}
\end{equation}
where $\mathcal{N}_i$ denotes the set of members in the same household as individual $i$. The model satisfies Assumption~\ref{ass:model} with $T_i = |\mathcal{N}_i|$ and $g_i(\mathbf{d}_{\mathcal{N}_i}) = \sum_{k \in \mathcal{N}_i} \mathbf{d}_k$. \qed
\end{exmp}

\begin{exmp}\label{exmp:2}
Consider the following equation \citep[as in][]{muralidharan2017general}:
\[
Y_i \;=\; \beta_0 + \beta_1 D_i + \beta_2 \frac{\sum_{k \in \mathcal{N}_i} D_k}{\max\{|\mathcal{N}_i|, 1\}} + \varepsilon_i.
\]
Then Assumption~\ref{ass:model} holds with $g_i(\mathbf{d}_{\mathcal{N}_i}) = \sum_{k \in \mathcal{N}_i} \mathbf{d}_k$ and $T_i = |\mathcal{N}_i|$. \qed
\end{exmp}

We allow $(T_i, D_i)$ to exhibit arbitrary dependence. Instead, we impose restrictions on the dependence structure of the unobservables $\varepsilon_i$.

\begin{ass}[One-degree dependence]\label{ass:modelb}
Assume that for all $i \in \{1, \ldots, N\}$,
\[
\small
\begin{aligned}
&\varepsilon_i \ \perp\!\!\!\perp\ \{\varepsilon_j\}_{j \notin \mathcal{N}_i \cup \{i\}} \,\big|\, \mathbf{A}, \mathbf{T}, \\[0.25em]
&(\varepsilon_i, \varepsilon_j) \ \stackrel{d}{=}\ (\varepsilon_{i'}, \varepsilon_{j'}) \,\big|\, \mathbf{A}, \mathbf{T}
\quad \text{for all } (i,j,i',j') \text{ such that } i \in \mathcal{N}_j,\ i' \in \mathcal{N}_{j'},\ T_i = T_{i'},\ T_j = T_{j'}.
\end{aligned}
\]
\end{ass}

The first condition in Assumption~\ref{ass:modelb} states that unobservables of non-adjacent units are independent, while $\varepsilon_i$ and $\varepsilon_{\mathcal{N}_i}$ may be statistically dependent. The second condition states that pairs of neighbors share the same joint distribution whenever their $(T_i, T_j)$ match. One-degree dependence is imposed for expository convenience and can be relaxed to higher-order dependence up to degree $M$, as discussed below.

\begin{rem}[Higher-order dependence]
Extensions to higher-order dependence of degree $M$, formally presented in Appendix~\ref{sec:higher_order_dependence}, read as follows:
\[
\varepsilon_i \ \perp\!\!\!\perp\ \{\varepsilon_j\}_{j \notin \cup_{u=1}^M \mathcal{N}_i^u \cup \{i\}} \,\big|\, \mathbf{A}, \mathbf{T},
\]
where $\mathcal{N}_i^u$ denotes the set of neighbors of degree $u$. In this case, unobservables associated with units separated by more than $M$ edges are independent. \qed
\end{rem}

\subsection{Network topology}

Without further restrictions on the network topology, dependence may be arbitrary, making inference difficult. In this paper, we consider sparse networks in which the maximum degree grows sufficiently slowly relative to $N$.

\begin{ass}\label{ass:sparsity}
Let $\mathcal{N}_{\max}^2 / N^{1/2} = o(1)$, where $\mathcal{N}_{\max} = \max_{i \in \{1, \ldots, N\}} |\mathcal{N}_i| + 1$.
\end{ass}

Assumption~\ref{ass:sparsity} is common in the dependency-graph literature (e.g., \cite{ross2011fundamentals}) and imposes restrictions on the network topology. It holds for economic models with bounded degree \citep[e.g.,][]{de2018identifying}. Economic applications where Assumption~\ref{ass:sparsity} holds include the Add Health Study \cite{jackson2012social} and, in development settings, \cite{cai2015social}, among others.\footnote{See, e.g., footnote~7 in \cite{de2018identifying} and footnote~37, p.~1879, in \cite{jackson2012social}.} Assumption~\ref{ass:sparsity} fails in the presence of a few units (hubs) with very large degree, such as in a star network. Dense networks, although interesting, are outside the scope of this paper.

\subsection{Problem description} \label{sec:discussion}

Our main goal is to minimize the variance on user-specific linear estimator of the form 
\begin{equation}\label{eqn:estimator}
\hat{\Gamma} \;=\; \frac{1}{n} \sum_{i=1}^N R_i\, w_{\mathbf{A},\mathbf{T}}\!\Big(i;\mathbf{R},\mathbf{D}^R\Big)\, Y_i,
\end{equation}
where $w_{\mathbf{A},\mathbf{T}}(\cdot)$ are user-specified (known) function of the main-experiment treatment assignments and the participant indicators $\mathbf{R}$. Examples include simple difference-in-means estimators, stratified weighted differences (e.g., \cite{tabord2018stratification}), and linear-regression estimators. Linearity in $Y_i$ rules out nonlinear outcome estimators such as feasible two-stage generalized least squares.

The ultimate goal is to conduct inference on the estimand
\begin{equation}\label{eqn:tau}
\tau_{\mathbf{A},\mathbf{T}}(\mathbf{R},\mathbf{D}^R)
\;:=\; \frac{1}{n} \sum_{i=1}^N R_i\, w_{\mathbf{A},\mathbf{T}}\!\Big(i;\mathbf{R},\mathbf{D}^R\Big)\,
m\!\Big(D_i, g_i(D_{\mathcal{N}_i}), T_i\Big).
\end{equation}
We write $\tau := \tau_{\mathbf{A},\mathbf{T}}(\mathbf{R},\mathbf{D}^R)$ when clear from context, leaving implicit its dependence on $\mathbf{R}$ and $\mathbf{D}^R$. Following \cite{abadie2017sampling}, we refer to $\tau$ as a model-based estimand since it is a function of $(\mathbf{R},\mathbf{D}^R,\mathbf{A},\mathbf{T})$ and its causal interpretation relies on the researcher’s maintained model that motivates the chosen estimator $\widehat{\Gamma}$.

\begin{exmp}[Difference in means]\label{exmp:diff_means}
Under the assumption that $T_i$ is discrete, let
\[
w\!\Big(i,\mathbf{R},\mathbf{D}^R\Big)
=\begin{cases}
\displaystyle \sum_{l=0}^{\infty} v(l)\!\left[
\frac{I_i(d,s,l)}{\;\sum_{j:R_j=1} I_j(d,s,l)/n\;}
-\frac{I_i(d',s',l)}{\;\sum_{j:R_j=1} I_j(d',s',l)/n\;}
\right], & \text{if } R_i=1,\\[1.25em]
0, & \text{otherwise},
\end{cases}
\]
where $v(l)$ are user-specified weights over individuals with $T_i=l$, and
$I_i(d,s,l)=\mathbf{1}\{D_i=d,\ g_i(D_{\mathcal{N}_i})=s,\ T_i=l\}$. Then
\[
\tau \;=\; \sum_{l=0}^{\infty} v(l)\,\big(m(d,s,l)-m(d',s',l)\big),
\]
for given exposures $(d,s)$ and $(d',s')$. Therefore, in this example, $\tau$ is not a function of $(\mathbf{R},\mathbf{D}^R)$ provided \eqref{eqn:first} holds. Linearity in $Y_i$ follows by construction. \qed
\end{exmp}

\begin{exmp}[Linear regression model]\label{exmp:linear}
Consider the weights
\[
w(i,\cdot)=
\begin{cases}
\bigg[\Big(\frac{1}{n}\sum_{i:R_i=1}\mathbf{X}_i\mathbf{X}_i'\Big)^{-1}\mathbf{X}_i\bigg]^{(3)}, & \text{if } R_i=1,\\
0, & \text{otherwise},
\end{cases}
\]
where $V^{(3)}$ denotes the third entry of a vector $V$ and
$\mathbf{X}_i=\big(1,\ D_i,\ \sum_{k\in\mathcal{N}_i} D_k/|\mathcal{N}_i|\big)$. Suppose $g_i(D_{\mathcal{N}_i})=\sum_{k\in\mathcal{N}_i}D_k$, $T_i=|\mathcal{N}_i|$, and, for coefficients $(\beta_0,\beta_1,\beta_2)$,
\[
m(d,s,l)=\beta_0+\beta_1 d+\beta_2\, s/l.
\]
Then $\tau=\beta_2$, the spillover effect of treating all neighbors (relative to none). Under correct linear specification, $\tau$ equals the structural parameter $\beta_2$ and is independent of $(\mathbf{R},\mathbf{D}^R)$. \qed
\end{exmp}

The above examples underscore that $\tau$ admits a causal interpretation under the model posited by the researcher. Such models are common in experimental economics; see, e.g., \cite{muralidharan2017general,kreindler2023optimal,egger2019general}.
Therefore, whenever $\hat{\Gamma}$ is unbiased for $\tau$ conditional on $(\mathbf{R},\mathbf{D}^R)$, a natural objective is to minimize its conditional variance. Valid confidence intervals for $\tau$ can then use the the conditional variance (see Theorem~\ref{thm:asym_t} and \cite{abadie2017sampling}). 

Minimizing the conditional variance has a long tradition in experimental design, dating back to information-theoretic optimality for i.i.d.\ data and linear models \citep{john1975d}. Specifically, define
\begin{equation}\label{eqn:optimal_variance}
\mathbb{V}_{\bar{n}}^\star
\;=\;
\min_{\mathbf{R},\,\mathbf{D}^R:\ \mathbf{1}^\top \mathbf{R}=\bar{n}}
\ \mathbb{V}\!\Big(\hat{\Gamma}\ \big|\ \mathbf{R},\mathbf{D}^R,\mathbf{A},\mathbf{T}\Big),
\end{equation}
the smallest conditional variance of $\hat{\Gamma}$ (implicitly a function of $\mathbf{A},\mathbf{T}$) with $\bar n$ participants in the main experiment.

We seek an experiment with the following properties: select a pilot and main-experiment participants $(\mathbf{P},\mathbf{R})$, together with a distribution of treatments $\mathbf{D}$, such that
\begin{equation}\label{eqn:oracle}
\begin{aligned}
&\mathbb{E}\big[\hat{\Gamma}\ \big|\ \mathbf{R},\mathbf{D}^R,\mathbf{A},\mathbf{T},\mathbf{P}\big]
\;=\; \tau_{\mathbf{A},\mathbf{T}}(\mathbf{R},\mathbf{D}^R),\\
&\mathbb{V}\big(\hat{\Gamma}\ \big|\ \mathbf{R},\mathbf{D}^R,\mathbf{A},\mathbf{T},\mathbf{P}\big)
\;-\; \mathbb{V}_{\bar n}^{\star} \;\le\; \zeta/\bar n,
\end{aligned}
\end{equation}
while imposing that no more than $\bar n$ units are in the main experiment and no more than $\bar m$ units are in the pilot study. Here, $\zeta\ge 0$ is a user-chosen tuning parameter governing allowable optimization slack (we will take $\zeta=0$ unless otherwise specified).

There are two main considerations. First, researchers may not know the variance of the estimator and will need a pilot to estimate it, raising the question of how to choose the pilot while guaranteeing unbiasedness. Second, the design that minimizes the conditional variance does not necessarily randomize treatments, since its goal is precise inference on a model-based estimand. It is therefore natural to ask whether one can introduce randomization to enable design-based inference on hypotheses of independent interest (e.g., on sharp null hypotheses). For expositional convenience, we focus first on settings where the experiment may not necessarily allow for design-based inference. Section~\ref{sec:randomization} and Appendix \ref{app:randomization} extend our framework to allow for randomization and design-based inference by letting $\zeta > 0$. 

Finally, 
Section~\ref{sec:multiple_estimators} extends the framework to multiple estimands by minimizing the worst-case variance across their corresponding estimators.

\section{Two-wave experiment: formal description} \label{sec:two_wave}

This section presents the experimental protocol. We begin with a brief overview and then formalize the pilot (Algorithm~\ref{alg:pilot}) and the main experiment (Algorithm~\ref{alg:coefficients2}). We defer a complete discussion of the practical choice of tuning parameters to Section~\ref{sec:guide_practice}, which provides an explicit guide for practitioners.

\subsection{Overview of the algorithm}

The algorithm proceeds as follows:

\begin{enumerate}
\item Researchers observe the network $\mathbf A$ and unit types $\mathbf T$, typically from pre-experimental data (e.g., surveys or administrative records).
\item Researchers select a set of pilot participants $\{i: P_i=1\}$ and collect
\[
\{(Y_i, D_i, T_i, D_{\mathcal{N}_i}): P_i=1\}.
\]
Units outside the pilot have treatment fixed at zero. The pilot sample is chosen to have few edges to the non-pilot units, while including some neighbor pairs within the pilot to identify covariances.
\item Using the pilot, researchers estimate conditional outcome variances and covariances. They then choose the participation vector $\mathbf R$ in the main experiment and the treatment assignments $\mathbf D^R$ (for participants and their neighbors) to minimize the conditional variance of $\widehat\Gamma$. Pilot units and their neighbors are excluded from the main experiment. Treatments for all other nonparticipants (including non-pilot units) remain at zero, and pilot assignments remain unchanged.
\item Researchers run the main study and collect
$
\{(Y_i, D_i, T_i, D_{\mathcal{N}_i}, \mathcal{N}_i): R_i=1\}.
$
\item Researchers compute $\widehat{\Gamma}$ as in~\eqref{eqn:estimator} and estimate its variance for inference on $\tau$.
\end{enumerate}

We now provide details on each of these steps.

\subsection{Selection of the pilot: formal algorithm}   \label{sec:pilot_variance}

The first step is selecting the pilot study.
If pilot outcomes inform the main-experiment design, then selecting any pilot unit or any neighbor of a pilot unit for the main experiment makes $\mathbf R$ and $\mathbf D^R$ statistically dependent on those units’ unobservables. This violates unconfoundedness and the first condition in~\eqref{eqn:oracle}.

\begin{algorithm}[!ht]
\caption{Pilot study}\label{alg:pilot}
\begin{algorithmic}[1]
\Require adjacency matrix $\mathbf A$, type vector $\mathbf T$, pilot size $\bar m$, threshold $\delta>0$
\State \textit{Select pilot participants:}
\begin{equation} \label{eqn:opt_pilot}
\mathbf P \;\in\; \operatorname*{arg\,min}_{p\in\{0,1\}^N}
\sum_{i=1}^N \sum_{j\in\mathcal{N}_i} p_i(1-p_j)
\quad\text{s.t.}\quad
\sum_{i=1}^N p_i=\bar m,\;
\sum_{i=1}^N p_i \sum_{j\in\mathcal{N}_i} p_j \ge \delta .
\end{equation}
\State \textit{Assign pilot treatments:}
\[
D_i \sim \text{Bernoulli}(1/2)\ \text{if }P_i=1,
\qquad D_i=0\ \text{otherwise}.
\]
\State \textit{Collect pilot data:}\quad
$\{(Y_i, D_i, T_i, D_{\mathcal{N}_i}): P_i=1\}$.
\State \textit{Estimate variance and covariance functions:}\quad
construct estimators $(\widehat\sigma_p^{2}(\cdot),\widehat\eta_p(\cdot))$ of the functions
$\sigma^2(\cdot)$ and $\eta(\cdot)$ defined in Lemma~\ref{lem:vara}
(see Section~\ref{sec:guide_practice} for practical recommendations).
\State \textbf{Return:}\quad $\mathbf P$ and $(\widehat\sigma_p^{2}(\cdot),\widehat\eta_p(\cdot))$.
\end{algorithmic}
\end{algorithm}

Figure~\ref{fig:1} illustrates the issue. In the figure, the pilot set includes nodes N4, N5, and N6. Because their outcomes inform the main-experiment design, treatments depend on the unobservables of these pilot units. Since pilot units are statistically dependent on their neighbors (e.g., N7), selecting N7 would make selection depend on both treatments and unobservables in the main experiment, thereby confounding the design. First, define
\begin{equation}\label{eqn:pilot_identity}
\mathcal{J} \;=\; \bigcup_{i:P_i=1} \big(\mathcal{N}_i \cup \{i\}\big),
\end{equation}
the set of pilot units and their neighbors. The experiment is unconfounded if it satisfies the following restrictions.

\begin{prop}[Unconfounded main-experiment assignments]\label{ass:cov}
Let $\mathcal J$ be as in Equation~\eqref{eqn:pilot_identity} (the set of pilot units and their neighbors).
Suppose Assumptions~\ref{ass:model} and~\ref{ass:modelb} hold, and that:
\[
\text{(i) } \{\varepsilon_i\}_{i\notin\mathcal J}\ \perp\!\!\!\perp\ (\mathbf R,\mathbf D^R)\ \big|\ \mathbf A,\mathbf T,\mathbf P;\qquad
\text{(ii) } \{\varepsilon_i\}_{i=1}^N\ \perp\!\!\!\perp\ \mathbf P\ \big|\ \mathbf A,\mathbf T;\qquad
\text{(iii) } R_i=0 \ \forall\, i\in\mathcal J.
\]
Then
\[
\mathbb{E}\!\left[\widehat\Gamma \,\big|\, \mathbf R, \mathbf D^R, \mathbf A, \mathbf T, \mathbf P \right]
= \tau_{\mathbf A,\mathbf T}(\mathbf R,\mathbf D^R).
\]
\end{prop}

The proof is in Appendix~\ref{proof:ass:cov}.
Proposition~\ref{ass:cov} provides sufficient conditions for unbiasedness conditional on $\mathbf R$ and the treatment assignments.

The first condition states that unobservables for units outside $\mathcal J$ (i.e., all units except the pilot units and their neighbors) are independent of assignment and selection, conditional on $(\mathbf A,\mathbf T,\mathbf P)$. The second condition requires that pilot selection depend only on $(\mathbf A,\mathbf T)$. The third condition excludes the pilot units and their neighbors from the main experiment.

Proposition~\ref{ass:cov} yields two insights: (a) pilot participants can be selected using only network information (and types); (b) the larger the set $\mathcal J$ (pilot units plus their neighbors), the stricter the exclusion constraint on the second wave.

Accordingly, Algorithm~\ref{alg:pilot} chooses pilot units to minimize the number of units in the pilot with friends outside the pilot study. It also requires that some neighbor pairs are within the pilot to identify and estimate covariances; treatments in the pilot are then randomized. The optimization problem in~\eqref{eqn:opt_pilot} is min-cut optimization program: it finds a set of units that are well separated from the remainder, subject to constraints on the number of pilot units and on the number of within-pilot neighbor pairs.

The corollary illustrates that the proposed algorithm yields unbiased estimators of treatment effects.

\begin{cor}
Let Assumptions~\ref{ass:model}, \ref{ass:modelb} hold. Then the two-wave experiment constructed with Algorithm~\ref{alg:pilot} and Algorithm~\ref{alg:coefficients2} satisfies the conditions in Proposition~\ref{ass:cov}.
\end{cor}

\begin{figure}
 \centering
    \begin{tikzpicture}
    \node[draw, circle] (a) at (0,0) {N1};
    \node[draw, circle] (b) at (2,1) {N2};
    \node[draw, circle] (d) at (4,1.5) {N3};
    \node[draw,  circle] (h) at (-4,0) {N4};
    \node[draw, circle] (i) at (-2,0.3) {N5};
    \node[draw, circle] (u) at (-2,2) {N6};
    \node[draw, fill = green, circle] (ii) at (0.2,2) {N7};
    \node[ellipse, draw=gray, fit= (h) (i) (u), inner sep=-0.1mm] (all) {};
    \node[] at (-3,-0.5) {$Pilot$};

    \draw [-] (a) edge (b) ;
    \draw[-] (d) edge (b);
    \draw[-]  (h) edge (i) (h) edge (u) (u) edge (ii) (ii) edge (b);
    \end{tikzpicture}
\caption{Illustrative network. Under one-hop dependence, node N7 violates the unconfoundedness condition if included in the main experiment because it is adjacent to the pilot set (N4–N6), whose outcomes informed randomization.}
\label{fig:1}
\end{figure}

Finally, the main statistical goal of the pilot is to learn the variance and covariance functions:
$
\mathrm{Var}\Big(Y_i\Big|A, D_i, T_i, D_{\mathcal{N}_i}, \mathbf{P}\Big), \quad \mathrm{Cov}\Big(Y_i, Y_j\Big|A, D_i, D_j, D_{\mathcal{N}_i}, D_{\mathcal{N}_j}, T_i, T_j, \mathbf{P}\Big).
$

The following lemma provides conditions for identification.

\begin{lem}\label{lem:vara}
Suppose Assumptions~\ref{ass:model} and~\ref{ass:modelb} hold and the pilot is chosen as in Algorithm~\ref{alg:pilot}. Then, for all $i,j$ with $P_i=P_j=1$,
\[
\small
\begin{aligned}
\mathrm{Var}\big(Y_i \mid \mathbf A, D_i, T_i, D_{\mathcal{N}_i}, \mathbf P\big)
&= \sigma^2\!\big(T_i, D_i, g_i(D_{\mathcal{N}_i})\big),\\
\mathrm{Cov}\big(Y_i, Y_j \mid \mathbf A, D_i, D_j, D_{\mathcal{N}_i}, D_{\mathcal{N}_j}, T_i, T_j, \mathbf P\big)
&=
\begin{cases}
\eta\!\big(T_i, D_i, g_i(D_{\mathcal{N}_i}),\, T_j, D_j, g_j(D_{\mathcal{N}_j})\big), & \text{if } i\in\mathcal{N}_j,\\
0, & \text{otherwise}.
\end{cases}
\end{aligned}
\]
for some functions $\sigma^2(\cdot)$ and $\eta(\cdot)$.
\end{lem}

The proof is in Appendix~\ref{proof:lemma:vara}. Building on the lemma above, estimation of $\sigma^2$ and $\eta$ can be carried out by a variety of methods; the rate of convergence affects the precision of the estimator in the main experiment. We denote by
\[
\big(\widehat\sigma_p^{2},\, \widehat\eta_p\big)
\]
the variance and covariance functions estimated from the pilot study. Section~\ref{sec:guide_practice} and Algorithm~\ref{alg:pilot_varcov} provide concrete estimation examples and practical recommendations for implementing Algorithm~\ref{alg:pilot}, including the choice of~$\delta, \widehat\sigma_p^{2},\, \widehat\eta_p$.

 \subsection{Main experiment: formal algorithm}

We now discuss the main experiment (Algorithm~\ref{alg:coefficients2}). Throughout, we omit the explicit dependence of the weights on $(\mathbf A,\mathbf T)$ and write $w(i;\mathbf R,\mathbf D^R)$ for brevity. 

Given the pilot selection $\mathbf P$, the main experiment design minimizes a plug-in estimate of the variance. Formally, define
\begin{equation}\label{eqn:var_est_pilot}
\small
\begin{aligned}
\widehat V_{\widehat\sigma_p,\widehat\eta_p}(\mathbf R,\mathbf D^R)
&= \frac{1}{n^2}\sum_{i:R_i=1} w^2(i;\mathbf R,\mathbf D^R)\;
  \widehat\sigma_p^{2}\!\Big(T_i, D_i, g_i(D_{\mathcal{N}_i})\Big) \\
&\quad + \frac{1}{n^2}\sum_{i:R_i=1}\sum_{j\in\mathcal{N}_i} R_j\,
  w(i;\mathbf R,\mathbf D^R)\, w(j;\mathbf R,\mathbf D^R)\;
  \widehat\eta_p\!\Big(T_i, D_i, g_i(D_{\mathcal{N}_i}),\; T_j, D_j, g_j(D_{\mathcal{N}_j})\Big).
\end{aligned}
\end{equation}
The first term captures heteroskedastic variances; the second captures covariances between neighbors, both estimated from the pilot.\footnote{In the presence of higher-order interference, we also add additional components which depend on higher-degree neighbors. See Appendix \ref{sec:higher_order_dependence}. } 

\begin{algorithm}[!ht]
\caption{Main experiment}\label{alg:coefficients2}
\begin{algorithmic}[1]
\Require $\mathbf A$, $\mathbf T$, pilot output $\big(\mathbf P,\widehat\sigma_p^{2}(\cdot),\widehat\eta_p(\cdot)\big)$, bounds $\underline n,\bar n$, and $\mathcal J$ as in~\eqref{eqn:pilot_identity}.
\State \textit{Design optimization:} find $(\mathbf{D}^R, \mathbf{R}$) by optimizing 
\begin{equation}\label{eqn:design1}
\small
\begin{aligned}
(\mathbf D^R,\mathbf R)\;\in\;
&\operatorname*{arg\,min}_{\mathbf r\in\{0,1\}^N,\;\mathbf d^{\,r}\in\{0,1\}^N}
\ \widehat V_{\widehat\sigma_p,\widehat\eta_p}\big(\mathbf r,\mathbf d^{\,r}\big)\\
\text{s.t.}\quad
&\mathbf 1^\top \mathbf r \in [\underline n,\bar n],\qquad
\mathbf{r}_j=0\ \ \forall j\in\mathcal J.
\end{aligned}
\end{equation}
\State \textit{Run the main study and collect data:}\quad
$\{(Y_i, D_i, T_i, D_{\mathcal{N}_i}, \mathcal{N}_i): R_i=1\}$.
\State \textit{Estimate the outcome regression:} obtain a consistent (possibly non-parametric) estimator $\hat m$ of $m$ and define
$\hat m_i := \hat m\big(D_i,g_i(D_{\mathcal{N}_i}),T_i\big)$.
\State \textit{Compute the estimator and its variance:} construct $\widehat\Gamma$ as in~\eqref{eqn:estimator} and estimate its variance by
\begin{equation}\label{eqn:final_vv}
\widehat V
= \frac{1}{n^2}\sum_{i=1}^N R_i\, w(i;\mathbf R,\mathbf D^R)\,\big(Y_i-\hat m_i\big)
\sum_{j\in\mathcal{N}_i\cup\{i\}} R_j\, w(j;\mathbf R,\mathbf D^R)\,\big(Y_j-\hat m_j\big).
\end{equation}
\State \textbf{Return:}\quad $\widehat\Gamma$ and $\widehat V$.
\end{algorithmic}
\end{algorithm}

The optimization problem is in Equation \eqref{eqn:design1}. The minimization is with respect to the participation indicators and the treatment assignments. The optimization problem selects a number of participants in the interval $n \in \{\underline{n}, \underline{n}+ 1, \cdots,  \bar{n}\}$, with $\underline{n} < \bar{n}$, denoting a lower bound on the number of participants.  The upper bound $\bar{n}$ typically arises due to cost constraints for the researcher. The lower bound $\underline{n}$  guarantees that sufficiently many units are selected in the main experiment useful in our theoretical derivation. In particular, our theoretical guarantees (Section~\ref{sec:asymptotic}) will impose that $\underline{n}/\bar{n} \in (0,1)$, requiring researchers some slackness factors in the smallest and largest sample size (we recommend $\underline{n} \ge 2 \bar{n}/3$). In practice, such lower bound is non-binding when the estimator's standard error is decreasing in the same size, although it can improve optimization error in few instances. Additional constraints on $R_i$ or $D_i$, although omitted for brevity, may be included without affecting our theoretical guarantees.\footnote{For example, only some units can participate in the experiments, corresponding to constraints on $R_i = 0$ for some of the units. An alternative constraint is to impose $D_i \times R_i \ge D_i$. This constraint imposes that those units which are not selected as participants have treatment assignments equal to zero.}

The constraint in Equation \eqref{eqn:design1} illustrates the trade-off in the selection of the pilot study: the larger the pilot study, the more precise the estimator of the variance. However, the larger the pilot study, the larger the set $\mathcal{J}$ and therefore, the more stringent the constraint imposed in the above optimization procedure.

Finally, Algorithm~\ref{alg:coefficients2} returns a variance estimator for inference on the main effect $\tau$. Consistency of the variance estimator $\widehat{V}$ in Equation \eqref{eqn:final_vv} requires consistent estimation of $m(d,s,l)$ using data from the main experiment (at a possibly slow rate). As formalized in Theorem~\ref{thm:asym_t} below, $\hat m$ may be a parametric (as in Example \ref{exmp:linear}) or nonparametric estimator (e.g., \citealp{leung2019treatment}) depending on the researcher's modeling assumption. Precise conditions are collected in Assumption~\ref{ass:inference_conditions}.
 

\begin{rem}[Temporal structure]
Pilot and main experiments are typically conducted sequentially. We assume treatments do not alter the network between the two waves. This is plausible in applications where the network is time-invariant \citep[e.g.,][]{egger2019general,muralidharan2017general}, the intervention does not affect link formation by design \citep{cai2015social}, or small pilot interventions do not meaningfully perturb large platforms \citep[e.g.,][]{karrer2021network}. The assumption may fail when interventions reshape the network (e.g., group-formation experiments; \citealp{basse2024randomization}), which we do not study. \qed
\end{rem}

\begin{rem}[Partial network information]
Appendix~\ref{sec:model_assisted} studies a design that uses only partial network information. The researcher observes a subset of entries of $\mathbf A$ and imputes the rest using a model, under the assumption that the pilot forms a cluster separated from the main experiment. \qed
\end{rem}

\section{Theoretical analysis and inference} \label{sec:asymptotic}

Next, we study how the variance of the estimator obtained from the two-wave experiment in Section~\ref{sec:two_wave} compares with the variance of the oracle study. Specifically, we study
\[
\mathbb{V}\!\Big(\widehat{\Gamma}\ \big|\ \mathbf{R}, \mathbf{D}^R, \mathbf{A}, \mathbf{T}\Big) \;-\; \mathbb{V}_{\bar{n}}^\star,
\]
where $(\mathbf{R},\mathbf{D}^R)$ solve the two-wave design problem in~\eqref{eqn:design1}. We refer to $\mathbb{V}_{\bar{n}}^\star$ as the variance attainable by an oracle experimentalist who optimally selects participants and assigns treatments with ex-ante perfect knowledge of the population variance and covariance functions (without a pilot). The oracle selects exactly $\bar{n}$ individuals for the experiment, corresponding to the upper bound in our feasible design.\footnote{We use this restriction in our theoretical results; it can be relaxed to an oracle that selects any number of participants between $\underline{n}+|\mathcal{J}|$ and $\bar{n}$. In the asymptotic regime where the pilot size grows more slowly than the main-experiment size, we have $|\mathcal{J}|/\underline{n} \lesssim \mathcal{N}_{\max}\,\bar{m}/\underline{n} = o(1)$ for $\underline{n}$ large relative to the maximum degree $\mathcal{N}_{\max}$ and the pilot size $\bar{m}$.}

Before stating our first theorem, we impose the following conditions.

\begin{ass}\label{ass:convergence}
Assume that for some $\xi>0$,
\[
\small
\begin{aligned}
\sup_{d,s,l}\ \Big|\widehat{\sigma}_p^{2}(s,d,l) - \sigma^{2}(s,d, l)\Big| &= \mathcal{O}_p(\bar{m}^{-\xi}),\\
\sup_{d,s,l,\; d',s',l'}\ \Big|\widehat{\eta}_p(s,d,l,s',d', l') - \eta(s,d,l, s',d',l')\Big| &= \mathcal{O}_p(\bar{m}^{-\xi}).
\end{aligned}
\]
\end{ass}

Assumption~\ref{ass:convergence} quantifies the pilot-based convergence of the variance and covariance functions. For parametric estimators, the rate is typically $\bar{m}^{-1/2}$. This assumption is not required for inference in the main experiment, but it is invoked to study regret properties: if the pilot estimates poorly approximate the variance and covariance functions, the design may be less precise.

\begin{ass}\label{ass:moment}
Let $Y_i \in [-M,M]$ for some $M<\infty$.
\end{ass}

Assumption~\ref{ass:moment} imposes bounded outcomes.\footnote{One can weaken it to sub-Gaussian tails at the cost of heavier notation.} The final assumption is a stability condition on the weights.

\begin{ass}[Weight stability]\label{ass:weights_stable}
(i) For any $\mathbf d\in\{0,1\}^N$ and any $i\in\{1,\ldots,N\}$,
$
\Big|\,w\big(i;\mathbf r,\mathbf d^{\,r}\big) - w\big(i;\mathbf r',\mathbf d^{\,r'}\big)\,\Big|
\;\le\; \bar C\,\frac{||\,\mathbf r -  \mathbf r'\,||_1}{\min\{\mathbf 1^\top \mathbf r,\ \mathbf 1^\top \mathbf r'\}},
$
for some finite constant $\bar C<\infty$, for any selection vectors $\mathbf r,\mathbf r'$ such that $r_i=r'_i$ and $r_{\mathcal N_i}=r'_{\mathcal N_i}$.\\
(ii) Moreover, $\max_{i:\,R_i=1}\big|\,w(i;\mathbf R,\mathbf D^R)\,\big| \le \bar C' < \infty$ almost surely, for some finite constant $\bar C'$.
\end{ass}

Assumption~\ref{ass:weights_stable}(i) states that, holding the inclusion of unit $i$ and its neighbors fixed, the weight for unit $i$ cannot change by more than a constant multiple of the relative change in sample size across two designs; this is a mild stability requirement. Part~(ii) requires that the resulting design produce bounded weights. This is typically satisfied by the algorithm, since unbounded weights would make the variance objective equal to infinity and can also be enforced directly as an additional constraint in Algorithm~\ref{alg:coefficients2}.

 \paragraph{Example \ref{exmp:diff_means} cont'd}
Consider the weights in Example~\ref{exmp:diff_means}. Then Assumption~\ref{ass:weights_stable} holds if the weights are bounded by a universal constant, i.e., for any $l$ in the support
\[
\, \min \left\{ \frac{\sum_{j} I_j(d,s,l)}{n},\ \frac{\sum_{j} I_j(d',s',l)}{n} \right\} \;>\; 1/\bar{C},
\]
with $I_j(\cdot)$ as defined in Example~\ref{exmp:diff_means}, for some finite $\bar{C}>0$. Therefore, Assumption~\ref{ass:weights_stable} holds for the difference-in-means estimator under a uniform boundedness restriction on the weights, which can be imposed by design in Algorithm~\ref{alg:coefficients2}. \qed

\begin{thm}\label{thm:regret}
Under Assumptions~\ref{ass:model}, \ref{ass:modelb}, \ref{ass:sparsity}, \ref{ass:convergence}, \ref{ass:moment}, and \ref{ass:weights_stable}, suppose $\bar{n}/\underline{n}=\alpha\in(1,C)$ for a universal constant $C<\infty$, and $\underline{n}\ge \mathcal{N}_{\max}\,\bar{m}/(\alpha-1)$. Then
\[
\bar{n}\!\left[\mathbb{V}\!\Big(\widehat{\Gamma}\ \big|\ \mathbf{R}, \mathbf{D}^R, \mathbf{A}, \mathbf{T}\Big) - \mathbb{V}_{\bar{n}}^\star \right]
\;\le\; \mathcal{O}_p\!\Big(\mathcal{N}_{\max}\,\bar{m}^{-\xi} \;+\; \frac{\mathcal{N}_{\max}^2\,\bar{m}}{\underline{n}}\Big).
\]
\end{thm}

The proof is in Appendix \ref{proof:thm:regretb}. Theorem~\ref{thm:regret} characterizes the difference between the variance of the experiment with a pilot and the variance of the oracle experiment with known variance and covariance functions. The result highlights a trade-off: (i) a larger pilot reduces estimation error, and (ii) a larger pilot also tightens the design constraints, potentially increasing regret relative to the oracle. A key challenge in the proof is comparing the constrained design with its oracle counterpart, which assigns treatments without restrictions induced by the pilot study.

An important assumption is that the network is sufficiently sparse. For instance, we require the minimum experiment size $\underline{n}$ to exceed (by a proportional factor) the product of the maximum degree and the pilot size. This condition is satisfied in large samples when $\mathcal{N}_{\max}$ is small relative to $N$ (and hence to $\underline{n}$), i.e., under Assumption~\ref{ass:sparsity}. It fails in dense networks such as a star, where $\mathcal{N}_{\max}=N$.

\begin{cor}\label{cor:pilot}
Suppose the conditions of Theorem~\ref{thm:regret} hold with $\xi=1/2$ (parametric rate). By choosing $\underline{n}/(\mathcal{N}_{\max}^2\log \underline{n}) \ge \bar{m} \ge (\underline{n}/\mathcal{N}_{\max})^{2/3}$ in Algorithm~\ref{alg:pilot}, and main study in Algorithm~\ref{alg:coefficients2}, we have
\[
\bar{n}\!\left[\mathbb{V}\!\Big(\widehat{\Gamma}\ \big|\ \mathbf{R}, \mathbf{D}^R, \mathbf{A}, \mathbf{T}\Big) - \mathbb{V}_{\bar{n}}^\star \right] \;\le\; o_p(1).
\]
\end{cor}

To our knowledge, this is the first result that formally provide sufficient conditions on the size of the pilot to achieve near-optimal precision. 

\subsection{Asymptotic inference} 

We conclude with asymptotic inference for the target estimand.

\begin{ass}[Conditions for inference]\label{ass:inference_conditions}
Suppose: (i) $n\,\mathbb{V}(\widehat{\Gamma}\mid \mathbf{R}, \mathbf{D}^R, \mathbf{A}, \mathbf{T})>\underline{v}$ for a constant $\underline{v} > 0$ almost surely; (ii) $\max_{d,s,t}|\hat{m}(d,s,t)-m(d,s,t)|=o_p(\underline{n}^{-1/4})$ and $\max_{d,s,t}|\hat{m}(d,s,t)|\le c_0$ for some constant $c_0<\infty$, where $\hat{m}$ is the conditional-mean estimator used for variance estimation in Step~4 of Algorithm~\ref{alg:coefficients2}.
\end{ass}

Condition~(i) rules out degeneracy of the variance. Condition~(ii) requires that the conditional-mean estimator $\hat{m}$ used in the variance formula be consistent at a rate as slow as $o_p(n^{-1/4})$. The condition can be met by semiparametric estimators that converge more slowly than $n^{-1/2}$.

\begin{thm}\label{thm:asym_t}
Suppose Assumptions~\ref{ass:model}, \ref{ass:modelb}, \ref{ass:sparsity}, \ref{ass:moment}, \ref{ass:weights_stable}\,(ii), and \ref{ass:inference_conditions} hold, and $\underline{n} \propto \bar{n} (\propto N)$. Then
\begin{equation}\label{eqn:thm1}
\frac{\sqrt{n}\,\big(\widehat{\Gamma} - \mathbb{E}[\widehat{\Gamma}\mid \mathbf{R}, \mathbf{D}^R, \mathbf{A}, \mathbf{T}]\big)}{\sqrt{n\,\widehat{V}}}
\ \rightarrow_d\ \mathcal{N}(0,1).
\end{equation}
\end{thm}

The proof is in Appendix \ref{sec:asymp_app}. Theorem~\ref{thm:asym_t} establishes asymptotic normality for a broad class of linear estimators. Notably, inference does not require rate conditions for the pilot estimators of the variance and covariance functions. Asymptotic properties of estimators for network data have been studied in various contexts \citep[e.g.,][]{ogburn2017causal, chin2018central, savje2017average}. Here, we derive asymptotics conditionally on the entire assignment mechanism.



\section{Main extensions} 

In this section, we sketch main extensions and defer formal details to Appendix \ref{sec:extension_app}.  

\subsection{Randomization for design-based inference (Appendix \ref{app:randomization})}\label{sec:randomization}

Appendix Algorithm~\ref{alg:coefficients3b} modifies Algorithm~\ref{alg:coefficients2} to allow for randomization. Specifically, first, we compute the optimal selection of participants and assignments in the main experiment using pilot information as before. 
Second, we re-randomize treatments: for each experimental participant $i$ with $R_i=1$, we independently randomize the treatment assignment to be equal to the optimal assignment with probability $1-\iota$, and accept the proposed new assignment only if its plug-in variance does not exceed the plug-in variance at the optimum by more than a small slack $\zeta/\bar n$. If the proposal fails this variance cap, we re-randomize until it passes.

Algorithm \ref{alg:coefficients3b} is designed so that, in addition to the model-based inference of Theorem~\ref{thm:asym_t} (which remains valid), one may perform design-based inference on hypotheses of independent interest. This is possible using e.g., procedures in  \citealp{puelz2022graph} by simulating from the same randomization scheme in Step 2 of Algorithm \ref{alg:coefficients3b}, after conditioning on the optimal solutions in Equation \eqref{eqn:design1b}.

The parameters $\iota$ and $\zeta$ trade off additional randomization against precision for model-based inference: larger $\iota$ increases randomization away from the optimizer; larger $\zeta$ admits more variance inflation. We formalize this intuition in Appendix Theorem \ref{thm:regretb} where we show that the regret bound on the conditional variance now holds as in our main Theorem \ref{thm:regret} up to the slack parameter $\zeta$.

 \subsection{Multiple estimands (Appendix \ref{app:multiple})}\label{sec:multiple_estimators}

As a second extension, we extend the results to settings with a finite set of estimands, each paired with a (linear) estimator. 
Let $\mathcal W=\{w^1,\ldots,w^E\}$ with $E<\infty$. 
As in~\eqref{eqn:estimator}, and omitting the explicit $(\mathbf A,\mathbf T)$ dependence of the weights for brevity, consider
\begin{equation}\label{eqn:estimator2}
\widehat\Gamma(w)
\;=\;
\frac{1}{n}\sum_{i=1}^N R_i\, w\!\big(i;\mathbf R,\mathbf D^R\big)\, Y_i ,
\end{equation}
with corresponding model-based estimand
$
\tau(w)
\;=\;
\frac{1}{n}\sum_{i=1}^N R_i\, w\!\big(i;\mathbf R,\mathbf D^R\big)\,
m\!\big(D_i,\, g_i(D_{\mathcal{N}_i}),\, T_i\big),
$
which is implicitly a function of $(\mathbf A,\mathbf T,\mathbf R,\mathbf D^R)$. 
Returning to Example~\ref{exmp:linear}, one might be interested separately in direct and spillover effects; then $(w^1,w^2)$ pick out two coordinates of the least-squares weight vector.  In principle, it is possible that multiple weights $w$ may correspond to the same estimand $\tau(w)$ under certain modeling assumptions. We abstract from this complication and, motivated by empirical practice, we consider a scenario where, for each estimand, researchers consider a single estimator. In addition, we let $E$ to be finite.  

The core idea is to conduct the experiment as in Algorithm \ref{alg:coefficients2}, by minimizing the worst-case conditional variance of each of the estimators. Inference and regret bounds follow verbatim as in Section \ref{sec:asymptotic} and formalized in Appendix Theorem \ref{thm:regret2}.

\section{Implementation guide and numerical studies} 

In this section we first discuss the choice of the tuning parameters and estimators in the experiments. We then provide a set of numerical studies.

 \subsection{Guide to practice: implementation details}\label{sec:guide_practice}

Algorithms~\ref{alg:pilot} and~\ref{alg:coefficients2} require (i) tuning choices and (ii) pilot-based estimators of the outcome variance and covariances. Below we give practical defaults.

\begin{itemize}

\item \textbf{Optimization for the pilot units.}  
The first step is the selection of the pilot study. 
Algorithm~\ref{alg:pilot} is a mixed-integer quadratic program (MIQP). Although NP-hard, off-the-shelf solvers are efficient in the problem sizes we study.
In our numerical studies, selecting a pilot of size $\bar m=70$ from $N=800$ (by solving exactly Equation \eqref{eqn:opt_pilot}) takes only a few seconds on a laptop.\footnote{We use \url{https://www.gurobi.com}, which is free for academic use.}   For very large $N$, one can use a stopping rule by stopping when the solver’s gap bound is below a user threshold  \citep{huang2021branch}.\footnote{Our regret and inference results (under network sparsity in Assumption \ref{ass:sparsity}) continue to hold even if~\eqref{eqn:opt_pilot} is not solved to global optimality. The theory imposes conditions on the main-experiment selection $\mathbf R$; it imposes no optimality conditions on $\mathbf P$ beyond $\mathbf P$ being only a function of $(\mathbf A,\mathbf T)$; see Proposition~\ref{ass:cov}.}  

A key tuning parameter is $\delta$ in~\eqref{eqn:opt_pilot}. It governs the effective pilot sample size for covariance estimation. Larger $\delta$ improves covariance precision but tightens the main experiment constraints for the selection of the pilot units. As a rule of thumb, we recommend setting
$
\delta \approx \max\{\bar m/4,\;30\},
$
so that a nontrivial fraction of pilot units have at least one neighbor in the pilot.

\item \textbf{Variance estimation using the pilot.}   The second step is the estimation of the variance and covariance function using information from the pilot study.

Any pilot estimator of the variance function may be used; the validity of Theorem~\ref{thm:asym_t} does not hinge on the particular choice (or even consistency of variance estimators from the pilot study). However, better pilot estimates typically yield a more efficient estimator from the main experiment.  

A formal algorithm for the variance and covariance estimation from the pilot study is in the Appendix Algorithm \ref{alg:pilot_varcov}. We provide a brief description below.

Let $W_i:=f\!\big(T_i,D_i,g_i(D_{\mathcal{N}_i})\big)$ be a user-chosen transformation (e.g., a simple identity function or a more flexible polynomial transformation). Let $\hat m_p(d,s,l)$ be a pilot-based estimator of $m(d,s,l)$ aligned with the outcome model the researcher considers for estimation of $\widehat{\Gamma}$ (e.g., OLS for Example~\ref{exmp:linear}; differences-in-means for Example~\ref{exmp:diff_means}). Define pilot residuals
$
\hat\varepsilon_i:=Y_i-\hat m_p\!\big(D_i,g_i(D_{\mathcal{N}_i}),T_i\big)$ for $P_i=1.
$
A simple and flexible pilot variance estimator is the nonnegative least-squares fit
\begin{equation} \label{eqn:vv}
\small 
\begin{aligned} 
\widehat\sigma_p^2\!\big(T_i,D_i,g_i(D_{\mathcal{N}_i})\big)\;=\;\max\{W_i^\top\hat\beta,\,0\},
\qquad
\hat\beta\in\arg\min_{\beta:\,W_i^\top\beta\ge 0 \forall i: P_i = 1}\sum_{i:P_i=1}\big(\hat\varepsilon_i^2-W_i^\top\beta\big)^2,
\end{aligned} 
\end{equation} 
which regresses squared residuals on $W_i$ subject to positivity. This captures heteroskedasticity driven by $(T_i,D_i,g_i)$ while guaranteeing $\widehat\sigma_p^2\ge 0$.

\item \textbf{Covariance estimation using the pilot.}   
For neighbor pairs, a convenient parametric form for the covariance function is
$$ 
\small 
\begin{aligned} 
\widehat\eta_p\!\Big(T_i,D_i,g_i(D_{\mathcal{N}_i}),\,T_j,D_j,g_j(D_{\mathcal{N}_j})\Big)
\;=\;
\alpha\;\widehat\sigma_p\!\big(T_i,D_i,g_i(D_{\mathcal{N}_i})\big)\;
\widehat\sigma_p\!\big(T_j,D_j,g_j(D_{\mathcal{N}_j})\big), 
\end{aligned} 
$$ 
assuming a constant correlation $\alpha$.\footnote{This is the analogue of intra-cluster correlation in studies with cluster experiments often assumed to be homogeneous, e.g., \cite{baird2018optimal}} We estimate $\alpha$ by regressing $\hat\varepsilon_i\hat\varepsilon_j$ on the product
$\widehat\sigma_p(\cdot)\widehat\sigma_p(\cdot)$ for individuals $(i,j)$ with $P_i=P_j=1$ and $j\in\mathcal{N}_i$. 

Additional constraints on $\alpha$ may also be added in the estimation step based on prior knowledge. For example, in many applications researchers may expect positive but small correlation, in which case researchers may impose upper and lower bounds on $\alpha$. If the particular application specifics suggest heterogeneity in correlations, researchers may also allow $\alpha$ to vary by observable characteristics.

\item \textbf{Optimization in the main experiment.}  Given the variance and covariance estimator, the next step is to solve over treatments and selection of participants in the main experiment. That is, given $(\widehat\sigma_p^2,\widehat\eta_p)$, Algorithm~\ref{alg:coefficients2} solves a nonlinear mixed-integer program in $(\mathbf R,\mathbf D^R)$, which is typically NP-hard. In our numerical studies we use a solver based on Ant Colony Optimization \citep{dorigo2005ant}.\footnote{In our experiments we use as a software \url{https://www.midaco-solver.com}.} This meta-heuristic performs well on large combinatorial instances. Even with a hard time limit, we find sizable variance reductions relative to competitors, making this our recommended choice.

Importantly, an approximate (instead of exact) optimization routine does not invalidate inference in Theorem~\ref{thm:asym_t}; it only adds an additional term (equal to the approximation error) to the regret bound.\footnote{The regret bound with an approximation error follows verbatim the one formally derived in Appendix Theorem~\ref{thm:regret2} with $\zeta/\bar{n}$ characterizing the optimization error (for arbitrary $\zeta$).} 
 
\item \textbf{Inference using the main experiment.}  Once the main experiment is conducted, 
inference on $\tau$ follows Theorem~\ref{thm:asym_t}, using the variance estimator~\eqref{eqn:final_vv}. Estimation of the variance needs a consistent estimator $\hat m(d,s,l)$ of $m(d,s,l)$; it can be parametric or nonparametric (as in Example~\ref{exmp:diff_means}).\footnote{The required rate is $o_p(n^{-1/4})$ under sparsity of the network in Assumption \ref{ass:sparsity}.} 

\item \textbf{Choosing sample sizes (pilot and main).} We conclude with guidance on selecting the pilot size $\bar m$ and the main-experiment bounds $(\underline n,\bar n)$. For sparse networks, a simple rule of thumb consistent with Corollary~\ref{cor:pilot} is
$
\bar m \;\gtrsim\; \bar n^{\,2/3}.
$
For example, in our numerical studies we set $\bar m=70$ when $\bar n=400$; for an experiment with about one thousand individuals, we recommend a pilot of roughly one hundred units.

On the other hand, the upper bound $\bar n$ is typically determined by budget constraints.\footnote{Alternatively, researchers may select $\bar n$ through a minimum detectable effect analysis which is standard in the analysis of experiments \citep{GerberGreen2012}.} The lower bound $\underline n$ is often nonbinding because larger samples generally reduce variance; nonetheless, we recommend verifying that the realized main-experiment size is sufficiently large, such as $n \ge 2\bar n/3$ as a simple check.
\end{itemize}

\subsection{Numerical Studies} \label{sec:numerics}

Lastly, we present our numerical studies. In simulations, we assume
$T_i = |\mathcal{N}_i|$ and $g_i\!\big(D_{\mathcal{N}_i}\big) = \sum_{k \in \mathcal{N}_i} D_k$.
Define $S_i := \sum_{k \in \mathcal{N}_i} D_k$ and $G_i := S_i/|\mathcal{N}_i|$.

\paragraph{Simulation model.}
We specify the conditional variance and covariance functions
\begin{equation} 
\small
\begin{aligned}
\sigma^2(d,s,l) \;=\; \mu + \beta_1 d + \frac{s\,\beta_2}{\max\{l,1\}}, 
\qquad
\eta(d,s,l,d',s',l') \;=\; \alpha \sqrt{\sigma^2(d,s,l)\,\sigma^2(d',s',l')},
\end{aligned}
\end{equation}
where $s$ denotes the number of treated neighbors and $l$ denotes the number of neighbors.
The covariance assumes a constant correlation $\alpha$ as in \citet{baird2018optimal}.
We set $\alpha=0.1$ and $\mu=0.5$, and consider $(\beta_1,\beta_2)\in\{(0,0),(0.5,0.5),(0.5,1)\}$, which we label, respectively, ``homoskedastic,'' ``small heteroskedasticity,'' and ``large heteroskedasticity.'' Results for more parameters are reported in Appendix \ref{sec:aa2}.

Following \citet{eckles2017design}, we consider outcome models of the form
\begin{equation}\label{eqn:linear1}
Y_i \;=\; \gamma_1 D_i \;+\; \gamma_2\, G_i \;+\; \varepsilon_i,
\qquad
\varepsilon_i \,\big|\, (D_i,D_{\mathcal{N}_i},T_i) \sim \mathcal N\!\big(0,\ \sigma^2(D_i,G_i,T_i)\big).
\end{equation}
We use $(\gamma_1,\gamma_2)=(0.5,1)$. These coefficients affect the mean but not the conditional variance of the estimators given $(\mathbf R,\mathbf D^R)$ and thus do not influence variance comparisons across designs.

\paragraph{Network design.}
Our main simulations use the village friendship networks of \citet{cai2015social}. We build two undirected adjacency matrices:
(i) a ``weak'' network with an edge if either member names the other as a friend, and
(ii) a ``strong'' network with an edge only if both name each other.
The weak network is denser, while the strong network is sparser. This allows us to study how network density affects the performance of the algorithm.
We consider as our population the first five villages, corresponding to $N=832$ nodes in total, and impose that no more than half of the individuals are in the main experiment ($\bar n=416$), with no binding constraints on $\underline{n}$. 

We also consider simulated graphs with $N=800$ nodes and a maximum number of participants $\bar n=400$. In particular, we study two models:
(i) \textit{Erd\H{o}s–R\'enyi (ER):} $A_{ij}\stackrel{\text{iid}}{\sim}\mathrm{Bernoulli}(p)$ for $i<j$, symmetrized with zero diagonal, using $p=2/N$;
(ii) \textit{Barab\'asi–Albert (BA):} start from an Erd\H{o}s–R\'enyi graph on $N_0=N/5$ nodes with $p=2/N$. Then, for $t=N_0+1,\ldots,N$, add one node and attach it to $m=2$ existing nodes with probability proportional to their degrees.


\vspace{-3mm} 

\subsubsection{Estimation details}

We implement Algorithms~\ref{alg:pilot} and~\ref{alg:coefficients2} using the defaults in Section~\ref{sec:guide_practice}. We solve the MIQP in~\eqref{eqn:opt_pilot} exactly with $\bar m=70$ and $\delta=30$. We estimate the conditional mean as
\[
\small
\hat m_p(d,s,l) \;=\; \hat \gamma_0 + \hat \gamma_1 d + \hat \gamma_2 \,\frac{s}{\max\{l,1\}},
\]
where $(\hat{\gamma}_0,\hat{\gamma}_1,\hat{\gamma}_2)$ are the OLS coefficients from regressing $Y_i$ on $\big[1,\ D_i,\ G_i\big]$ using the pilot data $\{(Y_i,D_i,T_i,D_{\mathcal{N}_i}): P_i=1\}$. We then form residuals
$
\hat\varepsilon_i \;:=\; Y_i - \hat m_p(D_i, S_i, T_i),
S_i=\textstyle\sum_{k\in\mathcal N_i} D_k. 
$
For the variance, we fit the nonnegative least-squares model in Algorithm~\ref{alg:pilot_varcov} with
$
W_i \;=\; \big[\,1,\ D_i,\ G_i\,\big],
$
which is a correctly specified model for the variance. As a more flexible alternative, we also consider a fourth-degree polynomial in $(D_i,G_i)$, described below.

For the covariance, we impose a constant correlation $\alpha\in[0,0.3]$ and estimate it by regressing $\hat\varepsilon_i \hat\varepsilon_j$ on $\widehat\sigma_p(T_i,D_i,G_i)\,\widehat\sigma_p(T_j,D_j,G_j)$ over pilot units $(i,j)$ with $j\in\mathcal{N}_i$, as in Algorithm~\ref{alg:pilot_varcov}. This constraint reflects prior (approximate) knowledge of small, positive correlations between neighbors.

Given $(\widehat\sigma_p^2,\widehat\eta_p)$, we solve the design problem~\eqref{eqn:design1} over $(\mathbf R,\mathbf D^R)$ using a nonlinear mixed-integer solver (Ant Colony Optimization; \citealp{dorigo2005ant}), with a time limit of 9{,}000 seconds.

\paragraph{Variants of our method.}
In addition to our main procedure, we consider three variants.
The first variant (\emph{ELI with Randomization}) follows Algorithm~\ref{alg:coefficients3b}: we independently flip each optimal assignment with probability $\iota\in\{0.05,0.10\}$ for $R_i=1$; for simplicity, we perform a single randomization draw (no re-randomization).  
The second variant (\emph{ELI-Unobs}) is described in Appendix~\ref{sec:model_assisted} and allows for partial network information: it selects a pilot of 70 units from the sixth village only and assumes the pre-experiment network is observed for the first 200 units in the main village. For missing edges, we use a simple Erd\H{o}s–R\'enyi imputation: draw $p\sim\mathrm{Uniform}(0,1)$ once per replication and, conditional on $p$, draw missing $A_{ij}$ i.i.d.\ $\mathrm{Bernoulli}(p)$. We alternate (a) Monte Carlo evaluation of $\widehat V_{\widehat\sigma_p,\widehat\eta_p}$ over imputed edges and (b) the design optimization over $(\mathbf R,\mathbf D^R)$. Full details are in Appendix~\ref{sec:model_assisted}. The third variant is the same but instead of using a linear model for the variance (with positivity constraints) it uses a more flexible model with a polynomial transformation of $(D_i,G_i)$ of degree four and the same positivity constraint (defined as \emph{ELI flexible}).

 \subsubsection{Competing methods}\label{sec:competitors}

We benchmark our procedure against a set of alternatives with either $n=400$ or $n_+=470$ (the latter equals $400$ plus the $70$ units used in our pilot). Competitors under the augmented budget are labeled with a ``$+$'' suffix. We consider:
\begin{itemize}
\item \textbf{Graph clustering (3-$\epsilon$ net).}
We implement the 3-$\epsilon$ net clustering of \citet{ugander2013graph} as follows: start with all nodes uncovered; iteratively select an uncovered node uniformly at random as a center; mark first and second-degree neighbors of centers as ineligible to be centers. Next, assign each node to the cluster of its closest center, breaking ties at random. For each cluster, draw a single treatment from $\mathrm{Bernoulli}(1/2)$ and assign that treatment to all units in the cluster. (We let $n\in\{400,470\}$.)

\item \textbf{Saturation designs on clustered graphs.}
Because classical saturation designs require a partition, we first form clusters using the 3-$\epsilon$ net above, then apply the two-stage saturation framework of \citet{baird2018optimal}. We implement several versions using the authors’ software. Within a cluster assigned probability $p$, units are treated independently with probability $p$. We consider:
\begin{enumerate}
\item \textbf{Saturation1:} the cluster-level probability $p$ is drawn independently and uniformly from $[0,1]$ across clusters; $n\in\{400,470\}$.
\item \textbf{Saturation2:} choose the set of cluster-level probabilities to minimize the sum of asymptotic standard errors of the direct and spillover effects (ITT and SNDT in \citet{baird2018optimal}), using their formulas. For comparability with an oracle benchmark, we plug in the true intra-cluster correlation $\alpha$ and assume homoskedastic idiosyncratic variance in those formulas (as assumed in \citet{baird2018optimal}).
\item \textbf{Saturation3:} as in Saturation2, but the objective minimizes the sum of standard errors for the direct, spillover, and slope effects as defined in \citet{baird2018optimal}.
\end{enumerate}

\item \textbf{Random assignment.} Sample $n_+=470$ participants uniformly at random and assign treatment i.i.d.\ $\mathrm{Bernoulli}(1/2)$.
\end{itemize}

\subsubsection{Results}

Table~\ref{tab:realworld} summarizes our main results on the real-world network by reporting the sampling variance of the estimator(s) across columns indexed by $(\beta_1,\beta_2)$. The left panel uses the network with strong ties; the right panel uses weak ties. The top panel targets only the global (overall) treatment effect, whereas the bottom panel estimates both the direct and the spillover effects (two estimands/two estimators, as in Appendix~\ref{app:multiple}). Across all specifications, our proposed method (ELI) uniformly attains the lowest variance relative to all competitors. The gains are larger when $(\beta_1,\beta_2)$ increase—that is, under stronger heteroskedasticity. Consistent with design-based identification, Figure~\ref{fig:bias} in the Appendix shows bias is zero.

For observed networks, the two ELI variants—(i) adding randomization and (ii) using a more flexible estimator—perform similarly to the baseline ELI, suggesting these variants have comparable performance to ELI even under greater flexibility. With a partially observed network, the only valid benchmark we consider is random allocation; in this case, ELI again dominates uniformly.

Figure~\ref{fig:best_fig} complements these findings. In the left panel (real-world networks), we plot the percentage reduction in sample size that would be required for the best-performing alternative to match the variance of our main ELI specification for the overall effect. For the “unobserved network’’ case, the comparison is between ELI with a partially observed network and random allocation only. The implied savings are up to 40 percentage points.  

In the right panel (simulated networks), we plot the log-variance of ELI against the competitor with the lowest median variance, which randomizes using the same total number of units as ELI (participants plus pilot). Under heteroskedasticity, ELI uniformly outperforms, with larger advantages as heteroskedasticity intensifies. Under homoskedasticity ($( \beta_1,\beta_2 )=(0,0)$), ELI still dominates except in a single case where graph clustering with 70 additional participants  slightly outperforms ELI.\footnote{Because the simulated-network analysis does not include a separate cluster as in the real-world setting, we do not report partially observed–network results for simulations.} Additional results are reported in Appendix~\ref{sec:aa2}.

\begin{figure}[!ht]
\centering 
\includegraphics[scale=0.5]{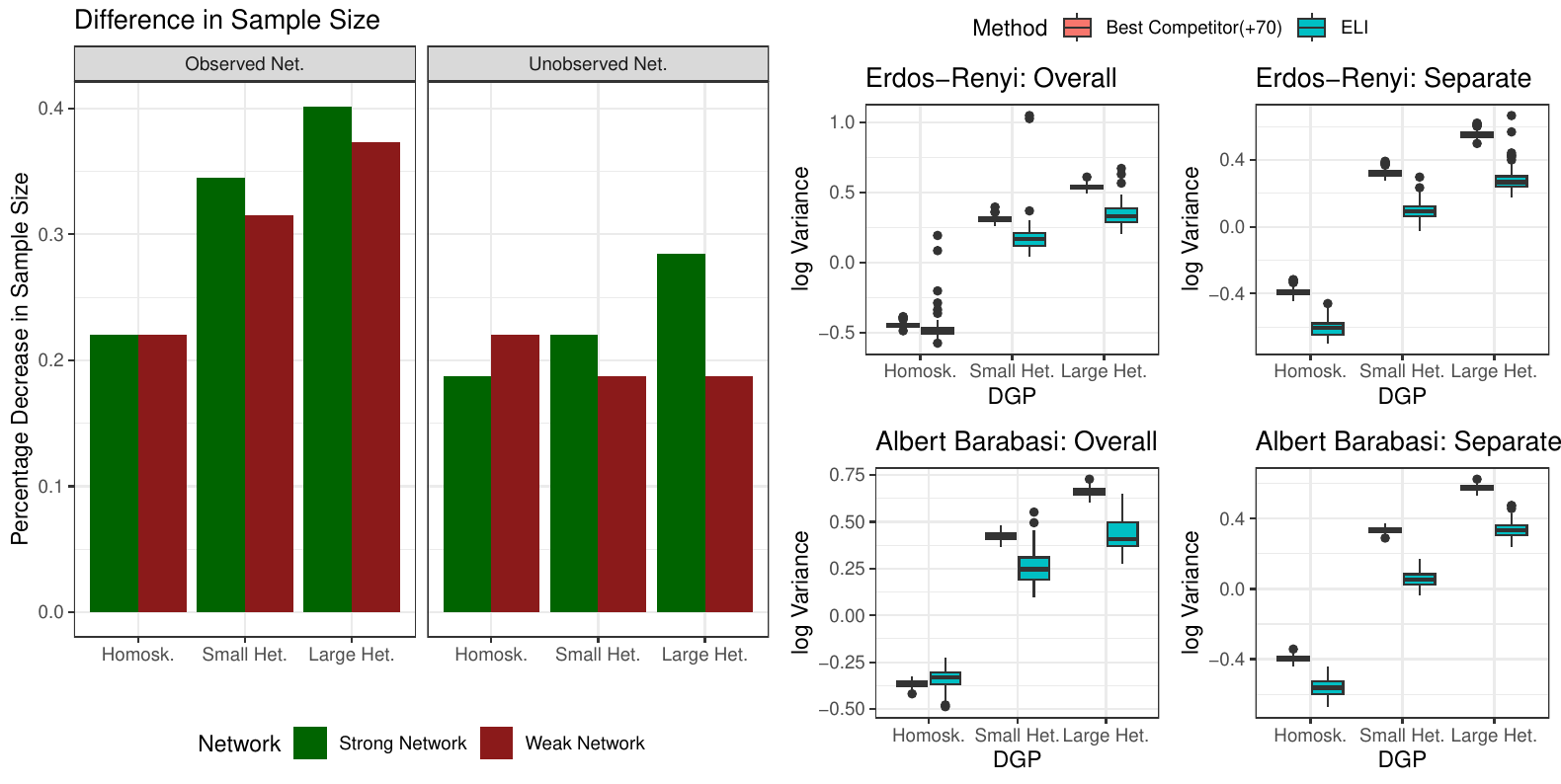}
\caption{Left panel: Percentage reduction in the total number of units required for the best-performing competitor to match the variance of ELI for the overall treatment effect in simulations using the real-world network. The “Unobserved network’’ case compares ELI with a partially observed network to random allocation only. “Total units’’ counts both participants and pilot units. Right panel: Log variance of ELI (blue) versus the competitor with the lowest median variance, which randomizes using the same total number of units as ELI (participants + pilot). 
}
\label{fig:best_fig} 
\end{figure}

\begin{table} [!htbp] \centering 
  \caption{
  The panels at the top reports the 
variance for estimating the overall effect on the network from \cite{cai2015social}, using the first five villages as the population of interest ($N=832$). The panel at the bottom reports the worst-case variance when estimating separately the direct and spillover effects. Columns correspond to alternative designs evaluated at different values of $(\beta_1,\beta_2)$. “ELI’’ denotes our baseline specification: $416$ participants in the main experiment and a pilot of $70$ units. We also report two observed-network variants (ELI with an added randomization layer and ELI with a more flexible variance estimator). The partially observed-network specification (“ELI, partially observed’’) uses only the first sub-block (the first $200$ observations) in the main experiment, with a pilot of $70$ units drawn from the sixth village. Methods marked with a “$+$’’ run the main experiment on $416\!+\!70$ units, whereas methods without “$+$’’ use $416$ participants. All competitors, except Random All$+$, exploit full knowledge of the network structure.  
 } \label{tab:realworld}
  
\scalebox{0.7}{\begin{tabular}{@{\extracolsep{5pt}} cccc|ccc} 
\\[-1.8ex]\hline 
\hline \\[-1.8ex] 
& Strong & &  & Weak &  &  \\ 
\textbf{Overall Effect} & (0,0) & (0.5,0.5)  & (0.5,1) & (0,0) & (0.5,0.5) & (0.5,1) \\ 
\hline \\[-1.8ex] 
ELI & $0.551$ & $1.134$ & $1.404$ & $0.769$ & $1.442$ & $1.665$ \\ 
ELI + Randomization 5\% & $0.601$ & $1.239$ & $1.538$ & $0.898$ & $1.710$ & $1.993$ \\ 
ELI + Randomization 10\% & $0.657$ & $1.355$ & $1.684$ & $1.057$ & $2.039$ & $2.397$ \\ 
ELI Flexible & $0.582$ & $1.082$ & $1.486$ & $0.769$ & $1.448$ & $1.966$ \\ 
ELI Flexible + Randomization 5\% & $0.626$ & $1.188$ & $1.787$ & $0.893$ & $1.730$ & $2.313$ \\ 
ELI Flexible + Randomization 10\% & $0.680$ & $1.304$ & $1.912$ & $1.043$ & $2.028$ & $2.713$ \\ 
ELI - Unobserved Net & $0.914$ & $1.829$ & $2.183$ & $2.018$ & $4.139$ & $ 5.067$ \\
ELI Unobs + Randomization 5\% & $0.951$ & $1.904$ & $2.286$ & $2.128$ & $4.362$ & $5.357$ \\ 
ELI Unobs + Randomization 10\% & $0.987$ & $1.981$ & $2.391$ & $2.238$ & $4.580$ & $5.636$ \\ 
Random All+ & $1.107$ & $2.249$ & $2.876$ & $2.430$ & $4.827$ & $6.127$ \\ 
Graph Clustering+ & $0.694$ & $1.591$ & $2.038$ & $0.874$ & $1.830$ & $2.345$ \\ 
Saturation1+ & $0.913$ & $1.985$ & $2.513$ & $1.523$ & $3.143$ & $3.866$ \\ 
Graph Clustering & $0.793$ & $1.847$ & $2.420$ & $0.989$ & $2.104$ & $2.623$ \\ 
Saturation1 & $1.059$ & $2.259$ & $2.940$ & $1.736$ & $3.603$ & $4.482$ \\ 
Saturation2+ & $0.719$ & $1.669$ & $2.104$ & $0.944$ & $1.973$ & $2.418$ \\ 
Saturation3+ & $0.931$ & $2.171$ & $2.772$ & $1.700$ & $3.844$ & $4.829$ \\ 
\hline \\[-1.8ex] 
\end{tabular} }
\scalebox{0.7}{\begin{tabular}{@{\extracolsep{5pt}} cccc|ccc} 
\\[-1.8ex]\hline 
\hline \\[-1.8ex]
& Strong & &  & Weak &  &  \\ 
\textbf{Treatment and Spillovers} & (0,0)  & (0.5,0.5) & (0.5,1) & (0,0) & (0.5,0.5) & (0.5,1)  \\ 
\hline \\[-1.8ex] 
ELI & $0.495$ & $1.028$ & $1.299$ & $0.790$ & $1.525$ & $1.825$ \\ 
ELI + Randomization 5\% & $0.510$ & $1.059$ & $1.350$ & $0.904$ & $1.769$ & $2.125$ \\ 
ELI + Randomization 10\% & $0.525$ & $1.099$ & $1.407$ & $1.035$ & $2.050$ & $2.472$ \\ 
ELI Flexible & $0.496$ & $1.002$ & $1.261$ & $0.712$ & $1.507$ & $2.263$ \\ 
ELI Flexible + Randomization 5\% & $0.511$ & $1.090$ & $1.400$ & $0.816$ & $1.805$ & $2.632$ \\ 
ELI Flexible + Randomization 10\% & $0.527$ & $1.124$ & $1.446$ & $0.937$ & $2.051$ & $2.956$ \\ 
ELI - Unobserved Net & $0.596$ & $1.286$ & $1.639$ & $1.613$ & $3.133$ & $4.165$ \\ 
ELI Unobs + Randomization 5\% & $0.599$ & $1.316$ & $1.676$ & $1.685$ & $3.281$ & $4.315$ \\
ELI Unobs + Randomization 10\% & $0.607$ & $1.348$ & $1.718$ & $1.755$ & $3.426$ & $4.464$ \\ 
Random All+ & $0.641$ & $1.431$ & $1.882$ & $1.813$ & $3.580$ & $4.477$ \\ 
Graph Clustering+ & $0.864$ & $2.147$ & $2.600$ & $1.838$ & $3.528$ & $4.431$ \\ 
Saturation1+ & $0.652$ & $1.500$ & $1.942$ & $1.403$ & $2.807$ & $3.569$ \\ 
Graph Clustering & $0.999$ & $2.491$ & $3.001$ & $2.165$ & $4.022$ & $5.302$ \\ 
Saturation1 & $0.760$ & $1.755$ & $2.286$ & $1.654$ & $3.283$ & $4.068$ \\ 
Saturation2+ & $0.773$ & $1.900$ & $2.371$ & $1.516$ & $2.986$ & $3.724$ \\ 
Saturation3+ & $0.801$ & $1.910$ & $2.449$ & $2.231$ & $4.202$ & $5.155$ \\  
\hline \\[-1.8ex] 

\end{tabular} }
\end{table}

\section{Conclusions} \label{sec:conclusions}

This paper introduced a method for designing experiments under network interference. We proposed a two-wave design that selects participation indicators and treatment assignments to minimize the variance of a pre-specified linear estimator, and we provided the first statistical framework for such two-wave experiments with interference, including regret guarantees.

Our main analysis considers settings in which the complete network is observed. In the Appendix and in simulations, we show how the framework extends to partially observed networks. Our numerical findings suggest that imputing missing edges can reduce the estimator's variance. A comprehensive theoretical analysis of partially observed networks, including model selection for the imputation step, remains an important direction for future work.

We focused on local interference. Future research should study designs under more general interaction structures in which interference propagates across larger portions of the network. Understanding how network topology and alternative exposure mappings affect the performance of the proposed design mechanisms is another open question.

\bibliography{my_bib2}
\bibliographystyle{chicago}

\appendix

\numberwithin{equation}{section}
\makeatletter 
\newcommand{\section@cntformat}{Appendix \thesection:\ }
\makeatother
 
 \numberwithin{figure}{section}
\numberwithin{algorithm}{section}
 \numberwithin{table}{section}
\makeatletter 
 
 \onehalfspacing

\section{Extensions: detailed description}  \label{sec:extension_app}

\subsection{ Randomized treatments} \label{app:randomization}

Algorithm~\ref{alg:coefficients3b} formalizes a randomized variant of our design. As discussed in Section~\ref{sec:randomization}, the algorithm depends on two user-chosen parameters, $\iota$ and $\zeta$: larger $\iota$ increases randomization away from the optimizer, while larger $\zeta$ allows more variance inflation relative to the optimum. In Step~1, we compute the optimal participants and assignments using pilot information (as in Algorithm~\ref{alg:coefficients2}). In Step~2, we apply an accept/reject re-randomization: for each participant $i$ with $R_i=1$, we flip the optimal assignment $\check D_i^R$ with probability $\iota$, and accept the proposed assignment vector only if its plug-in variance exceeds the optimal plug-in variance by at most $\zeta/\bar n$; otherwise we re-randomize until acceptance.

\begin{thm}\label{thm:regretb}
Under Assumptions~\ref{ass:model}, \ref{ass:modelb}, \ref{ass:sparsity}, \ref{ass:convergence}, \ref{ass:moment}, and \ref{ass:weights_stable}, suppose $\bar n/\underline n=\alpha\in(1,C)$ for some universal constant $C<\infty$, and $\underline n \ge \mathcal{N}_{\max}\,\bar m/(\alpha-1)$. Then
\[
\bar n\!\left[\mathbb V\!\big(\widehat\Gamma \,\big|\, \mathbf R,\mathbf D^R,\mathbf A,\mathbf T\big) - \mathbb V_{\bar n}^\star\right] 
\;\le\; \mathcal O_p\!\Big(\mathcal{N}_{\max}\,\bar m^{-\xi} \;+\; \frac{\mathcal{N}_{\max}^2\,\bar m}{\underline n}\Big) \;+\; \zeta.
\]
\end{thm}

Relative to the non-randomized design, Theorem~\ref{thm:regretb} adds the $\zeta$ term reflecting the variance cap in Step~2: researchers deliberately trade a small amount of model-based precision for additional randomization.

Design-based inference for hypotheses of independent interest (e.g., sharp nulls) is obtained by simulating the same scheme in Step~2 conditional on the optimized pair $(\mathbf R,\check{\mathbf D}^R)$ from Step~1. This conditioning is valid because $(\mathbf R,\check{\mathbf D}^R)$ are functions only of pilot outcomes (and network information), not of main-experiment outcomes.

In practice, although $\zeta$ is central to the theoretical bound, we find that choosing $\iota\in\{0.10,0.05\}$ and performing a single randomization draw (i.e., without re-randomization) typically leads to only a small variance increase, while enabling design-based analyses.

\begin{algorithm}[!ht]
\caption{Main experiment with re-randomization}\label{alg:coefficients3b}
\begin{algorithmic}[1]
\Require $\mathbf A$, $\mathbf T$, pilot output $\big(\mathbf P,\widehat\sigma_p^{2}(\cdot),\widehat\eta_p(\cdot)\big)$, bounds $\underline n,\bar n$, tolerance $\zeta\ge0$, re-randomization probability $\iota\in(0,1/2)$. Let $\mathcal J$ be defined as in~\eqref{eqn:pilot_identity}.
\State \textit{Optimize (as in Algorithm~\ref{alg:coefficients2}):}
\begin{equation}\label{eqn:design1b}
\small
\begin{aligned}
(\check{\mathbf D}^R,\mathbf R)\in\;
&\operatorname*{arg\,min}_{\mathbf r\in\{0,1\}^N,\;\mathbf d^{\,r}\in\{0,1\}^N}
\ \widehat V_{\widehat\sigma_p,\widehat\eta_p}\big(\mathbf r,\mathbf d^{\,r}\big)\\
\text{s.t.}\quad
&\mathbf 1^\top \mathbf r\in[\underline n,\bar n],\qquad
r_j=0\ \ \forall j\in\mathcal J.
\end{aligned}
\end{equation}
\State \textit{Re-randomize (variance-capped):}
\begin{itemize}
\item If $\zeta=0$, set $\mathbf D^R=\check{\mathbf D}^R$.
\item If $\zeta>0$, repeat: for each $i$ with $R_i=1$, draw $B_i\stackrel{\mathrm{iid}}{\sim}\mathrm{Bernoulli}(\iota)$ and set
\[
D_i^R \;=\; \check D_i^R(1-B_i) + (1-\check D_i^R)B_i.
\]
Accept the proposed $\mathbf D^R$ if
\[
\widehat V_{\widehat\sigma_p,\widehat\eta_p}(\mathbf R,\mathbf D^R)
-\widehat V_{\widehat\sigma_p,\widehat\eta_p}(\mathbf R,\check{\mathbf D}^R) \;\le\; \zeta/\bar n;
\]
otherwise, repeat.
\end{itemize}
\State \textit{Collect data:}\quad $\{(Y_i,D_i,T_i,D_{\mathcal{N}_i},\mathcal{N}_i): R_i=1\}$.
\State \textit{Estimate the outcome regression:} obtain a consistent (possibly non-parametric) estimator $\hat m$ of $m$ and define
$\hat m_i := \hat m\big(D_i,g_i(D_{\mathcal{N}_i}),T_i\big)$.
\State \textit{Estimate $\widehat\Gamma$ and its variance:} construct $\widehat\Gamma$ as in~\eqref{eqn:estimator}, and compute
\begin{equation}\label{eqn:final_vvb}
\widehat V = \frac{1}{n^2}\sum_{i=1}^N R_i\,w(i;\mathbf R,\mathbf D^R)\big(Y_i-\hat m_i\big)
\sum_{j\in\mathcal{N}_i\cup\{i\}} R_j\,w(j;\mathbf R,\mathbf D^R)\big(Y_j-\hat m_j\big).
\end{equation}
\State \textbf{Return:}\quad $\widehat\Gamma,\ \widehat V$.
\end{algorithmic}
\end{algorithm}

 \subsection{Multiple estimands and estimators} \label{app:multiple}

 Here, following Section \ref{sec:multiple_estimators}, we provide details on the algorithm in the presence of multiple estimands and estimators (with a single estimator for each estimand). 
The pilot is constructed exactly as in Algorithm~\ref{alg:pilot}. 
Given the pilot, we form plug-in variance estimates for each $w\in\mathcal W$:
\begin{equation}\label{eqn:var2b}
\small
\begin{aligned}
\widehat V^{\,w}_{\widehat\sigma_p,\widehat\eta_p}(\mathbf R,\mathbf D^R)
&= \frac{1}{n^2}\sum_{i:R_i=1} w^2\!\big(i;\mathbf R,\mathbf D^R\big)\,
\widehat\sigma_p^{2}\!\Big(T_i, D_i, g_i(D_{\mathcal{N}_i})\Big) \\
&\quad + \frac{1}{n^2}\sum_{i:R_i=1}\sum_{j\in\mathcal{N}_i} R_j\,
w\!\big(i;\mathbf R,\mathbf D^R\big)\, w\!\big(j;\mathbf R,\mathbf D^R\big)\,
\widehat\eta_p\!\Big(T_i, D_i, g_i(D_{\mathcal{N}_i}),\; T_j, D_j, g_j(D_{\mathcal{N}_j})\Big).
\end{aligned}
\end{equation}

We then mirror Algorithm~\ref{alg:coefficients2} but minimize the worst-case variance across the finite collection $\mathcal W$:
\[
(\check{\mathbf D}^R,\mathbf R)\in 
\operatorname*{arg\,min}_{\mathbf r\in\{0,1\}^N,\;\mathbf d^{\,r}\in\{0,1\}^N}
\ \max_{w\in\mathcal W}\;
\widehat V^{\,w}_{\widehat\sigma_p,\widehat\eta_p}(\mathbf r,\mathbf d^{\,r})
\quad\text{s.t.}\quad
\mathbf 1^\top\mathbf r\in[\underline n,\bar n],\ \ r_j=0\ \forall j\in\mathcal J.
\]
(Here $\mathcal J$ is as in~\eqref{eqn:pilot_identity} and enforces unconfoundedness).

Finally, note that, as in Section~\ref{sec:randomization}, one may introduce a re-randomization step.

For finite $E$, standard central limit theorems yield the marginal asymptotic normality of each $\widehat\Gamma(w)$ (Theorem~\ref{thm:asym_t} applies componentwise). 
Therefore, it suffices to focus on the regret guarantees. Define the oracle minimax variance
\[
\mathbb V_{\bar n}^{\star}(E)
:= \min_{\mathbf R,\mathbf D^R:\,\mathbf 1^\top\mathbf R=\bar n}\ \max_{w\in\mathcal W}\ 
\mathbb V\!\big(\widehat\Gamma(w)\,\big|\,\mathbf R,\mathbf D^R,\mathbf A,\mathbf T\big).
\]

\begin{thm}\label{thm:regret2}
Let Assumptions~\ref{ass:model}, \ref{ass:modelb}, \ref{ass:sparsity}, \ref{ass:convergence}, and \ref{ass:moment} hold, and suppose Assumption~\ref{ass:weights_stable} holds for all $w\in\mathcal W$ with $E < \infty$. Consider a design as in Algorithm \ref{alg:3}. 
If $\bar n/\underline n=\alpha\in(1,C)$ for some universal $C<\infty$ and $\underline n\ge \mathcal{N}_{\max}\,\bar m/(\alpha-1)$, then
\[
\bar n\!\left[\max_{w\in\mathcal W}\mathbb V\!\big(\widehat\Gamma(w)\,\big|\,\mathbf R,\mathbf D^R,\mathbf A,\mathbf T\big)
-\mathbb V_{\bar n}^{\star}(E)\right]
\;\le\; \mathcal O_p\!\Big(\mathcal{N}_{\max}\,\bar m^{-\xi} \;+\; \frac{\mathcal{N}_{\max}^2\,\bar m}{\underline n}\Big) \;+\; \zeta.
\]
\end{thm}

\noindent The proof is in Appendix~\ref{proof:thm:regretb}.

\begin{algorithm}[!ht]
\caption{Main experiment with multiple estimands and estimators}\label{alg:3}
\begin{algorithmic}[1]
\Require $\mathbf A$, $\mathbf T$, pilot output $\big(\mathbf P,\widehat\sigma_p^{2}(\cdot),\widehat\eta_p(\cdot)\big)$, bounds $\underline n,\bar n$, weights $\mathcal W=\{w^1,\ldots,w^E\}$; optionally tolerance $\zeta\ge0$ and re-randomization probability $\iota\in(0,1/2)$. Let $\mathcal J$ be as in~\eqref{eqn:pilot_identity}.
\State \textit{Minimax design:}
\[
(\check{\mathbf D}^R,\mathbf R)\in 
\operatorname*{arg\,min}_{\mathbf r,\mathbf d^{\,r}}
\ \max_{w\in\mathcal W}\ \widehat V^{\,w}_{\widehat\sigma_p,\widehat\eta_p}(\mathbf r,\mathbf d^{\,r})
\ \ \text{s.t.}\ \
\mathbf 1^\top\mathbf r\in[\underline n,\bar n],\ \ r_j=0\ \forall j\in\mathcal J.
\]
\State \textit{(Optional) Re-randomization:}
\begin{itemize}
\item If $\zeta=0$, set $\mathbf D^R=\check{\mathbf D}^R$.
\item If $\zeta>0$, repeat: for each $i$ with $R_i=1$, draw $B_i\sim\mathrm{Bernoulli}(\iota)$ and set 
\[
D_i^R=\check D_i^R(1-B_i)+(1-\check D_i^R)B_i,\quad D_i^R=0\ \text{if }R_i=0,
\]
accepting the proposal only if 
\[
\max_{w\in\mathcal W}\!\left\{\widehat V^{\,w}_{\widehat\sigma_p,\widehat\eta_p}(\mathbf R,\mathbf D^R)
-\widehat V^{\,w}_{\widehat\sigma_p,\widehat\eta_p}(\mathbf R,\check{\mathbf D}^R)\right\}\le \zeta/\bar n .
\]
\end{itemize}
\State \textit{Collect data:}\quad $\{(Y_i,D_i,T_i,D_{\mathcal{N}_i},\mathcal{N}_i): R_i=1\}$.
\State \textit{Estimate effects and variances:} for each $w\in\mathcal W$, compute $\widehat\Gamma(w)$ via~\eqref{eqn:estimator2} and
\[
\widehat V^{\,w}
=
\frac{1}{n^2}\sum_{i=1}^N R_i\, w\!\big(i;\mathbf R,\mathbf D^R\big)\big(Y_i-\hat m_i\big)
\sum_{j\in\mathcal{N}_i\cup\{i\}} R_j\, w\!\big(j;\mathbf R,\mathbf D^R\big)\big(Y_j-\hat m_j\big),
\]
where $\hat m_i:=\hat m\big(D_i,g_i(D_{\mathcal{N}_i}),T_i\big)$ for $\hat m$ being a consistent estimator for $m$.
\State \textbf{Return:}\quad $\{\widehat\Gamma(w),\widehat V^{\,w}\}_{w\in\mathcal W}$.
\end{algorithmic}
\end{algorithm}

\subsection{Higher-Order Dependence}\label{sec:higher_order_dependence}

In this section, we relax the local-dependence assumption and consider the general case in which unobservables exhibit $M$-dependence. Let $\mathcal{N}_i^u$ denote the set of individuals at graph distance exactly $u$ from $i$ (i.e., connected to $i$ by $u$ edges), and write $\mathcal{N}_i^M$ for the $M$-th degree neighbors.

We generalize Assumption~\ref{ass:model} with the following conditions.

\begin{ass}[Model under higher-order dependence]\label{ass:higherorder}
Assume that for all $i \in \{1,\ldots,N\}$, for some finite $M<\infty$,
\[
\begin{aligned}
&\varepsilon_i \ \perp\!\!\!\perp\ \{\varepsilon_j\}_{\,j \notin \,\cup_{u=1}^M \mathcal{N}_i^u \,\cup\,\{i\}} \ \Big| \ \mathbf{A},\mathbf{T}, \\[0.25em]
&(\varepsilon_i,\varepsilon_j)\ \stackrel{d}{=}\ (\varepsilon_{i'},\varepsilon_{j'}) \ \Big| \ \mathbf{A},\mathbf{T}
\quad \text{for all } (i,j,i',j') \text{ such that } i\in\mathcal{N}_j^u,\ i'\in\mathcal{N}_{j'}^u,\ T_i=T_{i'},\ T_j=T_{j'}, \\[0.25em]
&\mathcal{N}_{\max} < c, \forall u \in \{1, \cdots, M\}
\end{aligned}
\]
for a universal constant $c<\infty$.
\end{ass}

Assumption~\ref{ass:higherorder} states: (i) unobservables are independent whenever units are separated by more than $M$ edges; (ii) the joint distribution of pairs of neighboring unobservables is the same whenever the corresponding covariates match; and (iii) the maximum degree is uniformly bounded (this can be relaxed).\footnote{Inspecting the derivations in Appendix~\ref{sec:asymp_app} shows that the bounded-degree condition can be replaced by requiring that the maximum degree among sampled units and their neighbors up to order $M$ grows more slowly than $N^{1/4}$.}

We next impose the experimental restriction. Define
\[
\tilde{\mathcal{H}} \;=\; \{1,\ldots,N\} \setminus \bigcup_{i:\,P_i=1}\bigg(\,\bigcup_{u=1}^M \mathcal{N}_i^u \ \cup\ \{i\}\bigg),
\]
the set of units remaining after excluding pilot units and all neighbors of pilot units up to degree $M$.

\begin{ass}[Experimental restriction]\label{ass:covb}
Let the following hold:
\begin{itemize}
\item[(A)] $\{\varepsilon_i\}_{i\in \tilde{\mathcal{H}}}\ \perp\!\!\!\perp\ (\mathbf{R},\mathbf{D}^R)\ \big|\ \mathbf{A},\mathbf{T},\mathbf{P}$;
\item[(B)] $\{\varepsilon_i\}_{i=1}^N\ \perp\!\!\!\perp\ \mathbf{P}\ \big|\ \mathbf{A},\mathbf{T}$;
\item[(C)] $R_j=0$ for all $j \in \bigcup_{i:\,P_i=1}\big(\,\cup_{u=1}^M \mathcal{N}_i^u \ \cup\ \{i\}\big)$.
\end{itemize}
\end{ass}

The following theorem guarantees unbiasedness of the estimator for any design satisfying these restrictions.

\begin{thm}\label{thm:1C}
Under Assumptions~\ref{ass:model}, \ref{ass:higherorder}, and~\ref{ass:covb},
\[
\mathbb{E}\!\Big[\widehat{\Gamma}\ \Big|\ \mathbf{P},\mathbf{A},\mathbf{R},\mathbf{D}^R,\mathbf{T}\Big]
\;=\; \tau_{\mathbf{A},\mathbf{T}}(\mathbf{R},\mathbf{D}^R).
\]
\end{thm}

See Appendix \ref{proof:thm:1C} for the proof. 
The design of the main experiment minimizes the pilot-based plug-in variance subject to Assumption~\ref{ass:covb}, which requires excluding from the main experiment all neighbors of pilot units up to degree $M$.

The pilot procedure follows Algorithm~\ref{alg:pilot} in spirit, with the additional requirement that some units have neighbors within the pilot at degrees $u=1,\ldots,M$. Concretely, add the constraints
\[
\sum_{i=1}^N p_i \sum_{j\in \mathcal{N}_i^u} p_j \ \ge\ \delta
\quad\text{for each } u\in\{1,\ldots,M\}
\]
to~\eqref{eqn:opt_pilot}, so that covariances at distances $1$ through $M$ are identified. With higher-order dependence, the variance decomposes as
\begin{equation}\label{eqn:var2}
\begin{aligned}
n\,\mathbb{V}\!\big(\widehat{\Gamma}\ \big|\ \mathbf{A},\mathbf{T},\mathbf{R},\mathbf{D}^R\big)
&= \frac{1}{n}\sum_{i:\,R_i=1} w(i;\mathbf{R},\mathbf{D}^R)^2\ \mathrm{Var}\!\big(Y_i \mid \mathbf{A},\mathbf{R},\mathbf{D}^R,\mathbf{T}\big)\\
&\quad + \sum_{u=1}^{M} \frac{1}{n}\sum_{i:\,R_i=1}\sum_{j\in \mathcal{N}_i^u} R_j\, w(i;\mathbf{R},\mathbf{D}^R)\, w(j;\mathbf{R},\mathbf{D}^R)\,
\mathrm{Cov}\!\big(Y_i,Y_j \mid \mathbf{A},\mathbf{R},\mathbf{D}^R,\mathbf{T}\big).
\end{aligned}
\end{equation}
Thus, the variance aggregates covariances between each participant and their neighbors up to degree $M$.

In particular,
\[
\mathrm{Var}\!\big(Y_i \mid \mathbf{A},\mathbf{R},\mathbf{D}^R,\mathbf{T}\big)
= \mathrm{Var}\!\Big(r\big(D_i,g_i(D_{\mathcal{N}_i}),T_i,\varepsilon_i\big) \,\Big|\, D_i, g_i(D_{\mathcal{N}_i}), T_i\Big)
= \sigma^2\!\big(T_i, D_i, g_i(D_{\mathcal{N}_i})\big),
\]
so the variance function is identifiable. Similarly, for $j\in\mathcal{N}_i^u$,
\[
\small
\begin{aligned}
\mathrm{Cov}\!\big(Y_i,Y_j \mid \mathbf{A},\mathbf{R},\mathbf{D}^R,\mathbf{T}\big)
&= \mathrm{Cov}\!\Big(r\big(D_i,g_i(D_{\mathcal{N}_i}),T_i,\varepsilon_i\big),\ r\big(D_j,g_j(D_{\mathcal{N}_j}),T_j,\varepsilon_j\big)\ \Big|\ j\in\mathcal{N}_i^u,\ \mathbf{D}^R\Big) \\
&=: \eta_u\!\big(T_i, D_i, g_i(D_{\mathcal{N}_i}),\ T_j, D_j, g_j(D_{\mathcal{N}_j})\big).
\end{aligned}
\]
That is, the covariance between two units at shortest-path distance $u$ depends only on (a) $u$; (b) their own assignments; (c) the assignments of their respective neighbors through $g_i,g_j$; and (d) the network statistics $(T_i,T_j)$.

The experimental design consists of minimizing the variance subject to Assumption~\ref{ass:covb}. Below we discuss inference. 
\begin{thm}\label{thm:multipledependence}
Suppose Assumptions~\ref{ass:model}, \ref{ass:higherorder}, and \ref{ass:covb} hold. If $n\,\mathbb{V}\!\big(\widehat{\Gamma}\ \big|\ \mathbf{A},\mathbf{T},\mathbf{R},\mathbf{D}^R\big) > 0$ almost surely, then
\[
\frac{\sqrt{n}\,\big(\widehat{\Gamma} - \mathbb{E}[\widehat{\Gamma}\mid \mathbf{A}, \mathbf{T}, \mathbf{R}, \mathbf{D}^R] \big)}
{\sqrt{\,n\,\mathbb{V}\!\big(\widehat{\Gamma}\ \big|\ \mathbf{A}, \mathbf{T}, \mathbf{R}, \mathbf{D}^R\big)}} \ \rightarrow_d\ \mathcal{N}(0,1).
\]
\end{thm}
The proof is in Appendix~\ref{sec:asymp_app}. A consistent variance estimator is
\begin{equation}\label{eqn:var2hat}
\begin{aligned}
n\,\widehat{V}
= &\frac{1}{n}\sum_{i:\,R_i=1} w(i;\mathbf{R},\mathbf{D}^R)\,\big(Y_i-\hat m_i\big)^2
\\ &+ \sum_{u=1}^{M} \frac{1}{n}\sum_{i:\,R_i=1}\sum_{j\in \mathcal{N}_i^u} R_j\, w(i;\mathbf{R},\mathbf{D}^R)\,\big(Y_i-\hat m_i\big)\, w(j;\mathbf{R},\mathbf{D}^R)\,\big(Y_j-\hat m_j\big).
\end{aligned}
\end{equation}
Regret guarantees are analogous to those in the main text (under uniform degree bounds) and are omitted for brevity.

 \subsection{Design with Partial Network Information}\label{sec:model_assisted}

In this section, we consider the case in which the researcher observes only partial network information. The main assumption is that the population contains at least two disconnected components (clusters), and the pilot is drawn from a component that is disconnected from the set of units eligible for the main experiment.

\noindent\textit{Experimental protocol.}

\noindent\textbf{1. Pilot study.}
Researchers select a random sample of individuals and set $P_i=1$ for those units. This pilot set is assumed to be disconnected from all other eligible units. Assign treatments $D_i \stackrel{\text{iid}}{\sim} \mathrm{Bernoulli}(0.5)$ for $P_i=1$. The pilot may be drawn from a disconnected component $\mathcal H$ (e.g., a village \citep{banerjee2013diffusion}, a school \citep{paluck2016changing}, or a region \citep{muralidharan2017general}). Researchers observe
\[
\big\{P_i\,(Y_i, D_i, T_i, D_{\mathcal N_i})\big\},\qquad i\in\mathcal H:\quad \big(\{i\}\cup\mathcal N_i\big)\cap\{1,\ldots,N\}=\emptyset,
\]
where $\{1,\ldots,N\}$ denotes the set of units eligible for the main experiment. From the pilot, estimate $(\widehat\sigma_p,\widehat\eta_p)$; note that $\widehat\eta_p$ requires that some pilot units have at least one pilot neighbor.

\noindent\textbf{2. Survey.}
Researchers collect partial network information for a random subset $\mathcal S\subseteq\{1,\ldots,N\}$ of eligible units:
\[
\tilde{\mathbf A}=\big(\mathcal N_i,\ i\in\mathcal S\big),\qquad \mathbf T=(T_1,\ldots,T_N).
\]

\noindent\textbf{3. Main experiment.}
Select participants and treatment assignments by minimizing the posterior (or prior-predictive) expected plug-in variance:
\begin{equation}\label{eqn:design_unobserved}
\min_{\mathbf r,\mathbf d^{\,r}}\ \mathbb E_{\mathbf A}\!\left[\ \widehat V_{N,\widehat\sigma_p,\widehat\eta_p}(\mathbf r,\mathbf d^{\,r})\ \Big|\ \tilde{\mathbf A},\mathbf T\right]
\quad \text{s.t.}\quad
\sum_{i=1}^N r_i \in [\underline n,\bar n],\ \ r_i=0\ \ \forall i\in\mathcal H,
\end{equation}
where the expectation is with respect to the distribution of $\mathbf A$ given $(\tilde{\mathbf A},\mathbf T)$ under a network prior $\mathcal P_{\mathbf T}$. Additional constraints (e.g., bounds on weights as in \eqref{eqn:design1}) and optional randomization (as in Algorithm~\ref{alg:coefficients3b}) can be incorporated. By construction, no eligible unit in the main experiment belongs to the pilot component.

\noindent\textbf{4. Second survey and analysis.}
For each participant, collect
\[
\big\{R_i\,(Y_i, D_i, D_{j\in\mathcal N_i}, \mathcal N_i)\big\},
\]
construct $\widehat{\Gamma}$ as in \eqref{eqn:estimator}, and estimate variance $\widehat V$ as in \eqref{eqn:final_vv}. Inference proceeds as in Theorem~\ref{thm:asym_t}, provided the exposure mapping $g_i(D_{\mathcal N_i})$ is observed for participants (e.g., $g_i(D_{\mathcal N_i})=\sum_{k\in\mathcal N_i} D_k$).

\medskip
\noindent The experiment therefore consists of four steps: (i) a pilot (with pilot neighbors observed), (ii) a first survey collecting partial network data, (iii) the design stage, and (iv) the analysis. Notably, $\widehat{\Gamma}$ and the pilot-based functions $(\widehat\sigma_p^2,\widehat\eta_p)$ can be formed without observing the \emph{entire} network; it suffices to observe outcomes, covariates, and treatment assignments for participants and their neighbors.

In the main experiment, we optimize the expected variance, holding $(\widehat\sigma_p^2,\widehat\eta_p)$ fixed from the pilot and averaging over the distribution of missing edges. The network prior can be obtained from a formation model \citep[e.g.,][]{breza2017using}. We solve the objective by alternating (a) Monte Carlo draws of the missing edges and (b) mixed-integer nonlinear optimization over $(\mathbf R,\mathbf D^R)$; Section~\ref{sec:numerics} shows this performs well in practice.\footnote{One could instead adopt a fully Bayesian approach by placing priors on potential outcomes and integrating over the joint posterior. Developing such a hierarchical model is beyond our scope and left for future work.}

\begin{prop}\label{prop:unobserved}
Let Assumptions~\ref{ass:model} and~\ref{ass:modelb} hold. Then the experimental design in this section satisfies the restrictions in Proposition~\ref{ass:cov}.
\end{prop}

\subsection{Minimax Design in the Absence of Pilots}

Whenever the variance and covariance functions are unavailable to the researcher, we devise an optimization algorithm, worst-case over variances and covariances.  

Suppose that the researcher has prior knowledge on  
\begin{equation} 
\sigma^2 \in \mathcal{S}, \quad \eta \in \mathcal{E},  
\end{equation} 
where for some bound $B \in (0, \infty)$, 
\begin{equation}
\begin{aligned} 
&\mathcal{S} = \{f: \{0,1\} \times \mathbb{Z}^2 \mapsto \mathbb{R}_+, \quad ||f||_{\infty} \le B^2 \} \\
&\mathcal{E} = \{g: \{0,1\} \times \mathbb{Z}^2 \times \{0,1\} \times \mathbb{Z}^2 \mapsto [-B^2, B^2] \} . 
\end{aligned}  
\end{equation} 
The function class encodes upper and lower bounds on the variance and covariance function.

Then in this case we can minimize the worst-case variance, taking form 
\begin{equation} \label{eqn:consaa}
\begin{aligned} 
& \quad  \sup_{\sigma^2 \in \mathcal{S}, \eta \in \mathcal{E}} \hat{V}_{\sigma,\eta}(\cdot)
. \end{aligned} 
\end{equation} 

The optimization can be written with respect to additional parameters $\sigma_i^2$ which denote the variance of each element $i$ and the parameters $\eta_{i,j}$ which denote the covariance between $i,j$. The supremum is taken over a finite set of such parameters, under the constraint that $\sigma_i^2 = \sigma_j^2$ whenever $i$ and $j$ have the same treatment status, number of treated neighbors and $T_i = T_j$. Similarly for any pair $(\eta_{i,j}, \eta_{u,v})$.

\subsection{Algorithmic description for variance and covariance estimation} 

In Algorithm  \ref{alg:pilot_varcov} we formalize estimation of the variance and covariance recommended in Section \ref{sec:guide_practice}.

\begin{algorithm}[!ht]
\caption{Pilot variance–covariance estimation}\label{alg:pilot_varcov}
\begin{algorithmic}[1]
\Require Network $\mathbf A$, types $\mathbf T$, pilot indicator $\mathbf P$; pilot data $\{(Y_i,D_i,T_i,D_{\mathcal{N}_i}): P_i=1\}$; exposure mapping $g_i(\cdot)$; feature map $f(\cdot)$ for $W_i$; outcome model class for $\hat m_p$; prior lower and upper bounds on the correlations between two individuals $(\underline{\alpha}, \bar{\alpha})$ 
\Ensure $(\widehat\sigma_p^{2}(\cdot),\,\widehat\eta_p(\cdot))$
\State \textit{Outcome regression (pilot).} Fit $\hat m_p$ to $\{(Y_i,D_i,g_i(D_{\mathcal{N}_i}),T_i): P_i=1\}$ and compute residuals
\[
\hat\varepsilon_i \;:=\; Y_i - \hat m_p\!\big(D_i, g_i(D_{\mathcal{N}_i}), T_i\big),\qquad (P_i=1).
\]
\State \textit{Polynomial transformation.} For each pilot unit, set
\[
W_i \;:=\; f\!\big(T_i, D_i, g_i(D_{\mathcal{N}_i})\big).
\]
\State \textit{Nonnegative least squares variance fit.} Solve
\[
\hat\beta \in \arg\min_{\beta}\; \sum_{i:P_i=1}\big(\hat\varepsilon_i^2 - W_i^\top \beta\big)^2
\quad \text{s.t.}\quad W_i^\top \beta \ge 0\ \ \forall\, i\ (P_i=1).
\]
Define
\[
\widehat\sigma_p^{2}\!\big(T_i, D_i, g_i(D_{\mathcal{N}_i})\big) \;:=\; \max\{W_i^\top \hat\beta,\,0\},
\qquad
\widehat\sigma_i \;:=\; \sqrt{\widehat\sigma_p^{2}\!\big(T_i, D_i, g_i(D_{\mathcal{N}_i})\big)}.
\]
\State \textit{Neighbor edge set (pilot).} Let $\mathcal E_p := \{(i,j): P_i=P_j=1,\ A_{ij}=1,\ i<j\}$.
\State \textit{Constant correlation fit for covariances.} For each $(i,j)\in\mathcal E_p$, set $Z_{ij} := \widehat\sigma_i\,\widehat\sigma_j$. Compute
\[
\widehat\alpha \;:=\; 
\frac{\sum_{(i,j)\in \mathcal E_p} Z_{ij}\,\hat\varepsilon_i \hat\varepsilon_j}
     {\sum_{(i,j)\in \mathcal E_p} Z_{ij}^2},
\qquad \text{(optionally truncate $\widehat\alpha$ to $[\underline{\alpha},\bar{\alpha}]$).}
\]
\State \textit{Covariance function.} For neighbors $i\in\mathcal{N}_j$,
\[
\widehat\eta_p\!\big(T_i, D_i, g_i(D_{\mathcal{N}_i}),\, T_j, D_j, g_j(D_{\mathcal{N}_j})\big)
\;:=\; \widehat\alpha\;\widehat\sigma_i\,\widehat\sigma_j,
\]
and set $\widehat\eta_p(\cdot)=0$ otherwise.
\State \textbf{Return} $(\widehat\sigma_p^{2}(\cdot),\,\widehat\eta_p(\cdot))$.
\end{algorithmic}
\end{algorithm}

\section{Additional Tables and Figures} \label{sec:aa2}

We collect results of the simulated network in Table \ref{tab:1}, and Table \ref{tab:3}. Each table reports the variance averaged over two-hundred replications. Each design corresponds to a different set of parameters $(\beta_1, \beta_2)$, which can be found at the top of the table. In Figure \ref{fig:bias} we report the box plot for the bias.

\begin{table}[!htbp] \centering 
  \caption{Simulated network. Variance of the overall effect (sum of spillover and treatment effects). $200$ replications. Each column corresponds to different values of the coefficient. Panel at the top collects results for the Erd\H{o}s-Rényi graph and at the bottom for the Albert-Barabasi graph.} 
  \label{tab:1} 
\begin{tabular}{@{\extracolsep{-1.7pt}} ccccccccc} 
\\[-1.8ex]\hline 
\hline \\[-1.8ex] 
ER & (0,0) & (0, 0.5) & (0,1) & (0, 1.5) & (0.5,0.5) & (0.5,1) & (0.5,1.5) & (1,1.5) \\ 
\hline \\[-1.8ex] 
ELI & $0.624$ & $0.927$ & $1.194$ & $1.415$ & $1.199$ & $1.414$ & $1.637$ & $1.853$ \\ 
Random All+ & $1.162$ & $1.739$ & $2.315$ & $2.891$ & $2.329$ & $2.905$ & $3.479$ & $4.068$ \\ 
Graph Clust+ & $0.640$ & $0.991$ & $1.343$ & $1.694$ & $1.361$ & $1.713$ & $2.063$ & $2.434$ \\ 
Saturation1+ & $0.908$ & $1.378$ & $1.849$ & $2.316$ & $1.859$ & $2.330$ & $2.801$ & $3.282$ \\ 
Graph Clust & $0.767$ & $1.188$ & $1.607$ & $2.029$ & $1.631$ & $2.051$ & $2.471$ & $2.916$ \\ 
Saturation1 & $1.090$ & $1.654$ & $2.217$ & $2.781$ & $2.231$ & $2.794$ & $3.358$ & $3.932$ \\ 
Saturation2+ & $0.679$ & $1.047$ & $1.416$ & $1.783$ & $1.430$ & $1.800$ & $2.169$ & $2.550$ \\ 
Saturation3+ & $0.993$ & $1.587$ & $2.177$ & $2.771$ & $2.178$ & $2.771$ & $3.364$ & $3.954$ \\ 
\hline \\[-1.8ex] 
\end{tabular} 
\begin{tabular}{@{\extracolsep{-1.7pt}} ccccccccc} 
\\[-1.8ex]\hline 
\hline \\[-1.8ex] 
AB & (0,0) & (0, 0.5) & (0,1) & (0, 1.5) & (0.5,0.5) & (0.5,1) & (0.5,1.5) & (1,1.5) \\ 
\hline \\[-1.8ex] 
ELI & $0.714$ & $0.909$ & $1.278$ & $1.566$ & $1.294$ & $1.548$ & $1.595$ & $2.031$ \\ 
Random All+ & $1.144$ & $1.714$ & $2.284$ & $2.851$ & $2.299$ & $2.874$ & $3.482$ & $4.028$ \\ 
Graph Clust+ & $0.693$ & $1.098$ & $1.503$ & $1.908$ & $1.531$ & $1.938$ & $2.060$ & $2.773$ \\ 
Saturation1+ & $0.936$ & $1.435$ & $1.936$ & $2.434$ & $1.950$ & $2.451$ & $2.800$ & $3.464$ \\ 
Graph Clust & $0.837$ & $1.325$ & $1.811$ & $2.299$ & $1.845$ & $2.333$ & $2.471$ & $3.338$ \\ 
Saturation1 & $1.132$ & $1.733$ & $2.336$ & $2.934$ & $2.354$ & $2.955$ & $3.358$ & $4.179$ \\ 
Saturation2+ & $0.732$ & $1.152$ & $1.572$ & $1.992$ & $1.594$ & $2.015$ & $2.169$ & $2.882$ \\ 
Saturation3+ & $1.091$ & $1.762$ & $2.433$ & $3.103$ & $2.425$ & $3.096$ & $3.364$ & $4.425$ \\ 
\hline \\[-1.8ex] 
\end{tabular} 
\end{table}

\begin{table}[!htbp] \centering 
  \caption{Simulated network. Maximum variance between estimator of the direct treatment and spillover effect. $200$ replications. Each column corresponds to different values of the coefficient. Panel at the top collects results for the Erd\H{o}s-Rényi graph and at the bottom for the Albert-Barabasi graph.} 
  \label{tab:3} 
\begin{tabular}{@{\extracolsep{-1.7pt}} ccccccccc} 
\\[-1.8ex]\hline 
\hline \\[-1.8ex] 
ER & (0,0) & (0, 0.5) & (0,1) & (0, 1.5) & (0.5,0.5) & (0.5,1) & (0.5,1.5) & (1,1.5) \\ 
\hline \\[-1.8ex] 
ELI & $0.545$ & $0.782$ & $1.091$ & $1.308$ & $1.102$ & $1.321$ & $1.580$ & $1.864$ \\ 
Random All+ & $0.676$ & $1.037$ & $1.400$ & $1.763$ & $1.379$ & $1.740$ & $2.101$ & $2.441$ \\ 
Graph Clust+ & $1.224$ & $1.674$ & $2.120$ & $2.568$ & $2.601$ & $3.046$ & $3.497$ & $4.424$ \\ 
Saturation1+ & $0.678$ & $1.036$ & $1.395$ & $1.756$ & $1.409$ & $1.769$ & $2.128$ & $2.501$ \\ 
Graph Clust & $1.496$ & $2.038$ & $2.585$ & $3.129$ & $3.173$ & $3.717$ & $4.259$ & $5.397$ \\ 
Saturation1 & $0.825$ & $1.262$ & $1.698$ & $2.136$ & $1.715$ & $2.150$ & $2.588$ & $3.037$ \\ 
Saturation2+ & $0.969$ & $1.393$ & $1.820$ & $2.247$ & $2.053$ & $2.478$ & $2.901$ & $3.564$ \\ 
Saturation3+ & $0.930$ & $1.474$ & $2.016$ & $2.562$ & $1.930$ & $2.473$ & $3.013$ & $3.474$ \\ 
\hline \\[-1.8ex] 
\end{tabular} 
\begin{tabular}{@{\extracolsep{-1.7pt}} ccccccccc} 
\\[-1.8ex]\hline 
\hline \\[-1.8ex] 
AB & (0,0) & (0, 0.5) & (0,1) & (0, 1.5) & (0.5,0.5) & (0.5,1) & (0.5,1.5) & (1,1.5) \\ 
\hline \\[-1.8ex] 
ELI & $0.571$ & $0.792$ & $1.081$ & $1.359$ & $1.059$ & $1.399$ & $1.574$ & $1.879$ \\ 
Random All+ & $0.672$ & $1.057$ & $1.443$ & $1.830$ & $1.398$ & $1.782$ & $2.101$ & $2.510$ \\ 
Graph Clust+ & $0.984$ & $1.383$ & $1.784$ & $2.184$ & $2.192$ & $2.594$ & $3.495$ & $3.809$ \\ 
Saturation1+ & $0.676$ & $1.060$ & $1.444$ & $1.829$ & $1.453$ & $1.837$ & $2.127$ & $2.613$ \\ 
Graph Clust & $1.204$ & $1.689$ & $2.175$ & $2.661$ & $2.678$ & $3.163$ & $4.261$ & $4.638$ \\ 
Saturation1 & $0.827$ & $1.294$ & $1.763$ & $2.233$ & $1.773$ & $2.239$ & $2.587$ & $3.189$ \\ 
Saturation2+ & $0.859$ & $1.262$ & $1.665$ & $2.066$ & $1.902$ & $2.307$ & $2.904$ & $3.350$ \\ 
Saturation3+ & $0.984$ & $1.590$ & $2.196$ & $2.805$ & $2.107$ & $2.713$ & $3.015$ & $3.834$ \\ 
\hline \\[-1.8ex] 
\end{tabular} 
\end{table} 

\begin{figure}
\centering 
\includegraphics[scale=0.6]{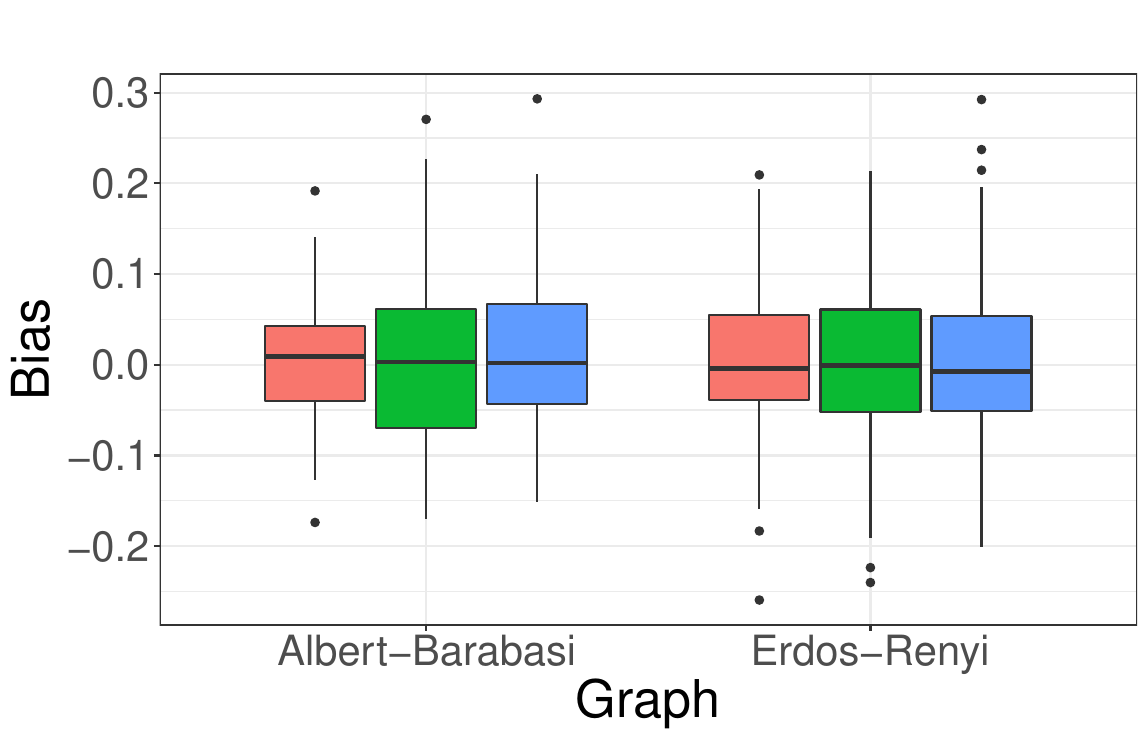}
\caption{Box-plot of the difference between point estimate and true value for Albert-Barabasi and Erdos-Renyi networks, with $(\beta_1, \beta_2) = (0.5, 0.5)$. The box plot reports the results for two-hundred replications. The red box corresponds to the bias for estimating the global effect (sum of direct and spillover effects), the green color corresponds to the direct effect and the blue color to the spillover effect.} \label{fig:bias} 
\end{figure}

\newpage

\section{Proofs} 

In the proofs, we write $\mathbf T_{-i}$ for the vector $\mathbf T$ with its $i$th entry removed; analogously, $\mathbf R_{-i}$ and $\mathbf D_{-i}$ denote $\mathbf R$ and $\mathbf D$ without their $i$th entries. 

When we condition on
\[
\big|\, g_i(D_{\mathcal N_i}) = s,\ \mathbf D_{-i}\,,
\]
we mean that we condition jointly on (i) the treatments of all units other than $i$ and (ii) the event that the exposure mapping for unit $i$ equals $s$ (i.e., $g_i(D_{\mathcal N_i})=s$).

\subsection{Identification}

\subsubsection{Proof of Proposition \ref{ass:cov}} \label{proof:ass:cov}

We want to show that 
\begin{equation}
\mathbb{E}\Big[Y_i \Big| D_i = d, g_i(D_{\mathcal{N}_i}) = s, T_i = l, R_i = 1, \mathbf{D}_{-i}, \mathbf{A}, \mathbf{R}_{-i}, \mathbf{T}_{-i}, \mathbf{P}\Big] = m(d, s, l). 
\end{equation}
 Under Assumption \ref{ass:model},
\begin{equation}
\begin{aligned} 
&\mathbb{E}\Big[Y_i \Big| D_i = d,  g_i(D_{\mathcal{N}_i}) = s, T_i = l, R_i = 1, \mathbf{D}_{-i}, \mathbf{A}, \mathbf{R}_{-i}, \mathbf{T}_{-i}, \mathbf{P}\Big]  = \\ &\mathbb{E}\Big[r\Big(d, s, l, \varepsilon_i \Big) \Big |D_i = d,  g_i(D_{\mathcal{N}_i}) = s, T_i = l, R_i = 1, \mathbf{D}_{-i}, \mathbf{A}, \mathbf{R}_{-i}, \mathbf{T}_{-i}, \mathbf{P} \Big].
\end{aligned} 
\end{equation} 
Observe now that under the conditions in Proposition \ref{ass:cov}, for all units not in the pilot study and not friends of individuals in the pilot study, and since $\varepsilon_i(\mathbf{d})$ is a constant function in $\mathbf{d}$, we have 
\begin{equation}
\small 
\begin{aligned} 
& \mathbb{E}\Big[r\Big(d, s, l, \varepsilon_i \Big) \Big | D_i = d,  g_i(D_{\mathcal{N}_i}) = s, T_i = l, R_i = 1, \mathbf{D}_{-i}, \mathbf{A}, \mathbf{R}_{-i}, \mathbf{T}_{-i}, \mathbf{P} \Big] = \\ & 
\mathbb{E}\Big[r\Big(d, s, l, \varepsilon_i \Big) \Big | T_i = l,\mathbf{A}, \mathbf{T}_{-i}, \mathbf{P} \Big].  
\end{aligned}  
\end{equation} 
Since $\mathbf{P}$ is exogenous conditional on $(\mathbf{A}, \mathbf{T})$, we have 
$$
\mathbb{E}\Big[r\Big(d, s, l, \varepsilon_i \Big) \Big | T_i = l,\mathbf{A}, \mathbf{T}_{-i}, \mathbf{P} \Big] =  \mathbb{E}\Big[r\Big(d, s, l, \varepsilon_i \Big) \Big | T_i = l,\mathbf{A}, \mathbf{T}_{-i} \Big]. 
$$ 
Under Assumption \ref{ass:model}, since $\varepsilon_i \perp (\mathbf{A}, \mathbf{T})$, the proof completes. 

\subsubsection{Proof of Theorem \ref{thm:1C}} \label{proof:thm:1C}

The proof follows similarly to the previous proof. Under Equation \eqref{eqn:first},
\begin{equation}
\begin{aligned} 
&\mathbb{E}\Big[Y_i \Big | D_i = d, g_i(D_{\mathcal{N}_i}) = s, T_i = l, R_i = 1, \mathbf{D}_{-i}, \mathbf{A}, \mathbf{R}_{-i}, \mathbf{T}_{-i}, \mathbf{P}\Big] = \\ &\mathbb{E}\Big[r\Big(d, s, l, \varepsilon_i \Big) \Big | D_i = d,  g_i(D_{\mathcal{N}_i}) = s, T_i = l, R_i = 1, \mathbf{D}_{-i}, \mathbf{A}, \mathbf{R}_{-i}, \mathbf{T}_{-i}, \mathbf{P} \Big].
\end{aligned} 
\end{equation} 
Observe now that under Assumption \ref{ass:covb}, since participants are not units in the pilot study and their neighbors up to the M$^{th}$ degree, we have 
\begin{equation}
\small 
\begin{aligned} 
& \mathbb{E}\Big[r\Big(d, s, l, \varepsilon_i \Big) \Big |D_i = d,  g_i(D_{\mathcal{N}_i}) = s, T_i = l, R_i = 1, \mathbf{D}_{-i}, \mathbf{A}, \mathbf{R}_{-i}, \mathbf{T}_{-i}, \mathbf{P} \Big] \\ 
&= 
\mathbb{E}\Big[r\Big(d, s, l, \varepsilon_i \Big) \Big | \mathbf{A}, \mathbf{T}_{-i}, T_i = l \Big].  
\end{aligned}  
\end{equation} 
Under Equation \eqref{eqn:first}, since $\varepsilon_i \perp (\mathbf{A}, \mathbf{T})$, the proof completes.

\subsubsection{Proof of Lemma \ref{lem:vara}} \label{proof:lemma:vara}

For all individuals $i$ selected in the pilot, we can write for $\mathbf{T}$ such that $T_i = l$,  
$$
\small 
\begin{aligned} 
 \mathrm{Var}\Big(Y_i \Big| \mathbf{A}, \mathbf{T}_{-i}, D_i = d, T_i = l,D_{\mathcal{N}_i} = \mathbf{s}, \mathbf{P}\Big) & = \mathrm{Var}\Big(r(D_i, g_i(D_{\mathcal{N}_i}), T_i, \varepsilon_i) \Big| \mathbf{A}, \mathbf{T}_{-i}, D_i = d, T_i = l,D_{\mathcal{N}_i} = \mathbf{s}, \mathbf{P}\Big) \\ 
 &= \mathrm{Var}\Big(r(d, g_i(\mathbf{s}), l, \varepsilon_i) \Big| \mathbf{A}, T_i = l,  \mathbf{T}_{-i}, \mathbf{P}\Big) \\ 
 &= \mathrm{Var}\Big(r(d, g_i(\mathbf{s}), l, \varepsilon_i) \Big| \mathbf{A}, \mathbf{T}_{-i},  T_i = l\Big) \\ 
 & = \mathrm{Var}\Big(r(d, g_i(\mathbf{s}), l, \varepsilon_i) | \mathbf{A}, \mathbf{T} \Big) \\ 
 &= \sigma^2(d,g_i(\mathbf{s}), l),   
\end{aligned} 
$$ 
for some function $\sigma^2$. 
The first equation follows from Assumption \ref{ass:model}. The second equation follows from the fact that treatments are randomized exogenously and $\varepsilon_i(\mathbf{D})$ is constant in $\mathbf{D}$. The third equation follows from the fact that $\mathbf{P}$ is independent of $\varepsilon_i$ conditional on $\mathbf{A}, \mathbf{T}$. The last equation follows from the fact that $\varepsilon_i \perp (\mathbf{T}, \mathbf{A})$. The analysis of the covariances follows similarly. 
Namely, following the same steps as before, for $\mathbf{T}$ such that $T_i = l, T_j = l'$
\begin{equation} \label{eqn:helpera}
\begin{aligned}
& \mathrm{Cov}\Big(Y_i, Y_j\Big|\mathbf{A}, \mathbf{T}_{-(i,j)}, D_i = d, D_j = d', D_{\mathcal{N}_i} = \mathbf{s}, D_{\mathcal{N}_j} = \mathbf{s}', T_i = l, T_j = l', \mathbf{P}\Big) =  \\ &
\mathrm{Cov}\Big(r(d, g_i(\mathbf{s}), l, \varepsilon_i), r(d', g_j(\mathbf{s}'), l', \varepsilon_j)\Big|\mathbf{A}, \mathbf{T}_{-(i,j)}, T_i = l, T_j = l', \mathbf{P}\Big) 
\end{aligned} 
\end{equation}  
where we used the exogeneity of the treatment assignment and the fact that $\varepsilon_i(\mathbf{d})$ is a constant function in $\mathbf{d}$. Using Assumption \ref{ass:modelb}, 
$$
\eqref{eqn:helpera} = \begin{cases}  \mathrm{Cov}\Big(r(d, g_i(\mathbf{s}), l, \varepsilon_i), r(d', g_j(\mathbf{s}'), l', \varepsilon_j)\Big| j \in \mathcal{N}_i, \mathbf{P}, \mathbf{T}_{-(i,j)}, \mathbf{A},  T_i = l, T_j = l'\Big) & \text{ if } j \in \mathcal{N}_i \\
0  &\text{ otherwise}. 
\end{cases} 
$$ 
Using Assumption \ref{ass:modelb}(ii), and the fact that $\mathbf{P}$ is independent of $(\varepsilon_i, \varepsilon_j)$ conditional on $\mathbf{A}, \mathbf{T}$, we can write 
$$
\small 
\begin{aligned} 
 & \mathrm{Cov}\Big(r(d, g_i(\mathbf{s}), l, \varepsilon_i), r(d', g_j(\mathbf{s}'), l', \varepsilon_j)\Big| j \in \mathcal{N}_i, \mathbf{A}, \mathbf{T}_{-(i,j)}, \mathbf{P}, T_i = l, T_j = l', j \in \mathcal{N}_i\Big) \\ &=  
 \mathrm{Cov}\Big(r(d, g_i(\mathbf{s}), l, \varepsilon_i), r(d', g_j(\mathbf{s}'), l', \varepsilon_j)\Big| j \in \mathcal{N}_i, T_i = l, T_j = l', \mathbf{A}, \mathbf{T}_{-(i,j)}\Big) \\ 
 &= \eta\Big(l, d, g_i(\mathbf{s}), l', d', g_j(\mathbf{s}')\Big), 
\end{aligned} 
$$ 
where the last equation follows from Assumption \ref{ass:modelb}.

 \subsubsection{Additional lemmas for identification of the variance and covariance in the main experiment}

The following lemma guarantees identification of the variance and covariance component in the main experiment. 
 
 \begin{lem} \label{lem:varab}
Suppose that Assumption \ref{ass:model}, \ref{ass:modelb} hold. Consider an experimental design in Algorithm \ref{alg:coefficients2} or Algorithm \ref{alg:3}, with pilot chosen as in Algorithm \ref{alg:pilot}. Then for all units $(i,j)$ such that $R_i = R_j = 1$, 
\begin{equation} 
\small 
\begin{aligned} 
\mathbb{V}(Y_i | \mathbf{A}, \mathbf{D}^R, \mathbf{R}, \mathbf{T}, \mathbf{P}) &=  \sigma^2\Big(T_i, D_i, g_i(D_{\mathcal{N}_i})\Big). \\
\mathrm{Cov}(Y_i, Y_j | \mathbf{A}, \mathbf{D}^R, \mathbf{R}, \mathbf{T}, \mathbf{P}) &= \begin{cases}  \eta\Big(T_i, D_i,  g_i(D_{\mathcal{N}_i}), T_j, D_j,  g_j(D_{\mathcal{N}_j})\Big) & \text{ if } i \in \mathcal{N}_j\\
0 &\text{ otherwise}  
\end{cases} 
\end{aligned} 
\end{equation} 
for the same functions $\sigma^2(.), \eta(.)$ as in Lemma \ref{lem:vara}.
\end{lem}

\begin{proof} Under Assumption \ref{ass:model}, 
 \begin{equation}
 \begin{aligned} 
 &\mathrm{Var}\Big(Y_i \Big| R_i = 1, D_i = d,  g_i(D_{\mathcal{N}_i}) = s, \mathbf{D}_{-i}, \mathbf{R}_{-i},\mathbf{A}, T_i = l, R_i = 1, \mathbf{T}_{-i}, \mathbf{P}\Big)   \\ 
 = &\mathrm{Var}\Big(r(d, s, l, \varepsilon_i ) \Big| D_i = d,  g_i(D_{\mathcal{N}_i}) = s, \mathbf{D}_{-i}, \mathbf{R}_{-i},\mathbf{A}, T_i = l, R_i = 1, \mathbf{T}_{-i}, \mathbf{P}\Big). 
\end{aligned}  
 \end{equation}
 Under Assumption \ref{ass:modelb}, since $R_i = 1$ only if $i$ is neither in the pilot study nor is a friend of individuals in the pilot study, it follows that 
 \begin{equation} 
 \begin{aligned}
 &\mathrm{Var}\Big(r(d, s, l, \varepsilon_i) \Big| D_i = d,  g_i(D_{\mathcal{N}_i}) = s, \mathbf{D}_{-i}, \mathbf{R}_{-i},\mathbf{A}, T_i = l, R_i = 1, \mathbf{P}, \mathbf{T}_{-i}\Big)  \\ &
=  \mathrm{Var}\Big(r(d, s, l, \varepsilon_i) \Big| \mathbf{T},\mathbf{A}, \mathbf{P} \Big) = \mathrm{Var}\Big(r(d, s, l, \varepsilon_i) \Big), 
 \end{aligned} 
 \end{equation} 
 where the first equality follows from the fact that $\mathbf{D}, \mathbf{R}$ are by design independent of $\varepsilon_i$ conditional on $\mathbf{P}$, for all units not in the pilot study and not friends of individuals in the pilot (by Assumption \ref{ass:modelb} and construction of Algorithms \ref{alg:coefficients2} and \ref{alg:3}). The second equality follows from Assumption \ref{ass:model}, since $\varepsilon_i \perp (\mathbf{A}, \mathbf{T})$, and the fact that $\mathbf{P}$ is a deterministic function of $\mathbf{A}, \mathbf{T}$.

For the covariance component, the same reasoning follows. Consider $(i,j): R_i = R_j = 1$, and $T_i = l, T_j = l'$. Then we write 
 \begin{equation} \label{eqn:helperb}
 \begin{aligned} 
 &\mathrm{Cov}\Big(Y_i, Y_j \Big| D_i = d, D_j = d',  g_i(D_{\mathcal{N}_i}) = s,  g_j(D_{\mathcal{N}_i}) = s', \mathbf{D}_{-(i,j)}, \mathbf{R}_{-(i,j)},\mathbf{A}, T_i = l, T_j = l', \\ &R_i = 1, R_j = 1, \mathbf{T}_{-(i,j)}, \mathbf{P}\Big) =  \\ 
 &\mathrm{Cov}\Big(r(d, s, l, \varepsilon_i), r(d', s', l', \varepsilon_j) \Big| D_i = d, D_j = d',  g_i(D_{\mathcal{N}_i}) = s,  g_j(D_{\mathcal{N}_i}) = s', \mathbf{D}_{-(i,j)}, \\ &\mathbf{R}_{-(i,j)},\mathbf{A}, T_i = l, T_j = l', R_i = 1, R_j = 1, \mathbf{T}_{-(i,j)}, \mathbf{P}\Big) .
\end{aligned}  
 \end{equation}
 By Assumption \ref{ass:modelb}, by design $\mathbf{D}, \mathbf{R}$ are independent of $(\varepsilon_i, \varepsilon_j)$ conditional on $(\mathbf{A}, \mathbf{T}, \mathbf{P})$ for all units $(i,j)$ that are neither in the pilot study nor are friends of individuals in the pilot study. Because $\varepsilon_i(\mathbf{d})$ is a constant function in $\mathbf{d}$ we obtain that Equation \eqref{eqn:helperb} equals
\begin{equation} \label{eqn:hgf}
 \begin{aligned} 
\mathrm{Cov}\Big(r(d, s, l, \varepsilon_i), r(d', s', l', \varepsilon_j) \Big| \mathbf{A}, T_i = l, T_j = l', \mathbf{T}_{-(i,j)}, \mathbf{P}\Big) .
\end{aligned}  
 \end{equation}
The covariance is zero if two individuals are not neighbors. In such a case the lemma trivially holds. Therefore, consider the case where individuals are neighbors. Then the above equation equals (under Assumption \ref{ass:modelb}) 
$$
 \begin{aligned} 
\eqref{eqn:hgf} = \begin{cases} 
\mathrm{Cov}\Big(r(d, s, l, \varepsilon_i), r(d', s', l', \varepsilon_j)\Big| T_i = l, T_j = l', j \in \mathcal{N}_i, \mathbf{A}, \mathbf{T} \Big) , \quad & j \in \mathcal{N}_i \\
 0 & \text{ otherwise} 
\end{cases} 
\end{aligned}  
$$ 
which follows since $\mathbf{P}$ is a deterministic function of $(\mathbf{A}, \mathbf{T})$. By the second condition in Assumption \ref{ass:modelb} the proof completes, where   the equivalence of the functions $\sigma^2, \eta$ with those in Lemma \ref{lem:vara} follows directly from the same expression for such functions in the proof of Lemma \ref{lem:vara} (Appendix \ref{proof:lemma:vara}). 

 \end{proof}

\subsection{Proof of Theorems \ref{thm:regret}, \ref{thm:regretb} and \ref{thm:regret2}} \label{proof:thm:regretb}

Note that Theorems~\ref{thm:regret} and~\ref{thm:regretb} are special cases of Theorem~\ref{thm:regret2}. 
For Theorem~\ref{thm:regret}, take $E=1$ and $\zeta=0$ in Theorem~\ref{thm:regret2}; for Theorem~\ref{thm:regretb}, take $E=1$ (allowing $\zeta\ge 0$). 
Therefore, it suffices to prove Theorem~\ref{thm:regret2}.

We first analyze the case $\zeta=0$, in which Algorithm~\ref{alg:3} yields $\check{\mathbf D}^{\,R}=\mathbf D^{\,R}$ (the re-randomization step is inactive). 
We then handle the case $\zeta>0$ at the end of the proof.

\paragraph{Preliminaries (studying the case of $\zeta = 0$)} We say that $a \lesssim b$ if $a \le \bar{C} b$ for a finite constant $\bar{C}$. Recall that $\mathbf{D}$ denotes the vector of treatments assigned to the pilot and main experiment as in Algorithms \ref{alg:pilot} and \ref{alg:3}.  Denote $\mathbf{R}$ the participation indicators obtained after solving the experimenter problem in Equation \eqref{eqn:design1b}.  We denote $\mathbf{D}^R$ the subvector of assignments for those units with $R_i = 1$ and their friends.  For an arbitrary $(\mathbf{D}^*, \mathbf{R}^*)$,  we denote 
 $$
 \hat{V}_p(\mathbf{D}^*, \mathbf{R}^*)= \max_{ e \in \{1, \cdots, E\} } \hat{V}_{\hat{\sigma}_p, \hat{\eta}_p}^{w^e}(\mathbf{D}^{*R^*},\mathbf{R}^*, \mathbf{T}, \mathbf{A}),
 $$ the maximum conditional variance over  the set of weights, using the estimated variance and covariance function from the pilot study. 
 
Similarly, let 
$$
V_N(\mathbf{D}^*, \mathbf{R}^*):= \max_{e \in \{1, \cdots, E\}} \mathbb{V}(\hat{\Gamma}(w^e)|\mathbf{A}, \mathbf{T}, \mathbf{D} = \mathbf{D}^*, \mathbf{R} = \mathbf{R}^*)
$$  
be its population counterpart. 

For notational convenience we refer to $w(i,\mathbf{D}, \mathbf{R})$, as the weight for unit $i$ evaluated at the values $\mathbf{R}, \mathbf{D}^R$.  Let
\begin{equation} \label{eqn:kk}
(\tilde{\mathbf{D}}, \tilde{\mathbf{R}}) \in  \mathrm{arg} \min_{\mathbf{d}, \mathbf{r}, \underline{n} \le \sum_{i = 1}^N r_i \le \bar{n}, r_j = 0 \forall j \in \mathcal{J}} V_N(\mathbf{r}, \mathbf{d}),
\end{equation}
the optimal participants' selection for known variance and covariance function and constraint on the pilot units as in Algorithm \ref{alg:3}. We define 
\begin{equation}
\sigma(i, \mathbf{D}) = \sigma(T_i, \mathbf{D}_i, g_i(D_{\mathcal{N}_i})) 
\end{equation} 
and similarly for $\eta(i,j,\mathbf{D})$, where $\sigma^2, \eta$ are defined as in Lemma \ref{lem:varab}. Recall that the population variance and covariance functions in Lemma \ref{lem:varab} only depend on treatment assignments (and $\mathbf{A}, \mathbf{T}$), but not on the experimental selector indicator $\mathbf{R}$. We define $\sigma(i, \cdot), \eta(i, j, \cdot)$ as a function of the vector of treatments $\mathbf{D}$ in the population of $N$ units, leaving implicit its dependence on $\mathbf{A}, \mathbf{T}$. Note that under Assumption \ref{ass:moment}, since $Y$ is uniformly bounded, also $\sigma^2(\cdot), \eta(\cdot)$ are uniformly bounded. 

Finally, note that in Algorithm \ref{alg:3} we can assign treatments optimally to individuals who are not in the pilot, but are their friends. This implies that, because individuals in the main experiment are not friends of those in the pilot, the only binding constraint in Algorithm \ref{alg:3} is the constraint in Equation \eqref{eqn:kk} (since the choice of the treatment indicators of pilot units does not affect $V_N$ or $\widehat{V}_N$).

\paragraph{Preliminary upper bound} First observe that $|\mathcal{J}| \le  \mathcal{N}_{\max} m$. Since $\underline{n} > \mathcal{N}_{\max} m/(\bar{n}/\underline{n} - 1) \ge |\mathcal{J}|/(\bar{n}/\underline{n} - 1)$  we have that the constraint $\mathbf{1}^\top \mathbf{r} = \bar{n}$ is a stricter constraint than $\underline{n} + |\mathcal{J}| \le \mathbf{1}^\top \mathbf{r} \le \bar{n}$. We can therefore write  
\begin{equation} \label{eqn:iii}
\small 
\begin{aligned} 
& \max_{e \in \{1, \cdots, E\}} \mathbb{V}\Big(\hat{\Gamma}(w^e) | \mathbf{R}, \mathbf{D}^R, \mathbf{A}, \mathbf{T}\Big) - \mathbb{V}_{\bar{n}}^\star(E) \\ &= V_N(\mathbf{D}, \mathbf{R}) - \min_{\mathbf{d}, \mathbf{r}, \sum_{i=1}^N r_i = \bar{n}} V_N(\mathbf{d}, \mathbf{r}) \\ &\le V_N(\mathbf{D}, \mathbf{R}) - \min_{\mathbf{d}, \mathbf{r},  \underline{n} +|\mathcal{J}| \le \sum_{i=1}^N r_i \le \bar{n}} V_N(\mathbf{d}, \mathbf{r})  \\ &\le V_N(\mathbf{D}, \mathbf{R}) - \min_{\mathbf{d}, \mathbf{r},  \underline{n} +|\mathcal{J}| \le \sum_{i=1}^N r_i \le \bar{n}} V_N(\mathbf{d}, \mathbf{r}) +  V_N(\tilde{\mathbf{D}}, \tilde{\mathbf{R}}) - \hat{V}_{p}(\mathbf{D}, \mathbf{R})   \\ & +\hat{V}_{p}(\mathbf{D}, \mathbf{R})  - V_N(\tilde{\mathbf{D}}, \tilde{\mathbf{R}})  
\\  &\le \underbrace{\Big(V_N(\mathbf{D}, \mathbf{R}) - \hat{V}_{p}(\mathbf{D}, \mathbf{R})\Big)}_{(i)} + \underbrace{\Big(\hat{V}_{p}(\tilde{\mathbf{D}}, \tilde{\mathbf{R}}) - V_N(\tilde{\mathbf{D}}, \tilde{\mathbf{R}})\Big)}_{(ii)} \\ &+ \underbrace{V_N(\tilde{\mathbf{D}}, \tilde{\mathbf{R}}) - \min_{\mathbf{d}, \mathbf{r}, \underline{n} + |\mathcal{J}| \le \sum_{i=1}^N r_i \le \bar{n}} V_N(\mathbf{d}, \mathbf{r})}_{(iii)}. 
\end{aligned} 
\end{equation} 
The last bound follows from the fact that $\hat{V}_{p}(\tilde{\mathbf{D}}, \tilde{\mathbf{R}}) \ge \hat{V}_{p}(\mathbf{D}, \mathbf{R})$, since $\tilde{\mathbf{R}}$ and $\mathbf{R}$ satisfy the same set of constraints, and $(\mathbf{D}^{\mathbf{R}}, \mathbf{R})$ minimizes $\hat{V}_{p}(\mathbf{D}, \mathbf{R})$. 
We study each component separately. 

\paragraph{Component $(i)$ and $(ii)$} We can write 
\begin{equation}
\begin{aligned} 
(i) \le \max_{e \in \{1, \cdots, E\}}  &\Big|\frac{1}{(\mathbf{1}^\top \mathbf{R})^2} \sum_{i = 1}^N w^{e}(i, \mathbf{D}, \mathbf{R})^2 R_i \Big(\sigma^2(i, \mathbf{D}) - \hat{\sigma}_p^2(i, \mathbf{D})   \Big)  \\ &+ \frac{1}{(\mathbf{1}^\top \mathbf{R})^2} \sum_{i = 1}^N \sum_{j \in \mathcal{N}_i} w^e(i, \mathbf{D}, \mathbf{R})w^e(j, \mathbf{D}, \mathbf{R}) R_i R_j \Big(\eta(i, j, \mathbf{D}) - \hat{\eta}_p(i, j, \mathbf{D})    \Big) \Big|.
\end{aligned} 
\end{equation} 
 Therefore, we obtain
\begin{equation} \label{eqn:vf}
\begin{aligned} 
&(i) \le \max_{ e \in \{1, \cdots, E\} } \underbrace{\Big|\frac{1}{(\mathbf{1}^\top \mathbf{R})^2} \sum_{i = 1}^N w^e(i, \mathbf{D}, \mathbf{R})^2R_i \Big(\sigma^2(i, \mathbf{D}) - \hat{\sigma}_p^2(i, \mathbf{D})   \Big) \Big|}_{(I)} \\
&+ \max_{ e \in \{1, \cdots, E\} } \underbrace{\Big|\frac{1}{(\mathbf{1}^\top \mathbf{R})^2} \sum_{i = 1}^N \sum_{j \in \mathcal{N}_i} w^e(i, \mathbf{D}, \mathbf{R})w^e(j, \mathbf{D}, \mathbf{R}) R_i R_j (\eta(i,j, \mathbf{D}) - \hat{\eta}_p(i,j, \mathbf{D}))   \Big|}_{(II)} .
\end{aligned} 
\end{equation} 

The above term satisfies (since $w(i, \cdot)$ is uniformly bounded by Assumption \ref{ass:weights_stable}(ii))
\begin{equation}
\small  
\begin{aligned} 
\eqref{eqn:vf} \lesssim & \mathcal{N}_{\max} \sup_{d,s,l,d',s',l'} \Big|\eta(d,s,l,d',s',l') - \hat{\eta}_p(d,s,d',s',l', l')\Big|\Big/\underline{n}   \\ &+ \sup_{d,s,l} \Big|\sigma^2(d,s,l) - \hat{\sigma}_p^2(d,s,l)\Big| \Big/\underline{n}.
\end{aligned} 
\end{equation}

The same reasoning also applies to the term $(ii)$ in Equation \eqref{eqn:iii}. Therefore, we can write under Assumption \ref{ass:convergence}
\begin{equation} \label{eqn:i_ii}
(i) + (ii) =  \mathcal{O}_p(\mathcal{N}_{\max} \bar{m}^{-\xi}/\underline{n}). 
\end{equation} 

\paragraph{$(iii)$ Part 1: Lower bound for $\min V_N(\mathbf{d}, \mathbf{r})$} Finally, consider the term $(iii)$. As a first step, we provide a lower bound to $\min_{\mathbf{d}, \mathbf{r},  \underline{n} + |\mathcal{J}| \le \sum_{i=1}^N r_i \le \bar{n}} V_N(\mathbf{d}, \mathbf{r})$.

We can write
\begin{equation} \label{eqn:hh}
\begin{aligned} 
&\min_{\mathbf{d}, \mathbf{r}, \underline{n} + |\mathcal{J}| \le \sum_{i=1}^N r_i \le \bar{n}} V_N(\mathbf{d}, \mathbf{r})  = \min_{\mathbf{d}, \mathbf{r}, \underline{n} + |\mathcal{J}| \le \sum_{i=1}^N r_i \le \bar{n}} \max_{e \in \{1, \cdots, E\}} \\ &\frac{1}{(\sum_{i=1}^N r_i)^2}\Big(  \sum_{i \in \mathcal{J}_c} r_i w^2(i, \mathbf{d}, \mathbf{r}) \sigma^2(i, \mathbf{d}) + \sum_{j \in \mathcal{N}_i  } r_i r_j w(i,  \mathbf{d},  \mathbf{r}) w(j, \mathbf{d}, \mathbf{r}) \eta(i,j, \mathbf{d})  \\ &+  
\sum_{i \in \mathcal{J}} r_i w^e(i,  \mathbf{R}, \mathbf{D}^R)^2 \sigma^2(i, \mathbf{d}) + \sum_{j \in \mathcal{N}_i} r_i r_j w^e(i,  \mathbf{d}, \mathbf{r}) w^e(j,  \mathbf{d}, \mathbf{r}) \eta(i,j, \mathbf{d}) \Big) \\
&\ge (A) + (B)
\end{aligned} 
\end{equation} 
where 
$$
\small 
\begin{aligned} 
(A) &= \min_{\mathbf{d}, \mathbf{r},  \underline{n} + |\mathcal{J}| \le \sum_{i=1}^N r_i \le \bar{n}} \max_{e \in \{1, \cdots, E\}} \frac{1}{(\sum_{i=1}^N r_i)^2}  \Big(\sum_{i \in \mathcal{J}_c} r_i w^e(i,  \mathbf{d}, \mathbf{r})^2 \sigma^2(i, \mathbf{d}) \\ &+ \sum_{j \in \mathcal{N}_i  } r_i r_j w^e(i,  \mathbf{d}, \mathbf{r}) w^e(j, \mathbf{d}, \mathbf{r}) \eta(i,j, \mathbf{d})\Big)  \\ (B) &=    \min_{\mathbf{d}, \mathbf{r},  \underline{n} + |\mathcal{J}| \le \sum_{i=1}^N r_i \le \bar{n}} \min_{e \in \{1, \cdots, E\}}
\frac{1}{(\sum_{i=1}^N r_i)^2}  \Big(\sum_{i \in \mathcal{J}} r_i w^e(i,  \mathbf{d}, \mathbf{r})^2 \sigma^2(i, \mathbf{d}) \\ &+ \sum_{j \in \mathcal{N}_i} r_i r_j w^e(i,   \mathbf{d}, \mathbf{r}) w^e(j,  \mathbf{d}, \mathbf{r}) \eta(i,j, \mathbf{d}) \Big)
\end{aligned} 
$$ 
where we decomposed the sum into the sum over two sets and flip the $\max_{w}$ with the $\min_{w}$ for one of the two sets to obtain a lower bound. Such sets are defined as 
$$
\mathcal{J}_c = \{1, \cdots, N\} \setminus \mathcal{J}, \quad \mathcal{J} = \bigcup_{i:P_i = 1} \{i\} \cup \mathcal{N}_i. 
$$
\paragraph{$(iii)$ Part two: lower bound decomposed into two groups $\mathcal{J}_c, \mathcal{J}$} We now analyze each component in the right hand side of Equation \eqref{eqn:hh}. 
Notice now that the following term
\begin{equation}
\begin{aligned}  
 (B) &\ge \min_{\mathbf{d}, \mathbf{r}, \underline{n} + |\mathcal{J}| \le \sum_{i=1}^N r_i \le \bar{n}} \min_{ e \in \{1, \cdots, E\} } \frac{1}{(\sum_{i=1}^N r_i)^2}  \sum_{i \in \mathcal{J}} \sum_{j \in \mathcal{N}_i} r_i r_j w^e(i,  \mathbf{d}, \mathbf{r}) w^e(j,  \mathbf{d}, \mathbf{r}) \eta(i,j, \mathbf{d}) \\ &\ge - \bar{C} |\mathcal{J}| \max_{i \in \mathcal{J}} |\mathcal{N}_i| /(\underline{n} + |\mathcal{J}|)^2
 \end{aligned} 
\end{equation} 
since the second moment and weights are bounded by Assumption \ref{ass:moment}, for a universal constant $\bar{C} < \infty$. 
 Therefore, the following holds:
\begin{equation} \label{eqn:hh2}
\begin{aligned} 
\eqref{eqn:hh} &\ge (A) - \bar{C} |\mathcal{J}| \max_{i \in \mathcal{J}} |\mathcal{N}_i| /(\underline{n} + |\mathcal{J}|)^2.
\end{aligned}   
\end{equation} 
\paragraph{$(iii)$ Part 3: Lower bound to $(A)$} Next, we provide a lower bound to $(A)$. 
We have 
\begin{equation} \label{eqn:hh4}
\begin{aligned} 
(A) &\ge \min_{\mathbf{d}, \mathbf{r}, \underline{n} + |\mathcal{J}| \le \sum_{i} r_i \le \bar{n}} \max_{e \in \{1, \cdots, E\}} \Big(\frac{1}{(\sum_{i \in \mathcal{J}_c} r_i + |\mathcal{J}|)^2} \sum_{i \in \mathcal{J}_c} r_i w^e(i,  \mathbf{d}, \mathbf{r})^2 \sigma^2(i,\mathbf{d}) \\ &+ \sum_{j \in \mathcal{N}_i} r_i r_j w^e(i,  \mathbf{d},  \mathbf{r}) w^e(j,  \mathbf{d}, \mathbf{r}) \eta(i,j, \mathbf{d}) \Big):= (J), 
\end{aligned}   
\end{equation} 
where we replaced in the denominator $\sum_{i \in \mathcal{J}} r_i$ with $|\mathcal{J}|$. 
Therefore, we can write 
\begin{equation} \label{eqn:helper5}
\begin{aligned} 
(iii) \le &V_N(\tilde{\mathbf{D}}, \tilde{\mathbf{R}}) - \min_{\mathbf{d}, \mathbf{r}, \underline{n} + |\mathcal{J}| \le \sum_i r_i \le \bar{n}} \max_{e \in \{1, \cdots, E\}} \frac{1}{(\sum_{i \in \mathcal{J}_c} r_i + |\mathcal{J}|)^2} \Big(\sum_{i \in \mathcal{J}_c} r_i w^e(i, \mathbf{d}, \mathbf{r})^2 \sigma^2(i,\mathbf{d}) \\ &+ \sum_{j \in \mathcal{N}_i} r_i r_j w^e(i, \mathbf{d}, \mathbf{r}) w^e(j, \mathbf{d}, \mathbf{r}) \eta(i,j, \mathbf{d}) \Big) + \bar{C} |\mathcal{J}| \max_{i \in \mathcal{J}} |\mathcal{N}_i| /(\underline{n} + |\mathcal{J}|)^2. 
\end{aligned} 
\end{equation}  
\paragraph{$(iii)$ Part 4: Upper bound to $V_N(\tilde{\mathbf{D}}, \tilde{\mathbf{R}})$} 
Consider now the right-hand side in Equation \eqref{eqn:hh4}, defined as $(J)$. Observe, that we can write 
$$
\begin{aligned} 
(J) =  & (L) + (M)
\end{aligned} 
$$ 
where 
\begin{equation} \label{eqn:M} 
\small 
\begin{aligned} 
(L) = & \min_{\mathbf{d}, \mathbf{r},  \underline{n} \le \sum_{i \in \mathcal{J}_c} r_i \le \bar{n}, r_i = 0 \forall i \in \mathcal{J}} \max_{e \in \{1, \cdots, E\}} \frac{1}{(\sum_{i \in \mathcal{J}_c} r_i + |\mathcal{J}|)^2} \Big(\sum_{i \in \mathcal{J}_c} r_i w^e(i, \mathbf{d}, \mathbf{r})^2 \sigma^2(i, \mathbf{d}) \\ &+ \sum_{j \in \mathcal{N}_i} r_i r_j w^e(i, \mathbf{d},  \mathbf{r}) w^e(j, \mathbf{d}, \mathbf{r}) \eta(i,j, \mathbf{d}) \Big)  \\ 
(M) =  & 
\min_{\mathbf{r},  \underline{n} + |\mathcal{J}| \le \sum_ir_i \le \bar{n}} \max_{e \in \{1, \cdots, E\}} \frac{1}{(\sum_{i \in \mathcal{J}_c} r_i + |\mathcal{J}|)^2} \Big(\sum_{i \in \mathcal{J}_c} r_i w^e(i, \mathbf{d}, \mathbf{r})^2 \sigma^2(i, \mathbf{d}) \\ &+ \sum_{j \in \mathcal{N}_i} r_i r_j w^e(i, \mathbf{d}, \mathbf{r}) w^e(j, \mathbf{d}, \mathbf{r}) \eta(i,j, \mathbf{d}) \Big) - (L). 
\end{aligned} 
\end{equation}  
where we added and subctracted $(L)$ which contains the condition $r_i = 0 \forall i \in \mathcal{J}$, and $ \underline{n} \le \sum_{i \in \mathcal{J}_c} r_i \le \bar{n}$. We study $(L)$ first. 
 Define 
$$
\small 
\begin{aligned} 
(\mathbf{D}^{**}, \mathbf{R}^{**}) \in \mathrm{arg} &\min_{\mathbf{d}, \mathbf{r}, \underline{n} \le \sum_{i \in \mathcal{J}_c} r_i \le \bar{n}, r_i = 0 \forall i \in \mathcal{J}} \max_{e \in \{1, \cdots, E\}} \frac{1}{(\sum_{i \in \mathcal{J}_c} r_i + |\mathcal{J}|)^2} \Big(\sum_{i \in \mathcal{J}_c} r_i w^e(i, \mathbf{d}, \mathbf{r})^2 \sigma^2(i, \mathbf{d}) \\ &+ \sum_{j \in \mathcal{N}_i} r_i r_j w^e(i,  \mathbf{d}, \mathbf{r}) w^e(j,\mathbf{d}, \mathbf{r}) \eta(i,j, \mathbf{d}) \Big).
\end{aligned} 
$$ 
Observe that by construction $V_N(\tilde{\mathbf{D}}, \tilde{\mathbf{R}}) \le V_N( \mathbf{D}^{**}, \mathbf{R}^{**})$, since $\tilde{\mathbf{D}}, \tilde{\mathbf{R}}$ minimize $V_N(\cdot)$, and since $\mathbf{R}^{**}$ satisfy the constraints in Equation \eqref{eqn:kk}. Therefore, we can write 
\begin{equation} \label{eqn:left2} 
\small
\begin{aligned} 
\eqref{eqn:helper5} \le &V_N(\mathbf{D}^{**}, \mathbf{R}^{**}) -  \max_{e \in \{1, \cdots, E\}} \frac{1}{(\sum_{i \in \mathcal{J}_c} R_i^{**} + |\mathcal{J}|)^2} \Big(\sum_{i \in \mathcal{J}_c} R_i^{**} w^e(i,\mathbf{D}^{**},  \mathbf{R}^{**})^2 \sigma^2(i,\mathbf{D}^{**}) \\ &+ \sum_{j \in \mathcal{N}_i } R_i^{**} R_j^{**} w^e(i, \mathbf{D}^{**}, \mathbf{R}^{**}) w^e(j, \mathbf{D}^{**}, \mathbf{R}^{**}) \eta(i,j, \mathbf{D}^{**}) \Big)  \\ & + \bar{C} |\mathcal{J}| \max_{i \in \mathcal{J}} |\mathcal{N}_i| /(\underline{n} + |\mathcal{J}|)^2 - (M). 
\end{aligned} 
\end{equation}   

\paragraph{$(iii)$ Part 4: first Term in the right-hand side of Equation \eqref{eqn:left2}} By simple algebra, and using the same argument for the weights used for $(i)$, we obtain, 
\begin{equation} \label{eqn:aajj} 
\begin{aligned} 
& V_N(\mathbf{D}^{**}, \mathbf{R}^{**}) -  \max_{e \in \{1, \cdots, E\}} \frac{1}{(\sum_{i \in \mathcal{J}_c} R_i^{**} + |\mathcal{J}|)^2} \Big(\sum_{i \in \mathcal{J}_c} R_i^{**} w^e(i, \mathbf{D}^{**}, \mathbf{R}^{**})^2 \sigma^2(i,\mathbf{D}^{**}) \\ &+ \sum_{j \in \mathcal{N}_i} R_i^{**} R_j^{**} w^e(i,\mathbf{D}^{**}, \mathbf{R}^{**}) w^e(j, \mathbf{D}^{**}, \mathbf{R}^{**}) \eta(i,j, \mathbf{D}^{**}) \Big) \\
&\le \max_{e \in \{1, \cdots, E\}} \Big| (\frac{1}{(\sum_{i \in \mathcal{J}_c} R_i^{**} + |\mathcal{J}|)^2} - \frac{1}{(\sum_{i \in \mathcal{J}_c} R_i^{**})^2}) \Big(\sum_{i \in \mathcal{J}_c} R_i^{**} w^e(i, \mathbf{D}^{**}, \mathbf{R}^{**})^2 \sigma^2(i,\mathbf{D}^{**}) \\ &+ \sum_{j \in \mathcal{N}_i } R_i^{**} R_j^{**} w^e(i, \mathbf{D}^{**}, \mathbf{R}^{**}) w^e(j, \mathbf{D}^{**}, \mathbf{R}^{**}) \eta(i,j, \mathbf{D}^{**}) \Big) \Big|. 
\end{aligned} 
\end{equation} 

\paragraph{$(iii)$ Part 5: bound on $(iii)$ which also depends on $(M)$} 
By Assumption \ref{ass:moment} (bounded outcome), for $(iii)$ as in Equation \eqref{eqn:iii}, using Equation \eqref{eqn:aajj} we can write
\begin{equation} 
\begin{aligned} 
(iii) &\le \bar{C} n \mathcal{N}_{\max} \frac{n|\mathcal{J}| + |\mathcal{J}|^2}{(\sum_{i=1}^N R_i^{**})^4} + \bar{C} |\mathcal{J}| \max_{i \in \mathcal{J}} |\mathcal{N}_i| /(\underline{n} + |\mathcal{J}|)^2 - (M) \\ &\le \bar{C}  \mathcal{N}_{\max} \frac{n^2|\mathcal{J}| + n|\mathcal{J}|^2}{\alpha^4 n^4} + \bar{C} |\mathcal{J}| \mathcal{N}_{\max} /(\underline{n} + |\mathcal{J}|)^2 - (M) 
\end{aligned} 
\end{equation} 
for a finite constant $\bar{C} < \infty$. Notice now that $|\mathcal{J}| \le  |\mathcal{N}_{\max}| \times m$ which implies that the above term is $\mathcal{O}(\mathcal{N}_{\max}^2 m/\underline{n}^2 + \mathcal{N}_{\max}^3 m^2/\underline{n}^3) + \mathcal{O}(|(M)|)$. Here $\mathcal{N}_{\max}^3 m^2/\underline{n}^3 = \mathcal{O}(\mathcal{N}_{\max}^2 m/\underline{n}^2)$ since under the assumptions $\underline{n} \ge \mathcal{N}_{\max} m/(\alpha - 1)$ for $\alpha > 1$. 

\paragraph{Bound on $(M)$} We are left to provide a bound for $(M)$. The bound on $(M)$ follows from the stability assumption (Assumption \ref{ass:weights_stable}). In particular, let 
$$
\begin{aligned} 
(\mathbf{d}^*, \mathbf{r}^*) \in \arg \min_{\mathbf{r},  \underline{n} + |\mathcal{J}| \le \sum_{i \in \mathcal{J}_c} r_i \le \bar{n}} \max_{e \in \{1, \cdots, E\}} & \frac{1}{(\sum_{i \in \mathcal{J}_c} r_i + |\mathcal{J}|)^2} \Big(\sum_{i \in \mathcal{J}_c} r_i w^e(i, \mathbf{d}, \mathbf{r})^2 \sigma^2(i, \mathbf{d}) \\ &+ \sum_{j \in \mathcal{N}_i} r_i r_j w^e(i, \mathbf{d}, \mathbf{r}) w^e(j, \mathbf{d}, \mathbf{r}) \eta(i,j, \mathbf{d}) \Big), 
\end{aligned} 
$$and $\tilde{\mathbf{r}}^*$ be such that $\mathbf{r}^*_j = \tilde{\mathbf{r}}_j^*, j \not \in \mathcal{J}$ and $\tilde{\mathbf{r}}_j^* = 0$ otherwise. 
We note that $(\mathbf{d}^*, \tilde{\mathbf{r}}_j^*)$ is a feasible solution to minimize $(L)$, since $(L)$ contains a slacker constraint $\underline{n} \le \sum_{j \in \mathcal{J}_c} r_i \le \bar{n}$, while $\underline{n}  + |\mathcal{J}| \le \sum_i r_j^* \le \bar{n}$.
 Define 
 $$
 \mathcal{J}_3 = \mathcal{J}_c \setminus \{j: \mathcal{N}_j \cap \mathcal{J} \neq \emptyset\}
 $$ 
 the set of individuals in $\mathcal{J}_c$ without friends in $\mathcal{J}$. We can then write 
\begin{equation} \label{eqn:bound_helper} 
\small 
\begin{aligned} 
|(M)| \le &  \max_{e \in \{1, \cdots, E\}} \Big|\frac{1}{(\sum_{i \in \mathcal{J}_c} r_i^* + |\mathcal{J}|)^2} \Big(\sum_{i \in \mathcal{J}_3} r_i^* (w^e(i,  \mathbf{d}^*, \mathbf{r}^*)^2 - w^e(i, \mathbf{d}^*,  \tilde{\mathbf{r}}^*)^2) \sigma^2(i, \mathbf{d}^*) \\ &+ \sum_{j \in \mathcal{N}_i} r_i^* r_j^* (w^e(i, \mathbf{d}^*,  \mathbf{r}^*) w^e(j, \mathbf{d}^*, \mathbf{r}^*) - w^e(i,  \mathbf{d}^*, \tilde{\mathbf{r}}^*) w^e(j,  \mathbf{d}^*, \tilde{\mathbf{r}}^*)) \eta(i,j, \mathbf{d}^* ) \Big) \Big| \\
&+ \max_{e \in \{1, \cdots, E\}} \Big| \frac{1}{(\sum_{i \in \mathcal{J}_c} r_i^* + |\mathcal{J}|)^2} \sum_{i \in \mathcal{J}_c} \sum_{j \in \mathcal{N}_i \cap \mathcal{J}} r_i^* r_j^* w^e(i, \mathbf{d}^*, \mathbf{r}^*) w^e(j, \mathbf{d}^* , \mathbf{r}^*) \eta(i,j, \mathbf{d}^*) \Big| \\ 
&+ \max_{e \in \{1, \cdots, E\}} \Big| \frac{1}{(\sum_{i \in \mathcal{J}_c} r_i^* + |\mathcal{J}|)^2} \sum_{i \in \mathcal{J}_c \setminus \mathcal{J}_3} r_i^* \Big(w^e(i, \mathbf{d}^*, \mathbf{r}^*)^2  - w^e(i, \mathbf{d}^* , \tilde{\mathbf{r}}^*)^2  \Big)  \sigma^2(i,\mathbf{d}^*)\Big|.
\end{aligned} 
\end{equation}  
The first term sums over variances and covariances of each individual in $\mathcal{J}_3$ and excludes the individuals in $\mathcal{J}_c$ which are friends with individuals in $\mathcal{J}$. The second term sums over the covariances of individuals in $\mathcal{J}_c$ and individuals in $\mathcal{J}$ friends with individuals in $\mathcal{J}_c$. Such covariances multiply by $r_i^* r_j^* w(i, \mathbf{d}^*, \mathbf{r}^*) w(j, \mathbf{d}^* , \mathbf{r}^*) \eta(i,j, \mathbf{d}^*)$ only (and not also $\tilde{r}_i^* \tilde{r}_j^* w(i, \mathbf{d}^*, \tilde{\mathbf{r}}^*) w(j, \mathbf{d}^* , \tilde{\mathbf{r}}^*) \eta(i,j, \mathbf{d}^*)$) since $\tilde{r}_j^*$ is zero for individuals in $\mathcal{J}$. The last term sums over the variances of individuals in $\mathcal{J}_c$ which are not in $\mathcal{J}_3$, for which $r_i^* = \tilde{r}_i^*$ by construction. 

 Using Assumption \ref{ass:weights_stable} (since $\mathbf{r}^*$ and $\tilde{\mathbf{r}}^*$ only differ by $\mathcal{N}_{\max} \bar{m}$ units at most) we obtain that the first component in the bound in Equation \eqref{eqn:bound_helper} is $\mathcal{O}(\mathcal{N}_{\max}^2 \bar{m}/\underline{n}^2)$. For the second component, note that individuals in $\mathcal{J}$ have at most $\mathcal{N}_{\max} |\mathcal{J}| \le \mathcal{N}_{\max}^2 \bar{m}$ many connections. Since $\mathcal{J}_c$ and $\mathcal{J}$ are disjoint sets (and by symmetry of $\mathbf{A}$), the second term is at most $\mathcal{O}(\mathcal{N}_{\max}^2 \bar{m}/\underline{n}^2)$.  Similarly, for the third term, there are at most $\mathcal{N}_{\max} \times |\mathcal{J}|$ many edges between $\mathcal{J}_c$ and $\mathcal{J}$, which implies that, by symmetry of $\mathbf{A}$, $|\mathcal{J}_c \setminus \mathcal{J}_3| \le |\mathcal{J}| \mathcal{N}_{\max} \le \bar{m} \mathcal{N}_{\max}^2$ which completes the proof, since weights and variances are uniformly bounded.  

\paragraph{Case with $\zeta > 0$} Consider now the case where $\zeta > 0$. Then returning to Equation \eqref{eqn:iii}, we can write 
$$
\begin{aligned} 
 & \max_{e \in \{1, \cdots, E\}} \mathbb{V}\Big(\hat{\Gamma}(w^e) | \mathbf{R}, \mathbf{D}^R, \mathbf{A}, \mathbf{T}\Big) - \mathbb{V}_{\bar{n}}^\star(E) \\ 
 &=  \underbrace{\max_{e \in \{1, \cdots, E\}} \mathbb{V}\Big(\hat{\Gamma}(w^e) | \mathbf{R}, \mathbf{D}^R, \mathbf{A}, \mathbf{T}\Big) - \max_{e \in \{1, \cdots, E\}} \mathbb{V}\Big(\hat{\Gamma}(w^e) | \mathbf{R}, \check{\mathbf{D}}^R, \mathbf{A}, \mathbf{T}\Big)}_{(I_1)} \\ &+ \underbrace{\max_{e \in \{1, \cdots, E\}} \mathbb{V}\Big(\hat{\Gamma}(w^e) | \mathbf{R}, \check{\mathbf{D}}^R, \mathbf{A}, \mathbf{T}\Big) -  \mathbb{V}_{\bar{n}}^\star(E)}_{(I_2)}
 \end{aligned} 
 $$ 
 where $\check{\mathbf{D}}^R$ is as defined in Equation \eqref{eqn:design1b} corresponding to the minimizer for $\zeta = 0$. Therefore, we have that  
$I_2$ is as bounded as for the case for $\zeta = 0$. To bound $I_1$ we canb write 
$$
\begin{aligned} 
(I_1) = & \underbrace{\max_{e \in \{1, \cdots, E\}} \mathbb{V}\Big(\hat{\Gamma}(w^e) | \mathbf{R}, \mathbf{D}^R, \mathbf{A}, \mathbf{T}\Big) - \hat{V}_p(\mathbf{D}, \mathbf{R})}_{(H_1)} \\ &+\underbrace{ \hat{V}_p(\mathbf{D}, \mathbf{R}) - \hat{V}_p(\check{\mathbf{D}},\mathbf{R})}_{(H_2)} \\ &+ \underbrace{\hat{V}_p(\check{\mathbf{D}},\mathbf{R}) - \max_{e \in \{1, \cdots, E\}} \mathbb{V}\Big(\hat{\Gamma}(w^e) | \mathbf{R}, \check{\mathbf{D}}^R, \mathbf{A}, \mathbf{T}\Big)}_{(H_3)}
\end{aligned} 
$$ 
By construction of Algorithm \ref{alg:3}, we have $(H_2) \le \zeta/\bar{n}$. Finally, we can bound $(H_1)$ and $(H_3)$ following verbatim the argument for $(i)$ and $(ii)$ in Equation \eqref{eqn:i_ii} completing the proof.

\subsection{Inference} \label{sec:asymp_app}

\begin{lem}\label{lem:locally_dep} \citep{ross2011fundamentals} Let $X_1, ..., X_n$ be random variables such that $\mathbb{E}[X_i^4] < \infty$, $\mathbb{E}[X_i] = 0$, $\sigma^2 = \mathrm{Var}(\sum_{i = 1}^n X_i)$ and define $W = \sum_{i = 1}^n X_i/\sigma$. Let the collection $(X_1, ..., X_n)$ have dependency neighborhoods $\mathcal{N}_i$, $i = 1, ..., n$ and also define $D = \mathrm{max}_{1 \le i \le n} |\mathcal{N}_i|$. Then for $Z$ a standard normal random variable, we obtain 
\begin{equation}
d_W(W, Z) \le \frac{D^2}{\sigma^3} \sum_{i = 1}^n \mathbb{E}| X_i|^3 + \frac{\sqrt{28} D^{3/2}}{\sqrt{\pi} \sigma^2} \sqrt{\sum_{i = 1}^n \mathbb{E}[X_i^4]},
\end{equation}  
where $d_W$ denotes the Wesserstein metric.
\end{lem}

\begin{thm} \label{thm:asym}
Suppose that Assumption \ref{ass:model}, \ref{ass:modelb}, \ref{ass:sparsity}, \ref{ass:moment}, \ref{ass:weights_stable}(ii), \ref{ass:inference_conditions}(i) hold. Let $n = \mathbf{1}^\top \mathbf{R} \propto \bar{n} \propto N$. Let $n \mathbb{V}(\hat{\Gamma} | \mathbf{R}, \mathbf{D}^R, \mathbf{A}, \mathbf{T}) > 0$ almost surely. Then 
\begin{equation} \label{eqn:thm1}
\frac{\sqrt{n} (\widehat{\Gamma} - \mathbb{E}[\widehat{\Gamma} | \mathbf{R}, \mathbf{D}^R, \mathbf{A}, \mathbf{T}] )}{\sqrt{n \mathbb{V}(\hat{\Gamma} | \mathbf{R}, \mathbf{D}^R, \mathbf{A}, \mathbf{T})}} \rightarrow_d \mathcal{N}(0,1). 
\end{equation} 
\end{thm}

\begin{proof}[Proof of Theorem \ref{thm:asym}]
We prove asymptotic normality after conditioning on the sigma algebra $\sigma(\mathbf{D}^R,\mathbf{A}, \mathbf{R}, \mathbf{T}, \mathbf{P})$. Notice that $\mathbb{E}[\widehat{\Gamma} | \mathbf{R}, \mathbf{D}^R, \mathbf{A}, \mathbf{T}] = \mathbb{E}[\widehat{\Gamma} | \mathbf{R}, \mathbf{D}^R, \mathbf{A}, \mathbf{T}, \mathbf{P}]$ by Proposition \ref{ass:cov}. Next, we show that $Y_i$ for all $i: R_i = 1$ are locally dependent, given $\sigma(\mathbf{A}, \mathbf{R}, \mathbf{D}^R, \mathbf{T}, \mathbf{P})$. To show this, it suffices to show that $$\{\varepsilon_i\}_{i: R_i = 1} \Big| \sigma(\mathbf{A}, \mathbf{R}, \mathbf{D}^R, \mathbf{T}, \mathbf{P})$$ are locally dependent, i.e., form a local dependency graph as described in \cite{ross2011fundamentals}. Define 
$$
\mathcal{H} = \{1, \cdots, N\} \setminus \mathcal{J}. 
$$  

The argument is the following. Under Assumption \ref{ass:modelb}, unobservables are locally dependent given the adjacency matrix $\mathbf{A}$ and covariates $\mathbf{T}$. Since $\mathbf{P}$ is exogenous conditional on $\mathbf{A}$, it follows that unobservables are locally dependent given $(\mathbf{A}, \mathbf{T}, \mathbf{P})$ 
 That is, 
$$
\varepsilon_{1, \cdots, N} \Big| \sigma(\mathbf{A}, \mathbf{P}, \mathbf{T})
$$
are locally dependent.  
 Consider now the distribution of all unobservables in the set $\mathcal{H}$, given $\mathbf{A}, \mathbf{P}, \mathbf{T}$. Here, unobservables are mutually independent on $\mathbf{R}, \mathbf{D}^R$, given $\sigma(\mathbf{A}, \mathbf{P}, \mathbf{T})$ and $\mathcal{H}$ is measurable with respect to $\mathbf{P}$. Therefore, 
 $$
 \varepsilon_{i \in \mathcal{H}} \Big| \sigma(\mathbf{A}, \mathbf{P}, \mathbf{R}, \mathbf{D}^R, \mathbf{T})
 $$
 are locally dependent. 
 Since $\{i: R_i = 1\} \subseteq \mathcal{H}$ the local dependence assumption of unobservables in such a set holds conditional on $\mathbf{A}, \mathbf{P}, \mathbf{R}, \mathbf{D}^R, \mathbf{T}$ for such units. 
 
Recall that by Assumption \ref{ass:model}
\begin{equation}
Y_i = r\Big(D_i, g_i(D_{\mathcal{N}_i}), T_i, \varepsilon_i\Big),   
\end{equation} 
where $g_i(D_{\mathcal{N}_i})$ is a deterministic function of $\mathbf{A}, \mathbf{D}_{\mathcal{N}_i}$ and $\mathbf{T}$. 
Therefore, given $\sigma(\mathbf{A}, \mathbf{P}, \mathbf{R}, \mathbf{D}^R, \mathbf{T})$ outcomes $\{Y_i\}_{i:R_i = 1}$ are locally dependent. Let
\begin{equation} \label{eqn:x_i}
X_i:= \frac{1}{n \sqrt{ \mathbb{V}(\hat{\Gamma} | \mathbf{R}, \mathbf{D}^R, \mathbf{A}, \mathbf{T})}} w_{\mathbf{A}, \mathbf{T}}(i, \mathbf{R}, \mathbf{D}^R) \Big(Y_i - m(D_i, g_i(D_{\mathcal{N}_i}), T_i)\Big). 
\end{equation} 
By Proposition \ref{ass:cov},  we have 
\begin{equation}
\mathbb{E}[X_i | \sigma(\mathbf{D}^R,\mathbf{A}, \mathbf{R}, \mathbf{P}, \mathbf{T})] = 0.  
\end{equation} 
To prove the theorem we invoke Lemma \ref{lem:locally_dep}. In particular, we observe that for $Z \sim \mathcal{N}(0,1)$, we have 
\begin{equation} 
\sup_{x \in \mathbb{R}} \Big|P\Big(\sum_{i: R_i = 1} X_i \le x \Big | \sigma(\mathbf{D}^R,\mathbf{A}, \mathbf{R}, \mathbf{P}, \mathbf{T})\Big) - \Phi(x)\Big| \le c \sqrt{d_{W|\sigma(\mathbf{D}^R,\mathbf{A}, \mathbf{R}, \mathbf{P}, \mathbf{T})}(\sum_{i: R_i = 1} X_i,Z)}. 
\end{equation} 
where $d_{W|\sigma(\mathbf{D}^R,\mathbf{A}, \mathbf{R}, \mathbf{P}, \mathbf{T})}(\sum_{i: R_i = 1} X_i,Z)$ denotes the Wesserstein metric taken with respect to the conditional marginal distribution of $\sum_{i: R_i = 1} X_i$ given $\sigma(\mathbf{D}^R,\mathbf{A}, \mathbf{R}, \mathbf{P}, \mathbf{T})$ and $\Phi(x)$ is the CDF of a standard normal distribution, and $c <\infty$ is a universal constant. 
To apply Lemma \ref{lem:locally_dep} we take $\sigma^2 = 1$ since $X_i$ already contains the rescaling factor defined in Lemma \ref{lem:locally_dep}. In addition, since $n \mathbb{V}(\hat{\Gamma} | \mathbf{R}, \mathbf{D}^R, \mathbf{A}, \mathbf{T})$ is strictly bounded away from zero we obtain under Assumption \ref{ass:moment}
\begin{equation}
\mathbb{E}[X_i^4 | \sigma(\mathbf{D}^R,\mathbf{A}, \mathbf{R}, \mathbf{P}, \mathbf{T})] \le \bar{C} \frac{1}{n^2}, \quad 
\mathbb{E}[X_i^3 | \sigma(\mathbf{D}^R,\mathbf{A}, \mathbf{R}, \mathbf{P}, \mathbf{T})] \le \bar{C} \frac{1}{n^{3/2}}.     
\end{equation} 
Therefore, the condition in Lemma \ref{lem:locally_dep} are satisfied. Then we obtain 
\begin{equation} 
\begin{aligned} 
&d_{W|\sigma(\mathbf{D}^R,\mathbf{A}, \mathbf{R}, \mathbf{P}, \mathbf{T})}(\sum_{i: R_i = 1} X_i,Z) \le \mathcal{N}_{\max}^2 \sum_{i: R_i = 1} \mathbb{E}[| X_i|^3 | \mathbf{R}, \mathbf{D}^R,\mathbf{A}, \mathbf{P}, \mathbf{T}] \\ &+ \frac{\sqrt{28} \mathcal{N}_{\max}^{3/2}}{\sqrt{\pi}} \sqrt{\sum_{i: R_i = 1} \mathbb{E}[X_i^4| \mathbf{R},\mathbf{A}, \mathbf{D}^R, \mathbf{P}, \mathbf{T}]} \\ &\le \frac{\mathcal{N}_{\max}^2}{n^{1/2}}  \bar{C} + \frac{\sqrt{28} \mathcal{N}_{\max}^{3/2}}{\sqrt{\pi n}} \bar{C} 
\end{aligned} 
\end{equation} 
 for a universal constan $\bar{C} < \infty$.  Since by Assumption \ref{ass:sparsity}, $\mathcal{N}_{\max}^2/N^{1/2} = o(1)$, and $n \propto \bar{n} \propto N$, we obtain 
 \begin{equation}
 \sup_{x \in \mathbb{R}} \Big|P\Big(\sum_{i: R_i =1} X_i \le x \Big | \sigma(\mathbf{D}^R,\mathbf{A}, \mathbf{R}, \mathbf{P}, \mathbf{T})\Big) - \Phi(x)\Big| \le  \sqrt{\frac{\mathcal{N}_{\max}^2}{n^{1/2}}  \bar{C} + \frac{\sqrt{28} \mathcal{N}_{\max}^{3/2}}{\sqrt{\pi n}} \bar{C}}  = o(1).
 \end{equation} 
  To prove that the result also holds unconditionally, we may notice that for some arbitrary measure $\mu_N$,  
 \begin{equation}
 \begin{aligned} 
 &\sup_{x \in \mathbb{R}} \Big|\int P\Big(\sum_{i: R_i = 1} X_i \le x \Big | \sigma(\mathbf{D}^R,\mathbf{A}, \mathbf{R}, \mathbf{P}, \mathbf{T})\Big) d\mu_N - \Phi(x)\Big| \\ &\le   \sup_{x \in \mathbb{R}} \int \Big| P\Big(\sum_{i: R_i = 1} X_i \le x \Big | \sigma(\mathbf{D}^R,\mathbf{A}, \mathbf{R}, \mathbf{P}, \mathbf{T})\Big)  - \Phi(x)\Big| d\mu_N \\ &\le \int \sup_{x \in \mathbb{R}}  \Big| P\Big(\sum_{i: R_I = 1} X_i \le x \Big | \sigma(\mathbf{D}^R,\mathbf{A}, \mathbf{R}, \mathbf{P}, \mathbf{T})\Big)  - \Phi(x)\Big| d\mu_N = o(1).  
\end{aligned} 
 \end{equation} 
 This concludes the proof. 
\end{proof}

\begin{cor} Theorem \ref{thm:multipledependence} holds.
\end{cor} 
\begin{proof} 
The proof follows similarly to the above theorem with an important modification. We observe that the variables $X_i$ in Equation \eqref{eqn:x_i} do not follow a dependence graph since they exhibit $M$ degree dependence. Instead, we construct a graph where two individuals are connected if they are connected by at most $M$ edges in the original graph. In such a graph, the variables $X_i$ as defined in Equation \eqref{eqn:x_i} satisfy the local dependence assumption in Lemma \ref{lem:locally_dep}. In order for the lemma to apply, we need to show that the maximum degree of such a graph, denoted as $\bar{\mathcal{N}}_M^2$ satisfies the condition $\bar{\mathcal{N}}_M^2/n^{1/2} = o(1)$. This follows under Assumption \ref{ass:higherorder}, since the maximum degree is uniformly bounded. This completes the proof. 
\end{proof}

 \begin{thm} \label{thm:variance} Let Assumptions  \ref{ass:model}, \ref{ass:modelb} \ref{ass:sparsity}, \ref{ass:moment}, \ref{ass:weights_stable}(ii), \ref{ass:inference_conditions} hold. Let $n = \mathbf{1}^\top \mathbf{R}$. Then  
 \begin{equation} 
 \frac{n \widehat{V} }{ n \mathbb{V}(\hat{\Gamma} | \mathbf{R}, \mathbf{D}^R, \mathbf{A}, \mathbf{T}) }   - 1 \rightarrow_p 0, 
 \end{equation} 
 with $\widehat{V}$ defined in Equation \eqref{eqn:final_vv}. 
 \end{thm} 
 
\begin{proof}[Proof of Theorem \ref{thm:variance}]

Let $\mu_i = m(D_i, g_i(D_{\mathcal{N}_i}), T_i)$. Similarly, denote $w_i = w_{\mathbf{A}, \mathbf{T}}(\mathbf{R}, \mathbf{D}^R)$. Because $n \mathbb{V}(\hat{\Gamma} | \mathbf{R}, \mathbf{D}^R, \mathbf{A}, \mathbf{T}) > 0$, it suffices to show that 
$$
\Big|n \widehat{V} - n \mathbb{V}(\hat{\Gamma} | \mathbf{R}, \mathbf{D}^R, \mathbf{A}, \mathbf{T})\Big| \rightarrow_p 0. 
$$ 
Note that under Assumption \ref{ass:modelb}, we can write 
$$
 n \mathbb{V}(\hat{\Gamma} | \mathbf{R}, \mathbf{D}^R, \mathbf{A}, \mathbf{T}) = \frac{1}{n} \sum_i  R_i w_i \sum_{j \in \mathcal{N}_i \cup \{i\}} R_j w_j   \mathbb{E}[(Y_i - \mu_i)(Y_j - \mu_j) ]. 
$$ 
 Consider now $n \widehat{V}$. We can write 
 $$
 n \widehat{V} = \frac{1}{n} \sum_{i=1}^N R_i w_i \sum_{j \in \mathcal{N}_i \cup \{i\}} R_j w_j (Y_i - \hat{m}_i)(Y_j - \hat{m}_j). 
 $$ 
 Also, define 
 $$
 n \tilde{V} = \frac{1}{n} \sum_{i=1}^N R_i w_i \sum_{j \in \mathcal{N}_i \cup \{i\}} R_j w_j (Y_i - \mu_i)(Y_j - \mu_j). 
 $$ 
 We can write 
 $$
 \footnotesize 
 \begin{aligned} 
 |n \widehat{V} - n \tilde{V}| & = |\frac{1}{n} \sum_{i=1}^N R_i w_i \sum_{j \in \mathcal{N}_i \cup \{i\}} R_j w_j (Y_i - \hat{m}_i + \mu_i - \mu_i)(Y_j - \hat{m}_j) - n \tilde{V}| \\ 
 &\le |\frac{1}{n} \sum_{i=1}^N R_i w_i \sum_{j \in \mathcal{N}_i \cup \{i\}} R_j w_j (\mu_i - \hat{m}_i)(Y_j - \hat{m}_j)| + |\frac{1}{n} \sum_{i=1}^N R_i w_i \sum_{j \in \mathcal{N}_i \cup \{i\}} R_j w_j (Y_i - \mu_i)(Y_j - \hat{m}_j) - n \tilde{V}| \\ 
 &\le |\frac{1}{n} \sum_{i=1}^N R_i w_i \sum_{j \in \mathcal{N}_i \cup \{i\}} R_j w_j (\mu_i - \hat{m}_i)(Y_j - \hat{m}_j)| +  |\frac{1}{n} \sum_{i=1}^N R_i w_i \sum_{j \in \mathcal{N}_i \cup \{i\}} R_j w_j (Y_i - \mu_i)(\mu_j - \hat{m}_j)| \\ 
 &\lesssim \mathcal{N}_{\max}\max_i |\hat{m}_i - \mu_i| \le \mathcal{N}_{\max} \max_{s,t,d} |\hat{m}(d,s,t) - m(d,s,t)|  = o_p(1)
 \end{aligned} 
 $$ 
 where for the second and third inequality we used the triangular inequality, for the fourth inequality we used a simple bound combined with the assumption that $Y_i, \mu_i$ and $\hat{m}_i$ (Assumptions \ref{ass:moment} and \ref{ass:inference_conditions}(iii)) are uniformly bounded and the last equality follows from Assumption \ref{ass:inference_conditions}(ii) since $\max_{s,t,d} |\hat{m}(d,s,t) - m(d,s,t)| = o_p(\underline{n}^{-1/4})$ and Assumption \ref{ass:sparsity}, since $\mathcal{N}_{\max}/N^{1/4} = o(1)$ and $\underline{n} \propto \bar{n} \propto N$. 
 
 To complete the proof it suffices to show that $\mathbb{V}(n \tilde{V} | \mathbf{R}, \mathbf{D}^R, \mathbf{A}, \mathbf{T}) = o(1)$ so that $n\tilde{V}$ is consistent for $n \mathbb{V}(\hat{\Gamma} | \mathbf{R}, \mathbf{D}^R, \mathbf{A}, \mathbf{T})$ conditional on $\mathbf{R}, \mathbf{D}^R, \mathbf{A}, \mathbf{T}$. To show this, define 
 $$
 L_i = R_i w_i \sum_{j \in \mathcal{N}_i \cup \{i\}} R_j w_j (Y_i - \mu_i)(Y_j - \mu_j). 
 $$  Note that under Assumption \ref{ass:modelb}, $L_i \perp L_k | \mathbf{R}, \mathbf{D}^R, \mathbf{A}, \mathbf{T}$ if $k$ is \textit{neither} a friend of $i$ nor $k$ is not a friend of a friend of $i$. It follows that $i$ has at most $\mathcal{N}_{\max}^2 + \mathcal{N}_{\max} \le 2\mathcal{N}_{\max}^2$ many second-degree friends. 
 
 Therefore, we can write 
 $$
 \mathbb{V}(n \tilde{V} | \mathbf{R}, \mathbf{D}^R, \mathbf{A}, \mathbf{T}) = \mathbb{V}\Big(\frac{1}{n} \sum_{i=1}^N R_i L_i| \mathbf{R}, \mathbf{D}^R, \mathbf{A}, \mathbf{T}\Big) = \frac{1}{n^2} \sum_{i=1}^N R_i \sum_{j \in \mathcal{N}_i \cup \mathcal{N}_i^2 \cup \{i\}} R_j \mathrm{Cov}(L_i, L_j)
 $$ 
 where $\mathcal{N}_i^2$ denotes the second order degree friends. From Assumption \ref{ass:moment} and Assumption \ref{ass:weights_stable}(ii), $L_i \lesssim \mathcal{N}_{\max}$ so their $\mathrm{Cov}(L_i, L_j) \le \mathcal{N}_{\max}^2$. Therefore, we can write 
 $$
  \mathbb{V}(n \tilde{V} | \mathbf{R}, \mathbf{D}^R, \mathbf{A}, \mathbf{T}) \lesssim \frac{1}{n^2} n \mathcal{N}_{\max}^4= o(1)
 $$ 
  where the last equality follows from Assumption \ref{ass:sparsity}. The proof completes.  
  \end{proof} 
  
\begin{cor} Theorem \ref{thm:asym_t} holds. 
\end{cor} 
\begin{proof}
The proof follows from Theorem \ref{thm:asym} and Theorem \ref{thm:variance} by Slutsky theorem. 
\end{proof}

\section{Optimization: MILP for Difference in Means Estimators} \label{app:MILP}

In this sub-section we discuss the optimization algorithm for difference in means estimators. 

\begin{thm} Consider the difference in means estimator in Example \ref{exmp:diff_means} with $g_i(D_{\mathcal{N}_i}) = \sum_{k \in \mathcal{N}_i} D_k$. Then the optimization program in Equation \eqref{eqn:design1} can be written as a mixed-integer fractional linear program. 
\end{thm} 

Although NP-hard, mixed-integer fractional linear programs are typically solved by off-the-shelf optimization routines and admits convenient representations \citep{charnes1962programming}. Existing routines such as those implied by commercial software provide useful upper bounds on the error loss when considering stopping time. Therefore, researchers may solve the optimization program up-to an observed error loss reported by the optimization routine.\footnote{Note that to solve the optimization problem over multiple weights, we can add an auxiliary variables $\lambda$, and solve the following program 
\begin{equation} 
\min \lambda, \quad \lambda \ge f_{w} \forall w \in \{w^1, \cdots, w^E\}
\end{equation}
where $f_{w}$ denotes the linearized objective function for each weight $w$ described in the proof below.}

\begin{proof}

To show that the optimization problem admits a mixed-integer linear program formulation, we first introduce the following lemma, which follows similarly to what discussed in \cite{viviano2019policy}. 

\begin{lem} \label{prop:opt} Any function $f_i$ that depends on $D_i$ and $\sum_{k \in \mathcal{N}_i} D_k$ can be written as 
\begin{equation} 
f_i(D_i, \sum_{k \in  \mathcal{N}_i} D_k) = \sum_{h = 0}^{|\mathcal{N}_i|}(f_i(1, h) - f_i(0,h))u_{i,h}  +  (t_{i,h,1} + t_{i,h,2} - 1)f_i(0,h), 
\end{equation} 
where
$u_{i,h}, t_{i,h,1}, t_{i,h,2}$ are defined by the following linear inequalities. 
\begin{equation} \label{eqn:constraints} 
\small 
\begin{aligned}
&(A) \quad \frac{D_i + t_{i,h,1} + t_{i,h,2}}{3} - 1 < u_{i,h} \le  \frac{D_i + t_{i,h,1} + t_{i,h,2}}{3}, \quad u_{i,h} \in \{0,1\} \quad \forall h \in \{0, ..., |\mathcal{N}_i|\}, \\ 
&(B) \quad  \frac{(\sum_k A_{i,k} D_k - h)}{ |\mathcal{N}_i| + 1} < t_{i,h,1} \le  \frac{(\sum_k A_{i,k} D_k - h)}{ |\mathcal{N}_i|+1} + 1, \quad  t_{i,h,1} \in \{0,1\}, \quad \forall h \in \{0, ..., |\mathcal{N}_i|\} \\
&(C) \quad  \frac{( h -\sum_k A_{i,k} D_k)}{  |\mathcal{N}_i| + 1} < t_{i,h,2} \le \frac{(h - \sum_k A_{i,k} D_k)}{  |\mathcal{N}_i| + 1} + 1, \quad  t_{i,h,2} \in \{0,1\}, \quad \forall h \in  \{0, ....,|\mathcal{N}_i|\}.
\end{aligned} 
\end{equation} 

\end{lem}

\begin{proof}
We define the following variables: 
$$
t_{i,h,1} = 1\{\sum_{k \in \mathcal{N}_i} D_k \ge h\}, \quad  
t_{i,h,2} = 1\{\sum_{k \in \mathcal{N}_i} D_k \le h\}, \quad h \in  \{0, ....,|\mathcal{N}_i|\}.
$$

The first variable is one if at least $h$ neighbors are treated, and the second variable is one if at most $h$ neighbors are treated. 

 Since each unit has $|\mathcal{N}_i|$ neighbors and zero to $|\mathcal{N}_i|$ neighbors can be treated, there are in total $\sum_{i = 1}^n (2|\mathcal{N}_i|  + 2)$ of such variables.
 
  The variable $t_{i,h,1}$ can be equivalently be defined as 
\begin{equation}
\quad  \frac{(\sum_k A_{i,k} D_k - h)}{ |\mathcal{N}_i| + 1} < t_{i,h,1} \le  \frac{(\sum_k A_{i,k} D_k - h)}{ |\mathcal{N}_i|+1} + 1, \quad  t_{i,h,1} \in \{0,1\}.
\end{equation} 
The above equation holds for the following reason. Suppose that $h < \sum_k A_{i,k} D_k$. Since $\frac{(\sum_k A_{i,k} D_k - h)}{ |\mathcal{N}_i|+1} < 0$, the left-hand side of the inequality is negative and the right hand side is positive and strictly smaller than one. 
Since $t_{i,h,1}$ is constrained to be either zero or one, in such case, it is set to be zero. Suppose now that $h \ge \sum_k A_{i,k} D_k$. Then the left-hand side is bounded from below by zero, and the right-hand side is bounded from below by one. Therefore $t_{i,h,1}$ is set to be one. 
Similarly, we can write 
\begin{equation}
 \frac{( h -\sum_k A_{i,k} D_k)}{  |\mathcal{N}_i| + 1} < t_{i,h,2} \le \frac{(h - \sum_k A_{i,k} D_k)}{  |\mathcal{N}_i| + 1} + 1, \quad  t_{i,h,2} \in \{0,1\}.
\end{equation} 
By definition, 
\begin{equation} 
t_{i,h,1} + t_{i,h,2} = \begin{cases} &1 \text{ if and only if } \sum_{k \in \mathcal{N}_i} D_ k \neq h \\
&2 \text{ otherwise }.
\end{cases} 
\end{equation} 
 Therefore, we can write
\begin{equation} 
\frac{1}{n} \sum_{i = 1}^n \sum_{h = 0}^{|\mathcal{N}_i|}(f_i(1, h) - f_i(0,h))D_i (t_{i,h,1} + t_{i,h,2} - 1)  +  (t_{i,h,1} + t_{i,h,2} - 1)f_i(0,h).
\end{equation} 
Finally, we introduce the variable $u_{i,h} = D_i (t_{i,h,1} + t_{i,h,2} - 1)$. Since $D_i, t_{i,h,1}, t_{i,h,2} \in \{0,1\}$ it is easy to show that such variable is completely determined by the above constraint. This completes the proof. 
\end{proof} 
Next, define 
$$
\begin{aligned} 
&\tilde{\sigma}_i^2(D_i, \sum_{k \in \mathcal{N}_i} D_k) = \sigma(T_i, D_i, \sum_{k \in \mathcal{N}_i} D_k) \\
&\tilde{\eta}_{i, j}(D_i, \sum_{k \in \mathcal{N}_i} D_k, D_j, \sum_{k \in N_j} D_k) = \eta(T_i, D_i, \sum_{k \in \mathcal{N}_i} D_k, T_j, D_j, \sum_{k \in N_j} D_k)
\end{aligned} 
$$ 
the variance function and $\tilde{\eta}_{i,j}(.)$ the covariance for unit $i$ and $j$, given their number of neighbors and the observed treatment assignments. 

We define 
\begin{equation} 
\begin{aligned} 
&v_i^1(D_i, \sum_{k \in \mathcal{N}_i} D_k) = 1\{D_i = d_1, \sum_{k \in \mathcal{N}_i} D_k = s_1, T_i = l\}, \\ &v_i^0(D_i, \sum_{k \in \mathcal{N}_i} D_k) = 1\{D_i = d_0, \sum_{k \in \mathcal{N}_i} D_k = s_0, T_i = l\}.
\end{aligned} 
\end{equation} 
The objective function reads as follows. 
\begin{equation} \label{eqn:j}
\tiny
\begin{aligned} 
&\sum_{i: R_i = 1} R_i \Big(\frac{v_i^1(D_i, \sum_{k \in \mathcal{N}_i} D_k) \tilde{\sigma}_i(D_i, \sum_{k \in \mathcal{N}_i} D_k)}{\sum_{i: R_i = 1} R_i v_i^1(D_i, \sum_{k \in \mathcal{N}_i} D_k)/n}\Big)^2 + R_i \Big(\frac{v_i^0(D_i, \sum_{k \in \mathcal{N}_i} D_k) \tilde{\sigma}_i(D_i, \sum_{k \in \mathcal{N}_i} D_k)}{\sum_{i: R_i = 1} R_i v_i^0(D_i, \sum_{k \in \mathcal{N}_i} D_k)/n}\Big)^2   \\ &+  \frac{R_i v_i^1(D_i, \sum_{k \in \mathcal{N}_i} D_k)}{\sum_{i: R_i = 1} R_i v_i^1(D_i, \sum_{k \in \mathcal{N}_i} D_k)/n} \times \\ &\times 
\sum_{j \in \mathcal{N}_i} R_j \Big(\frac{v_j^1(D_j, \sum_{k \in N_j} D_k)}{\sum_{i: R_i = 1} R_i v_i^1(D_i, \sum_{k \in \mathcal{N}_i} D_k)/n} - \frac{v_j^0(D_j, \sum_{k \in N_j} D_k) }{\sum_{i: R_i = 1} R_i v_i^0(D_i, \sum_{k \in \mathcal{N}_i} D_k)/n}\Big) \tilde{\eta}_{i,j}\Big(D_i, \sum_{k \in \mathcal{N}_i} D_k, D_j, \sum_{k \in N_j} D_k\Big) \\ 
&- \frac{R_i v_i^0(D_i, \sum_{k \in \mathcal{N}_i} D_k)}{\sum_{i: R_i = 1} R_i v_i^0(D_i, \sum_{k \in \mathcal{N}_i} D_k)/n} \times \\ &\times 
\sum_{j \in \mathcal{N}_i}R_j  \Big(\frac{v_j^1(D_j, \sum_{k \in N_j} D_k)}{\sum_{i: R_i = 1} R_i v_i^1(D_i, \sum_{k \in \mathcal{N}_i} D_k)/n} - \frac{v_j^0(D_j, \sum_{k \in N_j} D_k)}{\sum_{i: R_i = 1} R_i v_i^0(D_i, \sum_{k \in \mathcal{N}_i} D_k)/n}\Big)  \tilde{\eta}_{i,j}\Big(D_i, \sum_{k \in \mathcal{N}_i} D_k, D_j, \sum_{k \in N_j} D_k\Big).
\end{aligned} 
\end{equation} 

We now introduce the following auxiliary variables: $n \times \sum_{i: R_i = 1} |\mathcal{N}_i|$   variables $t_{i,h,1 } = 1\{\sum_{k \in \mathcal{N}_i} D_k \ge h\}$ and  $n \times \sum_{i: R_i = 1} |\mathcal{N}_i|$   variables $t_{i,h,2 } = 1\{\sum_{k \in \mathcal{N}_i} D_k \le h\}$. We define $\tilde{t}_{i,h} = t_{i,h,1} + t_{i,h,2} - 1$ and we define $u_{i,h} = D_i \times \tilde{t}_{i,h}$. Such variables are fully characterize by the two linear constraints for each variable as discussed in Lemma \ref{prop:opt} and the 0-1 constraint for each variable. By Lemma \ref{prop:opt}, each function or product of functions of the variables $(D_i, \sum_{k \in \mathcal{N}_i} D_k)$ can now be described as a linear function of these new decision variables.  
Consider for example, $(v_i^1(D_i, \sum_{k \in \mathcal{N}_i} D_k) \tilde{\sigma}_i(D_i, \sum_{k \in \mathcal{N}_i} D_k))^2$ first.
Then such function is rewritten as 
\begin{equation} 
\small 
\begin{aligned} 
& (v_i^1(D_i, \sum_{k \in \mathcal{N}_i} D_k) \tilde{\sigma}_i(D_i, \sum_{k \in \mathcal{N}_i} D_k))^2 = \\ &\sum_{h=1}^{|\mathcal{N}_i|} (v_i^1(1, h)^2  \tilde{\sigma}_i(1, h)^2- v_i^1(0, h)^2  \tilde{\sigma}_i(0, h)^2)u_{i,h} + v_i^1(0, h)^2  \tilde{\sigma}_i(0, h)^2 \tilde{t}_{i,h}.
\end{aligned} 
\end{equation} 
 
Similarly, consider the following function
\begin{equation} \label{lem:function_K}
K(D_i, D_j, \sum_{k \in \mathcal{N}_i} D_k, \sum_{k \in N_j} D_k):= v_i^1(D_i, \sum_{k \in \mathcal{N}_i} D_k)v_j^1(D_j, \sum_{k \in N_j} D_k) \tilde{\eta}_{i,j}\Big(D_i, \sum_{k \in \mathcal{N}_i} D_k, D_j, \sum_{k \in N_j} D_k\Big).
\end{equation} 
By Lemma \ref{prop:opt}, the left-hand side of Equation \eqref{lem:function_K} can be written as 
\begin{equation} \label{eqn:aaa}
\sum_{h = 0}^{|\mathcal{N}_i|} \Big(K(1, D_j, h, \sum_{k \in N_j} D_k) - K(0, D_j, h, \sum_{k \in N_j} D_k)\Big)u_{i,h} + \tilde{t}_{i,h}K(0, D_j, h, \sum_{k \in N_j} D_k). 
\end{equation} 
We can now linearize the function and obtain the following equivalent formulation 

\begin{equation} \label{eqn:aaa}
\begin{aligned} 
\sum_{h' = 0}^{|N_j|} \Big(\sum_{h = 0}^{|\mathcal{N}_i|} \Big(&K(1, 1, h, h') - K(0, 1, h, h)\Big)u_{i,h} + \tilde{t}_{i,h}K(0, 1, h, h') \\ &- \Big(K(1, 0, h, h') - K(0, 0, h, h)\Big)u_{i,h} + \tilde{t}_{i,h}K(0, 0, h, h') \Big)u_{j,h'} \\ 
&+ \Big(K(1, 0, h, h') - K(0, 0, h, h)\Big)u_{i,h} \tilde{t}_{j,h'} + \tilde{t}_{i,h}K(0, 0, h, h')\tilde{t}_{j,h'}\Big). 
\end{aligned} 
\end{equation} 

which is quadratic in the decision variables, as defined in Lemma \ref{prop:opt}. 
Therefore, each function in the numerators and denominators of Equation \eqref{eqn:j} can be written as a linear or quadratic function in the decision variables $D_i, u_{i,h}, \tilde{t}_{i,h}$. 
We now linearize the quadratic expressions in the numerator and denominators, to show that also quadratic expression admits a mixed-integer linear representation. To do so we introduce a new set of variables that we denote as 
\begin{equation} 
A_{i,j,h', h'} = u_{i,h} u_{j,h'}, \quad B_{i,j,h', h'} = u_{i,h} \tilde{t}_{j,h'}, \quad C_{i,h,h',h} = \tilde{t}_{i,h} \tilde{t}_{j,h'}.
\end{equation} 

Since each of the above variable takes values in $\{0,1\}$, such variables can be expressed with linear constraints. For instance, $A_{i,j,h', h'}$ is defined as follows. 

\begin{equation} 
\frac{u_{i,h} + u_{j,h'}}{2} - 1 < A_{i,j,h', h'}   \le \frac{u_{i,h} + u_{j,h'}}{2}, \quad  A_{i,j,h', h'} \in \{0,1\}. 
\end{equation} 
This is because if both $u_{i,h}, u_{j,h'}$ are both equal to one, the left hand size is zero, and under the 0-1 constraint, the resulting variable is equal to one. 
This follows similarly also for the other variables. 
Finally, notice that since also $R_i \in \{0,1\}$, the product of $R_i$ for any other 0-1 variable can be similarly linearized.  
Therefore, the above problem reads as a mixed-integer fractional linear program as described for instance in \cite{charnes1962programming}. 
\end{proof}

\end{document}